\documentclass[10pt]{article}
\usepackage{plain}
\usepackage{amsthm}
\usepackage{amsfonts}
\usepackage{amsmath}
\usepackage{amssymb}
\usepackage{physics}
\usepackage{mathtools}
\usepackage{comment}
\usepackage{url}
\usepackage{tikz}
\usepackage{eurosym}

\usetikzlibrary{arrows,decorations.pathmorphing,backgrounds,positioning,fit,petri,decorations}
\usetikzlibrary{calc,intersections,through,backgrounds,mindmap,patterns,fadings}
\usetikzlibrary{decorations.text}
\usetikzlibrary{decorations.fractals}
\usetikzlibrary{fadings}
\usetikzlibrary{shadings}
\usetikzlibrary{shadows}
\usetikzlibrary{shapes.geometric}
\usetikzlibrary{shapes.callouts}
\usetikzlibrary{shapes.misc}
\usetikzlibrary{spy}
\usetikzlibrary{topaths}

\newtheorem{prop}{Proposition}

\numberwithin{equation}{section}
\begin{document}

\begin{titlepage}
\title{\textbf{Cooking pasta with Lie groups}}
\maketitle
\begin{center}
\author{S.L.~Cacciatori,$^{1}$ F.~Canfora,$^2$ M.~Lagos,$^3$ F.~Muscolino$^1$ and A.~Vera$^{2,3}$}
\linebreak
\linebreak
\textsl{$^{1}$ \quad Department of Science and High Technology, Universit\`a dell'Insubria, Via Valleggio 11, IT-22100 Como, Italy, and INFN sezione di Milano, via Celoria 16, IT-20133 Milano, Italy;\\
$^{2}$ \quad Centro de Estudios Cient\'\i ficos (CECS), Casilla 1469, Valdivia, Chile\\
$^3$ \quad Instituto de Ciencias F\'\i sicas y Matem\'aticas, Universidad Austral de Chile, Valdivia, Chile
}
\abstract{We extend the (gauged) Skyrme model to the case in which the global isospin group (which usually is taken to be $SU(N)$) is a generic compact connected Lie group $G$. We analyze the corresponding field equations in (3+1)
dimensions from a group theory point of view. Several solutions can be
constructed analytically and are determined by the embeddings of three
dimensional simple Lie groups into $G$, in a generic irreducible
representation. These solutions represent the so-called nuclear pasta state
configurations of nuclear matter at low energy. We employ the Dynkin
explicit classification of all three dimensional Lie subgroups of
exceptional Lie group to classify all such solutions in the case $G$ is an
exceptional simple Lie group, and give all ingredients to construct them
explicitly. As an example, we construct the explicit solutions for $G=G_{2}$%
. We then extend our ansatz to include the minimal coupling of the Skyrme
field to a $U(1)$ gauge field. We extend the definition of the topological
charge to this case and then concentrate our attention to the
electromagnetic case. After imposing a \textquotedblleft free force
condition\textquotedblright\ on the gauge field, the complete set of coupled
field equations corresponding to the gauged Skyrme model minimally coupled
to an Abelian gauge field is reduced to just one linear ODE keeping alive
the topological charge. We discuss the cases in which such ODE belongs to
the (Whittaker-)Hill and Mathieu types.}
\end{center}
\end{titlepage}

%
%
%
%
%
%
%


\section{Introduction}

Nuclear pasta is a phase of matter that appears organized in some ordered
structures when a large number of Baryons is confined in a finite volume 
\cite{pasta1}, \cite{pasta2}, \cite{pasta3}, \cite{pasta4}, \cite{pasta5}, 
\cite{pasta6}. These configurations appear, for instance, in the crust of
neutron stars. Such aggregations of Baryons may take the form of tubular
structures, called Spaghetti states, or layers having a finite width, called
Lasagna states, or even globular shape, the gnocchi. Until very recently, it
was always tacitly assumed that nuclear pasta phase is the prototypical
situation in which it is impossible to reach a good analytic grasp. This is
related to the fact that such structures appear in the low energy limit of
Quantum Chromodynamics (QCD) in which perturbation theory does not work and,
at a first glance, the strong non-linear interactions prevent any attempt to
find exact solutions. Now, the low energy limit of QCD is described by the
Skyrme model \cite{skyrme} at the leading order in the 't Hooft expansion
(see \cite{Witten}, \cite{GSkyrme}, \cite{witten0}, \cite{bala0}, \cite%
{Bala1}, \cite{ANW}, as well as \cite{manton}, \cite{BaMa} and references
therein). Unsurprisingly, the highly non-linear character of the Skyrme
field equations discouraged any mathematical description of this kind of
structures. Consequently, as the above references show, numerical methods
(which, computationally, are quite demanding) are dominating in this regime.
The situation is even worse when one wants to analyze the electromagnetic
field generated in the nuclear pasta phase as, when the minimal coupling
with the $U(1)$ gauge field is taken into account; even the available
numerical methods are not effective.

\textit{On the other hand, one may ask: is the mathematical dream of an
analytic description of nuclear pasta structure really out of reach?}
Analytical methods to infer the general dependence of the nuclear pasta
phase on relevant physical parameters (such as the Baryon density) not only
would greatly improve our understanding of the nuclear pasta phase itself,
but they could also shed considerable light on the interactions of dense
nuclear matter with the electromagnetic field.

From the mathematical viewpoint, the problem is very deep and yet simple to
state: \textit{can we find analytic solutions of the (gauged) Skyrme model
able to describe typical configurations of the nuclear pasta phase?} Despite
the fact that this model has been introduced in the early sixties, for
several years only numerical solutions had been available (the only
exceptions being \cite{manton0.5}, in which the authors constructed analytic
solutions of the Skyrme field equations in a suitable fixed curved
background). Nevertheless, the mathematical beauty of the Skyrme model
attracted the attention of many leading mathematicians and physicists. In
particular, in \cite{manton0}, \cite{manton0.25}, \cite{manton0.3}, \cite%
{manton1} and \cite{spreight}, the authors were able to disclose the
geometrical structures of configurations with two Skyrmions, to analyze the
interaction energy of well separated solitons, to establish necessary
conditions for the existence of Skyrmionic crystals and so on. All these
remarkable results have been obtained without the availability of analytic
solutions of the Skyrme field equations. These efforts (together with the
comparison with Yang-Mills theory in which explicit solutions representing
instantons and non-Abelian monopoles shed considerable light on the
mathematical and physical properties of Yang-Mills theory itself) show very
clearly the importance to search for new analytic tools to analyze the
gauged Skyrme model in sectors with high Baryonic charge.

Quite recently, new methods have been introduced that allowed the
construction of explicit analytical solutions of the Skyrme field equations.
Such solutions are suitable to describe nuclear Lasagna and Spaghetti
states, see \cite{56b}, \cite{56}, \cite{56b1}, \cite{56c}, \cite{crystal1}, 
\cite{crystal2}, \cite{crystal3}, \cite{crystal3.1}, \cite{crystal4}, \cite%
{firstube}, \cite{firsttube2}, and \cite{gaugsksu(n)}. Let us recall that
the Skyrme model is a non-linear field theory for a scalar field $U$ taking
values in the $SU(N)$ Lie group, where $N$ is the flavor number. This theory
possesses a conserved topological charge (the third homotopy class) which
physically is interpreted as the Baryonic charge of the configuration.

Most of the solutions found so far have been constructed by employing ad hoc
ans\"atze adapted to the properties of the $SU(2)$ group, but soon it has
been realized that particular group structures seem to be at the root of the
solvability of the Skyrme field equations. For example, the exponentiation
of certain linear functions taking value in the Lie algebra lead to
Spaghetti-like configurations, while Euler parameterization of the field $U$%
, with suitable linear exponents, lead to Lasagna-like solutions. In all
these cases, the solutions are also topologically non-trivial with arbitrary
Baryonic charge. A proper mathematical understanding and generalization of
the strategy devised in \cite{56b}, \cite{56}, \cite{56b1}, \cite{56c}, \cite%
{crystal1}, \cite{crystal2}, \cite{crystal3}, \cite{crystal3.1}, \cite%
{crystal4}, \cite{firstube}, \cite{firsttube2} and \cite{gaugsksu(n)},
offers the unique opportunity to disclose the deep connections of the
nuclear pasta phase with the theory of Lie groups; two topics which (until
very recently) could have been considered extremely far from each other. 
\textit{The present paper is devoted to this opportunity:} to provide
nuclear pasta configurations of Lasagna and Spaghetti types with the
mathematical basis of Lie group theory.

A first step in this direction was to link certain properties of the
semi-simple Lie group to the possibility of getting explicit solutions of
the Skyrme equations in the Lasagna configurations for the case of $SU(N)$
groups with arbitrary $N$ \cite{gaugsksu(n)}. More in general, using the
methods developed in \cite{crystal1}, \cite{crystal2}, \cite{crystal3}, \cite%
{crystal4}, \cite{firstube}, \cite{gaugsksu(n)}, with the generalization of
the Euler angles to $SU(N)$ of \cite{euler1}, \cite{euler2}, \cite{euler3},
it has been possible to construct non-embedded multi-Baryonic solutions of
nuclear Spaghetti and nuclear Lasagna, at least for the case for the $SU(N)$
groups, see \cite{Cacciatori:2021neu}.

A fundamental ingredient in the theory of Lie groups with relevant
applications in the Skyrme model is the concept of\textit{\ non-embedded
solutions} introduced in \cite{bala0} and \cite{Bala1}. These are solutions
of the $SU(N)$-Skyrme model which cannot be written as trivial embeddings of 
$SU(2)$ in $SU(N)$. However, the techniques used to get such results, for
example in \cite{gaugsksu(n)}, where quite specific of the group $SU(N)$. In
fact, as we will show in the present manuscript, there is a very interesting
relation between Lie group theory and such families of solutions, which
allows to generalize the above results in a much more general setting and to
classify the solutions: this is exactly the main goal of the present paper.

\subsection{Resume of the results}

\textit{Firstly}, we will prove that, having fixed a compact connected Lie
group $G$ with a given irreducible representation (irrep), the solutions are
determined in general by deformations of embeddings of three dimensional Lie
groups into $G$.

\textit{Secondly}, we will prove that inequivalent families of solutions
correspond to inequivalent embeddings (not related by conjugation in $G$).
The problem of determine all possible three dimensional subgroups of a
simple Lie group has been solved by E.~B.~Dynkin in \cite{Dy-57}. In
particular, in that paper, all possible three dimensional subalgebras of the
exceptional Lie algebras are written down.

\textit{Thirdly}, we will show that such classification also classifies the
Spaghetti and Lasagna solutions determined via group theory methods. The
difference between Spaghetti and Lasagna depends on the realization of the
subgroup element of $G$: if it is generated by the exponentiation of a
linear combination of the generators of a three-dimensional subalgebra of $%
\mathfrak{g}=\mathcal{L}ie(G)$, then we get Spaghetti-like solutions, while
if the realization is through Euler parameterization we get Lasagna-like
solutions. Then, we will compute explicitly relevant quantities such as the
energy of these configurations.

\textit{Fourthly}, we will extend this classification to the case of the
gauged Skyrme model minimally coupled to Maxwell theory. In particular, we
will extend the definition of topological (Baryonic) charge to this case. We
will reduce the complete set of coupled field equations both in the gauged
Lasagna case and in the gauged Spaghetti case to a single linear equation
and we will analyze the integrable cases which correspond to Whittaker-Hill
and Mathieu types linear differential equations.

\subsection{Main tools employed in the analysis}

In the present work, we will employ abstract techniques and general
properties of semi-simple Lie groups in order to investigate their relation
with solvability of the Skyrme equations. This allows to extend all the
results found in \cite{gaugsksu(n)} for the special unitary groups to an
arbitrary semi-simple compact Lie group. Indeed, all results will be based
on the properties of the roots and weights of the associated Lie algebras,
while a generalized Euler parameterization of the Skyrme field $U$, taking
values in $G$, will lead in general to Lasagna configurations. Similarly,
the direct exponentiation of the algebra, as discussed above, will lead us
to Spaghetti structures extending the results of \cite{Cacciatori:2021neu}.
In any case, we will compute the energy of such configurations and will show
that they have always a non-trivial Baryon (topological) charge.
Interestingly enough, a strategy for constructing non-trivial non-$SU(2)$
solutions in the sense of \cite{bala0} and \cite{Bala1} will result to be
strictly related to the classification of all three dimensional groups in
any given simple Lie group, provided by Dynkin in his PhD thesis work, see 
\cite{Dy-57}. As an application of our general analysis, we will show how to
construct all non-trivial Lasagna and Spaghetti configurations in any
exceptional Lie group, making very explicit the case of $G=G_{2}$.

The generalization of our ans\"{a}tze which allows to include the minimal
coupling of the model to a $U(1)$ electromagnetic field will be introduced
as follows. As usual, the gauge field will work as a connection making all
derivative covariant under the action of the $U(1)$ gauge field, while their
dynamics is expressed by the usual Maxwell action (although our methods also
work in the Yang-Mills case). The covariant derivatives break the
topological nature of the original term expressing the Baryonic charge.
Therefore, generalizing the result in \cite{Witten}, we will deform the
Baryonic density expression in order to recover topological invariance .

The introduction of the electromagnetic field makes the field equations of
the gauged Skyrme model minimally coupled to Maxwell theory extremely more
complicated than in the Skyrme case. Nevertheless, quite surprisingly, the
equations will be separable (in a suitable sense) and once again solvable,
after imposing the \textit{free force conditions} on the gauge field. This
condition appears quite naturally in Plasma physics (see \cite{FFP1}, \cite%
{FFP2}, \cite{FFP3}, \cite{FFP4}, \cite{FFP5} and references therein). Quite
interestingly, such condition implies that the gauge field disappears from
the gauged Skyrme field equations (without being a trivial gauge field, of
course) and therefore, in this way the gauged Skyrme field equations can be
solved as in the ungauged case. It is a very non-trivial result that the
remaining field equations (which correspond to the Maxwell equations with
the source term arising from the gauged Skyrme model) reduce just to one
linear equation for a suitable component of the gauge field in which the
Skyrmion act as a source-like term. We will analyze the integrable cases in
which this last remaining equation takes the form of a Hill equation for the
case of Lasagna states, while a Schr\"{o}dinger equation with a bi-periodic
potential of finite type in the Spaghetti case.

Interestingly enough, for the Lasagna case another nice coincidence shows up
here: the relevant solutions we need are exactly the periodic solutions
whose existence has been investigated in \cite{Magnus}, and which explicit
form for the case of a Whittaker-Hill equation has been determined in \cite%
{Whittaker-Hill}.

It is a truly remarkable result that such a complicated phase such as the
nuclear pasta phase of the low energy limit of QCD (even taking into account
the minimal coupling with Maxwell theory) can be understood so cleanly in
terms of the theory of Lie group.

\subsection{Notations and conventions}

Our conventions are as follows. The action of the Skyrme model in $(3+1)$
dimensions is 
\begin{gather}  \label{I}
I= \int d^4x\sqrt{-g} \biggl[ \frac{K}{4}\text{Tr}\biggl( \mathcal{L}_\mu 
\mathcal{L}^\mu + \frac{\lambda}{8} G_{\mu\nu}G^{\mu\nu} \biggl) \biggl] , \\
\mathcal{L}_\mu= U^{-1}\partial_\mu U , \qquad G_{\mu\nu}=[\mathcal{L}_\mu,%
\mathcal{L}_\nu] , \qquad U(x) \in G ,  \notag
\end{gather}
where $K$ and $\lambda$ are positive coupling constants and $g$ is the
metric determinant. The Skyrme field $U$ is a map 
\begin{align*}
U:\mathbb{\mathbb{R}}^{1,3} \longrightarrow G
\end{align*}
where $G$ is semi-simple compact Lie group, so that 
\begin{equation*}
\mathcal{L}_\mu=\sum_{i=1}^{dim(G)}\mathcal{L}^i_\mu T_i ,
\end{equation*}
where $\{T_i\}$ is a basis for the Lie algebra $\mathfrak{g}=Lie(G)$.

The system is confined in a box of finite volume with a flat metric. For
Lasagna states we will use a metric of the form 
\begin{equation}
ds^{2}=-dt^{2}+L_{r}^{2}dr^{2}+L_{\gamma }^{2}d\gamma ^{2}+L_{\phi}^{2}d\phi
^{2} ,  \label{box}
\end{equation}
where the adimensional spatial coordinates have the ranges 
\begin{equation}
0\leq r\leq 2\pi ,\quad 0\leq \gamma \leq 2\pi ,\quad 0\leq \phi \leq 2\pi ,
\label{ranges}
\end{equation}
so that the solitons are confined in a box of volume $V=(2\pi)^{3}L_{r}L_{%
\gamma }L_{\phi }$.\newline
For nuclear Spaghetti we will use the metric ansatz 
\begin{align}
ds^2 =-dt^2 + L_r^2dr^2 + L_\theta^2 d\theta^2 + L_\phi^2d\phi^2,
\end{align}
with adimensional coordinates ranging in 
\begin{align}
0\leq r\leq 2\pi ,\quad 0\leq \theta \leq \pi ,\quad 0\leq \phi \leq 2\pi ,
\label{ranges-spag}
\end{align}
and a total volume $V=4\pi^3 L_r L_\theta L_\phi$.

The energy-momentum tensor associated to the Skyrme field is given by 
\begin{equation}
T_{\mu \nu } \ = \ -\frac{K}{2}\text{Tr}\biggl(\mathcal{L}_{\mu }\mathcal{L}%
_{\nu }-\frac{1}{2}g_{\mu \nu }\mathcal{L}_{\alpha }\mathcal{L}^{\alpha }+%
\frac{\lambda }{4}(g^{\alpha \beta }G_{\mu \alpha }G_{\nu \beta }-\frac{1}{4}%
g_{\mu \nu }G_{\alpha \beta }G^{\alpha \beta })\biggl) \ .  \label{Tmunu}
\end{equation}
The topological charge is defined by (see Proposition \ref{prop3}) 
\begin{gather}  \label{B}
B=\frac{1}{24\pi^{2}}\int_{\mathcal{V}} \mathrm{Tr} (\mathcal{L }\wedge 
\mathcal{L }\wedge \mathcal{L}),
\end{gather}
where $\mathcal{V}$ is the spatial region spanned by the coordinates at any
fixed time $t$, $\mathcal{L}=U^{-1} d U$ and Tr is the trace over the matrix
indices.


\section{Lasagna groups}

In \cite{AlCaCaCe} it has been shown how Lasagna configurations can be
determined as solutions of the Skyrme equations realized as Euler
parameterizations of three dimensional cycles in $SU(N)$. Indeed, these
cycles result to be suitable deformations of different non-trivial
embeddings of $SU(2)$ into $SU(N)$. Here we want to prove that such
construction can be easily extended to any simple Lie group (at least for
the case of the undeformed embedding). Recall that in the case of $SU(N)$
the embedding was defined \cite{gaugsksu(n)} by the generalized Euler map 
\begin{align}
U(t,r,\gamma,\phi)=e^{\left(\frac t{L_\phi}-\phi\right)\sigma\kappa}
e^{h(r)} e^{m\gamma\kappa},  \label{Ansatz-FP}
\end{align}
where $\sigma$ is a constant, $m$ is an integer, $\kappa$ a suitable matrix
in $SU(N)$ and $h(r)$ results to be a linear function of $r$ with values in
the Cartan algebra $H$. Indeed, the main trick was to determine a suitable
matrix $\kappa$ able to make everything easily computable and to grant
periodicity of $e^{m\gamma\kappa}$. The convenient strategy has been
composed in two steps: first we have taken a basis of eigenmatrices of the
simple roots, $\lambda_j$, $j=1,\ldots,r$, where $r=N-1$ is the rank of the
group, and defined the matrix 
\begin{align}  \label{kappa}
\kappa=\sum_{j=1}^r (c_j \lambda_j-c^*_j \lambda_j^\dagger),
\end{align}
where $\dagger$ means hermitian conjugate and $c_j$ are complex constants.
The second step consisted in determining the allowed values for the $c_j$.
We want to do the same with a generic simple Lie group $G$ replacing $SU(N)$%
. The first problem we ran into is the following. If $\lambda\in \mathfrak{g}%
_{\mathbb{C}}$ is an eigenmatrix of a root $\alpha$ of the Lie algebra $%
\mathfrak{g}$ of $G$ (so it belongs to the complexification $\mathfrak{g}_{%
\mathcal{C}}$ of $\mathfrak{g}$), in general $\lambda^\dagger$ doesn't
belong to $\mathfrak{g}_{\mathbb{C}}$ if $G\neq SU(N)$ for some $N$. So in
general $\kappa$ defined above is not a matrix of $\mathfrak{g}$.\newline
In order to overcome this problem, we notice that a compact simple Lie group 
$G$ always contain a \textit{split} maximal subgroup \cite{CaDaPiSc}, which
is a maximal subgroup $K$ with the property that $2\mathrm{dim}(K)+r=\mathrm{%
dim}(G)$ and that there exists a Cartan subalgebra $H$ of $\mathfrak{g}$ all
contained in $\mathfrak{p}$, the orthogonal complement of the Lie algebra $%
\mathfrak{k}$ of $K$ in $\mathfrak{g}$ (w.r.t. the Killing product): 
\begin{align}
\mathfrak{g}=\mathfrak{k}\oplus \mathfrak{p}.
\end{align}
Of course $\mathfrak{k}$ is a subalgebra of $\mathfrak{g}$, while $\mathfrak{%
p}$ is not, since 
\begin{align}
[\mathfrak{k},\mathfrak{k}]\subset \mathfrak{k}, \qquad [\mathfrak{k},%
\mathfrak{p}]\subset \mathfrak{p}, \qquad [\mathfrak{p},\mathfrak{p}]\subset 
\mathfrak{k},
\end{align}
which says that $\mathfrak{p}$ is a representation space for $G$ and $K$ is
an isotropy group for $\mathfrak{p}$. One can easily show that a root matrix 
$\lambda$, associated to a root $\alpha$, must have the form 
\begin{align}
\lambda=k+i p, \quad k\in \mathfrak{k},\ p\in \mathfrak{p},\ p\neq 0.
\end{align}
Then, $k-ip$ also is a root matrix, corresponding to the root $-\alpha$. We
replace the hermitian conjugation with the $\sim$ conjugation defined by 
\begin{align}
\widetilde{k+i p}\equiv (k+i p)^\sim:=k-i p.
\end{align}
This way, if $\lambda_j$, $j=1,\ldots,r$ are matrix roots corresponding to
the simple roots of $\mathfrak{g}$, then 
\begin{align}
\kappa=\sum_{j=1}^r (c_j \lambda_j+c^*_j \tilde \lambda_j) \in \mathfrak{g}.
\label{Ansatz-FP1}
\end{align}
Notice that for $G=SU(N)$ we have $\tilde \lambda=-\lambda^\dagger$.\newline
If we choose normalizations as in Appendix \ref{app:roots}, we can use the
matrices $J_k$ to decompose $h(z)=\sum_{j=1}^r y_j(z) J_j$. The properties
of the roots can be inferred case by case from the lists in Appendix \ref%
{app:roots}. Exactly the same calculations as in \cite{AlCaCaCe} show that
the field equations for the Skyrme field are equivalent to the system 
\begin{align}
h^{\prime \prime }+\frac{\lambda m^{2}}{2L_{\gamma }^{2}} \sum_{j=1}^{N-1}%
\alpha _{j}(h^{\prime \prime })|c_{j}|^{2}J_{j}&=0,  \label{ridotta} \\
\sum_{j<k}\left( \alpha _{j}(h^{\prime})^2-\alpha _{k}(h^{\prime})^2-i\left(
\alpha _{j}(h^{\prime \prime })-\alpha _{k}(h^{\prime \prime })\right)
\right) c_{j}c_{k} (\alpha_j|\alpha_k) \lambda_{\alpha_j+\alpha_k} &=0.
\label{vincolo}
\end{align}
The first equation has solution $h^{\prime \prime }=0$, as a consequence of
the strict positivity of the Cartan matrix for each simple group. The second
system, using that the $\lambda_{\alpha_j+\alpha_k}$ are independent,
reduces to the set of equations 
\begin{align}
\alpha _{j}(h^{\prime})^2-\alpha _{k}(h^{\prime})^2=0, \qquad j<k, \quad
s.t. \ (\alpha_j|\alpha_k)\neq 0.
\end{align}
Since $(\alpha_j|\alpha_k)\neq 0$ if and only if $\alpha_j$ and $\alpha_k$
are linked and since there are $r-1$ links in a connected Dynkin diagram,
these are exactly $r-1$ equations. These are independent and assuming $%
a=\alpha_1(h)\neq 0$ have the general solution 
\begin{align}
\alpha_j(h^{\prime })=\epsilon_j a, \quad j=2,\ldots, r,
\end{align}
where $\epsilon_j$ are signs. As in \cite{AlCaCaCe}, we can solve it by
writing 
\begin{align}
h^{\prime }=a \sum_{j=1}^r 2 \frac {w_j}{(\alpha_j|\alpha_j)} J_j.
\end{align}
Applying $\alpha_k$ to both hands and defining $\epsilon_1=1$ we get 
\begin{align}
\epsilon_k=\sum_{j=1}^r 2 \frac {w_j}{(\alpha_j|\alpha_j)}
\alpha_k(J_j)=\sum_{j=1}^r w_j C^G_{jk},
\end{align}
where $C^G$ is the Cartan matrix associated to $G$. The Cartan matrix is
positive definite and is therefore always invertible, so that 
\begin{align}
w_k=\sum_{j=1}^r \epsilon_j (C^G)^{-1}_{j,k}.
\end{align}
Therefore, we have proven the following generalization of Proposition 2 in 
\cite{AlCaCaCe}.

\begin{prop}
\label{propuno} All local solutions of the Skyrme field equations of the
form determined by the ansatz (\ref{Ansatz-FP}), (\ref{Ansatz-FP1}), with
metric 
\begin{align}
ds^2=-dt^2+L_r^2dr^2+L_\gamma^2 d\gamma^2+L_\phi^2 d\phi^2,
\end{align}
are given by 
\begin{align}
h(r)&=ar v_{\epsilon}, \\
v_{\epsilon}&=\sum_{j,k} (C^G)^{-1}_{j,k} \epsilon_j \frac
2{\|\alpha_k\|^2}J_k,  \label{vepsilon}
\end{align}
where $a$ is a real constant and $\epsilon_j$ are signs, with $\epsilon_1=1$.
\end{prop}

Now we have to discuss which choices of the coefficients $c_j$ are allowed.
To this hand, we have that the solution must cover a topological cycle
entirely. First, we notice that as a consequence of our normalizations, if
we want to get it with $r$ varying in $[0,2\pi]$, we must take 
\begin{align}
a=\frac 12,
\end{align}
see \cite{AlCaCaCe}, Proposition 3. \newline
The second step is to grant periodicity of $e^{\gamma \kappa}$. This is the
difficult part and determines the allowed values for the $c_j$. Notice that $%
\kappa$ is diagonalizable (over $\mathbb{C}$). Indeed, since $G$ is compact,
in the adjoint representation $\kappa$ results to be antihermitian and then
diagonalizable with imaginary eigenvalues. It follows that it is
diagonalizable in any representation with purely imaginary eigenvalues. If $%
N $ is the dimension of the representation, then the eigenvalues $%
i\mu_1,\ldots, i\mu_N$ must be in rational ratios, which means that for any $%
\mu_a\neq 0$ it must exist integers $n_a\neq 0$ such that 
\begin{align}
\mu_a n_b=\mu_b n_a,
\end{align}
or, equivalently, that it exists a non-vanishing real number $\mu$ and $N$
integers $n_a\in \mathbb{Z}$, such that 
\begin{align}
\mu_a=\mu n_a.
\end{align}
This condition in general will depend on $N$, $G$ and the constants $c_j$.
In \cite{AlCaCaCe} this problem has been shown to have a set of solutions
for the particular case of $G=SU(N)$ in the fundamental representation. Here
we have to generalize that procedure without exploiting a very explicit
realization. Indeed, we can prove that there are solutions with all $c_j$
different from zero by following a strategy developed by Dynkin in \cite%
{Dy-57}, that we will recall in the next section. Let us choose $f$ in the
Cartan subalgebra, such that $\alpha_j(f)=b$, a positive constant
independent on $j$, so that 
\begin{align}  \label{f}
[f,\lambda_j]=ib \lambda_j, \qquad [f,\tilde\lambda_j]=-i b \tilde\lambda_j.
\end{align}
We can easily determine it as follows. If $h_j=i[\lambda_j, \tilde\lambda_j]$%
, then set $f=\sum_{k=1}^r p_k h_k$. Thus, the above condition is equivalent
to 
\begin{align}
b=\sum_{k=1}^r p_k\alpha_j(h_k)=\sum_{k=1}^r p_k (\alpha_j|\alpha_k) =
\sum_{k=1}^r p_k \frac {\|\alpha_k\|^2}2 C^G_{kj}.
\end{align}
from which we immediately get 
\begin{align}  \label{Coeff_f}
p_j= b\frac 2{\|\alpha_j\|^2} \sum_{k=1}^r (C^G)^{-1}_{kj}.
\end{align}
By the properties of the Cartan matrix it follows that $p_j$ are all
positive. Finally, we set 
\begin{align}  \label{CondCF}
c_j =e^{i\psi_j} \sqrt {\frac b2 p_j}.
\end{align}
Then, we have the following proposition:

\begin{prop}
\label{prop-dynkin} If $\kappa$ is constructed with the above choice of $c_j$%
, then $e^{\kappa z}$ is periodic with period $n\frac {2\pi}b$ where $n$ may
be 1 or 2 depending on the representation, $n=1$ for the adjoint
representation.
\end{prop}

\begin{proof}
We first show that periodicity is independent on the phases of $c_j$. If $%
e^{\kappa z}$ is periodic, then, for any fixed $g\in G$, $ge^{\kappa z}g^{-1}
$ is also periodic with the same period. Since the simple roots $\alpha_j$
are linearly independent,  for any fixed $j$ we can find an element $h_j$ of
the Cartan algebra such that $\alpha_k(h_j)=\delta_{kj}$. Let us set $%
g=e^{\psi h_j}$. Then,  
\begin{align}
g\kappa g^{-1}&= \sum_{k=1}^r (c_k e^{\psi h_j} \lambda_k e^{-\psi
h_j}+c^*_k e^{\psi h_j}\tilde \lambda_k e^{-\psi h_j}) =\sum_{k=1}^r (c_k
e^{i\psi \alpha_k(h_j)} \lambda_k +c^*_k e^{-i\psi \alpha_k(h_j)}\tilde
\lambda_k ) \cr & =\sum_{k\neq j} (c_k \lambda_k +c^*_j \tilde \lambda_k +c_j
e^{i\psi} \lambda_j +c^*_j e^{-i\psi } \tilde \lambda_j)  ,
\end{align}
which shows that $g\kappa g^{-1}$ differs from $\kappa$ only by the phase of 
$c_j$. This proves our assert.  So, it is sufficient to prove the proposition
for $\psi_j=0$. In this case, $T_3:=f, T_1:=\kappa$ and $T_2=\frac 1b
[f,\kappa]$ form an $A_1$ subalgebra of $\mathfrak{g}$, and $\kappa$ is
conjugate to $f$ in $G$. Indeed,  $[T_i,T_j]=b\epsilon_{ijk}T_k$ from which  
\begin{align}
T_1=e^{\frac \pi2 T_2} T_3 e^{-\frac \pi2 T_2}.
\end{align}
Therefore, as before, the periodicity of $e^{\kappa z}$ is equivalent to the
periodicity of $e^{f z}$. But  
\begin{align}
e^{f z} h_j e^{-f z}=h_j,
\end{align}
and  
\begin{align}
e^{f z} \lambda_\alpha e^{-f z}=e^{i\alpha(f) z} \lambda_\alpha.
\end{align}
Now, any given root $\alpha$ is 
\begin{align}
\alpha=\sum_j n_j \alpha_j,
\end{align}
with the $n_j$ all non-negative or all non-positive integers. Therefore, 
\begin{align}
e^{i\alpha(f) z} =e^{i\sum_j n_j b z}.
\end{align}
All this exponentials are therefore periodic, with the longest period
determined by the simple roots, for which $e^{i\alpha(f) z}=e^{i b z}$,
which has period $T=2\pi/b$. But 
\begin{align}
e^{f T} \lambda_\alpha e^{-f T}=\lambda_\alpha
\end{align}
for any root $\alpha$ implies that $g=e^{f T}$ is in the center of the
group. Since the center of a simple compact group is finite, this means that 
$g^n=I$ is the unit matrix for some integer $n$. Now, since $\kappa$ is not
in the Cartan subalgebra, it follows from \cite{He}, Section VII, Theorem
8.5 (see also \cite{CaDaPiSc}) that $n=1$ or $n=2$ depending on the specific
representation. \flushright{$\Box$}
\end{proof}

This shows that there exist always at least a set of solutions with all non-
vanishing $c_{j}$, parametrized by a torus of phases. In \cite{AlCaCaCe} it
has been shown that indeed, for the case of $SU(N)$ in the smallest
irreducible representation, there is a further set of deformations that has
been called a moduli space. This is a very difficult task to be investigated
in general and we will not consider it here.

\subsection{On the physical meaning of the time-dependence in the ansatz}

It is worth to discuss the physical meaning of the time-dependent ansatz in
Eq. (\ref{Ansatz-FP}) for the Lasagna-type configurations as well as the one in
Eqs. (\ref{exponential}), (\ref{SA1}) and (\ref{SA2}) for the spaghetti-type configurations. First of all,
despite the time-dependence of the ansatz of the $U$ field, the
energy-momentum tensor is still stationary (so that it describes a static
distribution of energy and momentum). This approach is inspired by the usual
time-dependent ansatz that is used for Bosons stars \cite{Kaup,Liebling} (and generalize it to
arbitrary Lie group) in which the $U(1)$ charged scalar field depends on
time in such a way to avoid the Derrick theorem (see \cite{Derrick1964}).
Secondly, the peculiar time-dependence is chosen in order to simplify as
much as possible the field equations without loosing the topological charge
(as, until very recently, the Skyrme field equations have always been
considered a very hard nut to crack from the analytic viewpoint). Thirdly
(as it will be discussed in the next sections on the minimal coupling with
Maxwell), the present ansatz (both for lasagna and spaghetti type
configurations) produces $U(1)$ currents associated to the minimal coupling
with Maxwell with a manifest superconducting current. Indeed (as it is clear
from Eqs. (\ref{EulerMaxwellDecoupled}) and (\ref{SpaghettiMaxwellDecoupled})), the present $U(1)$ current always has the form 
\begin{equation}
J_{\mu }\ =\ \Omega (\partial _{\mu }\Phi -2A_{\mu })\ ,  \label{supercurr1}
\end{equation}%
where $\Omega $ depends on either the Lasagna or the spaghetti profiles (see
Eqs. (\ref{EulerMaxwellDecoupled}) and (\ref{SpaghettiMaxwellDecoupled})) while $\Phi $ is a field which is defined modulo $2\pi $.
Consequently, the following observations are important.

\textbf{1)} The current does not vanish even when the electromagnetic
potential vanishes ($A_{\mu}=0$).

\textbf{2)} Such a \textquotedblleft left over\textquotedblright\ 
\begin{equation}
J_{\mu }^{(0)}=\left. J_{\mu }\right\vert _{A_{\mu }=0}=\Omega \partial
_{\mu }\Phi \ ,  \label{supercurr111}
\end{equation}%
is maximal where $\Omega $ is maximal (and this corresponds to the local
maxima of the energy density: see Eqs. (\ref{EulerMaxwellDecoupled}) and (\ref{FFEnergyDensity}).

\textbf{3)} $J_{(0)\mu }$ \textit{cannot be turned off continuously}. One
can try to eliminate $J_{(0)\mu }$ either deforming the profiles appearing
in $\Omega $ integer multiples of $\pi $ (but this is impossible as such a
deformation would kill the topological charge as well) or deforming $\Phi $
to a constant (but also this deformation cannot be achieved for the same
reason). Moreover, as it is the case in \cite{wittenstrings}, $\Phi $ is
only defined modulo $2\pi $. Consequently, $J_{(0)\mu }$ defined in Eq. (\ref%
{supercurr111}) is \textit{a superconducting current supported by the
present gauged configurations}.

These are the three of the main physical reasons to choose this peculiar time-dependent ansatz. On the other hand, it is worth to emphasize that the peculiar time-dependence we have chosen (for the reasons explained above) prevents one from using the usual techniques (see, for instance, \cite{ANW}) to "quantize" the present topologically non-trivial solutions. In particular, the typical hypothesis of a static $SU(N)$-valued field $U$ is violated in our case (since, as it has been already emphasize, the requirement to have a static $T_{\mu\nu}$ which describes a stationary distribution of energy and momentum does not imply that U itself is static). Therefore, to estimate the "classical isospin" of the present configurations we will proceed in a different manner in the next sections.


\subsection{Energy and Baryon number}

The energy of these solutions can be easily computed by means of Proposition %
\ref{prop6} in Appendix B. We get 
\begin{align}
E&= L_rL_\gamma L_\phi \|\underline c\|^2 \frac K2\pi^3 \left[ 16 \frac {%
\sigma^2}{L^2_\phi} + \frac {\| v_{\epsilon}\|^2}{\|\underline c\|^2L_r^2}+%
\frac {\sigma^2 \lambda }{L^2_\phi L^2_r} \right. \cr & \left. + 4\frac {m^2 
}{L_\gamma^2} \left( 2 + \frac {\lambda }{8 L_r^2} + \frac {\lambda \sigma^2%
}{L^2_{\phi} \|\underline c\|^2} \left(\sum_{j=1}^{r}
\|\alpha_j\|^2|c_j|^4+\sum_{j<k} |c_j|^2 |c_k|^2 (\alpha_j|\alpha_k)
\left(2\epsilon_j\epsilon_k +(\alpha_j|\alpha_k)(1-\epsilon_j
\epsilon_{k})\right) \right) \right) \right],
\end{align}
with 
\begin{align}
\|v_\epsilon\|^2&=-\mathrm{Tr} v_\epsilon^2, \\
\|\underline c\|^2&=\sum_{j=1}^r |c_j|^2,
\end{align}
and where $\sigma$ depends on the representation and has to be chosen so
that the solution correctly covers a cycle when $m=1$ and $\phi$ varies from 
$0$ to $2\pi$. To specify it, let us investigate the Baryon number integral.
To this hand, let us look better at Proposition \ref{prop-dynkin}. The fact
that $n=1$ or $2$ obviously distinguishes the $SO(3)$-type solutions from
the $SU(2)$-type ones (see \cite{AlCaCaCe}), since only in the first case
the period remains invariant when passing to the adjoint representation. The
right ranges are then understood by considering the correct Euler
parameterizations for $SO(3)$ and for $SU(2)$. If we write it generically as 
\begin{align}
U(x,y,z)=e^{xT_3} e^{yT_1} e^{z T_3},
\end{align}
one finds that, if $T$ is the period of the exponential functions, in both
cases $z$ must vary in a period and $y$ in a range of $T/4$. The difference
is in $x$, which has to vary in a period for $SO(3)$ and half a period for $%
SU(2)$, for example, see Appendix C in \cite{AlCaCaCe}. If we set $x=\sigma
\phi$, $y=r$ and $z=m\gamma$ and we want to normalize the ranges of the
coordinates $\phi, r, \gamma$, so that all vary in $[0,2\pi]$, we see that
we always have to require 
\begin{align}
b=n, \qquad \sigma=\frac n2,
\end{align}
where $n$ is an integer. With these conventions we can state the following
proposition.

\begin{prop}
\label{prop3} The Baryonic topological charge is 
\begin{align}
B=\frac 1{24\pi^2} \int \varepsilon^{ijk} \mathrm{Tr} (\mathcal{L}_i\mathcal{%
L}_j\mathcal{L}_k ) \sqrt g\ dr d\gamma d\phi =mn \|\underline c\|^2,
\end{align}
where $\mathcal{L}_i=U^{-1} \partial_i U$.
\end{prop}

The proof is exactly the same as in Appendix F of \cite{AlCaCaCe}, so we
omit it. \newline
The energy per Baryon $g=E/B$ is therefore 
\begin{align}
g&= L_rL_\gamma L_\phi \frac K{2mn}\pi^3 \left[ 16 \frac {\sigma^2}{L^2_\phi}
+ \frac {\| v_{\epsilon}\|^2}{\|\underline c\|^2L_r^2}+\frac {\sigma^2
\lambda }{L^2_\phi L^2_r} \right. \cr & \left. + 4\frac {m^2 }{L_\gamma^2}
\left( 2 + \frac {\lambda }{8 L_r^2} + \frac {\lambda \sigma^2}{L^2_{\phi}
\|\underline c\|^2} \left(\sum_{j=1}^{r} \|\alpha_j\|^2|c_j|^4+\sum_{j<k}
|c_j|^2 |c_k|^2 (\alpha_j|\alpha_k) \left(2\epsilon_j\epsilon_k
+(\alpha_j|\alpha_k)(1-\epsilon_j \epsilon_{k})\right) \right) \right) %
\right].
\end{align}


\section{Spaghetti groups}

\label{SpaghettiGroup} Another kind of configurations is obtained by
starting from a different ansatz, which leads to Spaghetti like solutions.
Spaghetti can be parameterized by the following ansatz: 
\begin{align}
U(t,r,\theta,\phi)&= \exp (\chi (r) \tau_1),  \label{exponential}
\end{align}
where $\tau_1 = \vec n \cdot \vec T=n_1 T_1+n_2 T_2+n_3 T_3$ is defined by 
\begin{align}\label{SA1}
\vec T &=(T_1,T_2,T_3), \qquad \vec n =(\sin \Theta \cos \Phi, \sin \Theta
\sin \Phi, \cos \Theta) \\\label{SA2}
\Theta& =q\theta ,\quad \Phi =p\left( \frac{t}{L_\phi}-\phi \right) , \quad
q=2v+1 ,\quad p, v\in \mathbb{N} , \quad p\neq 0\ .
\end{align}
In the ansatz, $T_i$ are matrices of a given representation of the Lie
algebra of $G$ and are required to define a three dimensional subalgebra
that we can choose to normalize so that 
\begin{align}
[T_j,T_k]=\varepsilon_{jkm} T_m,
\end{align}
and satisfy 
\begin{align}
\mathrm{Tr} (T_jT_k)=-2 I_{G,\rho} \delta_{jk},
\end{align}
where $I_{G,\rho}$ is the Dynkin index of $su(2)$ in $G$ (see \cite{Dy-57}),
that is the coefficient relating the trace product in the representation $%
\rho$ of $Lie(G)$ to the Killing product of $su(2)$. We also define 
\begin{align}
\tau_2&= \partial_\Theta \tau_1, \\
\tau_3&= \frac 1{\sin\Theta}\partial_\Phi \tau_1.
\end{align}
Together with $\tau_1$, they satisfy 
\begin{align}
[\tau_j,\tau_k]=\epsilon_{jkm} \tau_m.
\end{align}
With these rules, we get for $\mathcal{L}_\mu =U^{-1} \partial_\mu U$: 
\begin{align}
\mathcal{L}_r =& \tau_1\ \chi^{\prime }(r).  \label{MLr}
\end{align}
For the other terms, set $\alpha=\Theta, \Phi$ and using 
\begin{align}
U^{-1} \partial_\alpha U&=\chi \int_0^1 e^{-\sigma \chi \tau_1}
\partial_\alpha \tau_1 e^{\sigma \chi \tau_1}, \\
e^{-\sigma \chi \tau_1} \tau_2 e^{\sigma \chi \tau_1}&=\cos (\sigma \chi)
\tau_2 -\sin (\sigma \chi) \tau_3, \\
e^{-\sigma \chi \tau_1} \tau_3 e^{\sigma \chi \tau_1}&=\sin (\sigma \chi)
\tau_2 +\cos (\sigma \chi) \tau_3,
\end{align}
we get 
\begin{align}
\mathcal{L}_\Theta=& \sin \chi \ \tau_2 - (1-\cos \chi) \ \tau_3, \\
\mathcal{L}_\Phi=& \sin \Theta (\sin \chi \ \tau_3 + (1-\cos \chi) \ \tau_2),
\end{align}
and 
\begin{align}
\mathcal{L}_t&=\frac p{L} \mathcal{L}_\Phi, \\
\mathcal{L}_\theta&=q \mathcal{L}_\Theta,  \label{MLthet} \\
\mathcal{L}_\phi&=-p \mathcal{L}_\Phi.
\end{align}
This shows that the expression of the $\mathcal{L}_\mu$ is universal (depend
only on the algebra of the $\tau_j$), so the field equations are always the
same for any choice of the group. These are 
\begin{align}
4\chi^{\prime \prime }(r) \left(\lambda q^2 \sin^2 \left(\frac \chi2\right)+
L_\theta^2\right)- q^2 \sin \chi \left(4L_r^2 - \lambda
\chi^{\prime2}\right)= 0 \ .  \label{SkyrmeEqsSpaghetti}
\end{align}
What is expected to change is just the topological charge and the energy.
Given this universality property, we see immediately that, for any given
group $G$, these kind of solutions are classified by all possible ways of
finding a three dimensional simple subalgebra of the lie algebra $\mathfrak{g%
}$. Luckily, we don't need to tackle such a program, since has already been
solved by E. B. Dynkin in \cite{Dy-57}, chapter III. This work as follows.

\ 

First, it is convenient to complexify the algebra, recombine and normalize
the generators $f, e_+, e_-$ of the subgroup so that 
\begin{align}
[e_+,e_-]=-if, \qquad [f,e_\pm]=\pm2ie_\pm.
\end{align}
Each complex three dimensional simple algebra is isomorphic to this.
However, we must consider as equivalent only the ones which are isomorphic
through an automorphism of the group. Let $\alpha_j$, $j=1,\ldots,r$ be
simple roots defined from a cartan subalgebra containing $f$. Then, it
results that $(\alpha_j|f)$ must be integer numbers that can assume only the
values 0,1,2. The set of numbers $d_j=\alpha_j(f)$ are called the Dynkin
characteristic of the subgroup. The main result of \cite{Dy-57} is that the
three dimensional simple subalgebras $A$ are in one to one correspondence
with the characteristics and one can indeed classify the characteristics. A
subalgebra is said to be regular if its roots are indeed roots of $\mathfrak{%
g}$. The subalgebra $A$ is said to be integral if the projection of the
roots of $\mathfrak{g}$ along the direction of the roots of $A$ are integer
multiples of the simple root $\alpha_A$ of $A$. Since $\alpha_A(f)=2$, we
see that the dual of $\alpha_A$ in the Cartan subalgebra $H$ is 
\begin{align}
h_{\alpha_A}=\frac 2{(f|f)} f.
\end{align}
From this it follows immediately that $A$ is integral if and only if all the
numbers of the Dynkin characteristic $\chi_A=(d_1,\ldots,d_r)$ of $A$ are
even (so are 0 and 2). All inequivalent characteristics for the exceptional
Lie groups are listed in \cite{Dy-57}. Furthermore, given such a
characteristic $\chi=(d_1,\ldots,d_r)$, there it is explained how to
construct explicitly the associated subalgebra. First, if $J_j$ is the dual
of $\alpha_j$ in $H$, write 
\begin{align}
f=\sum_{j=1}^r p_j J_j,
\end{align}
and choose $p_j$ so that $\alpha_j(f)=d_j$. This gives 
\begin{align}
d_j=\sum_{k=1}^r p_k (\alpha_k|\alpha_j)=\sum_{k=1}^r \frac {\|\alpha_k\|^2}%
2 p_k C^G_{kj}.  \label{323}
\end{align}
From this we get 
\begin{align}
p_k=\sum_{j=1}^r d_j (C^G)^{-1}_{jk} \frac {2}{\|\alpha_k\|^2}.
\end{align}
As usual, $C^G$ is the Cartan matrix. In general, the construction of the
remaining generators is non-trivial. To do it, one has to consider the
subset of the root system $\Sigma$ defined by 
\begin{align}
\Sigma_{\chi_G}=\{\alpha\in \Sigma| \alpha(f)=2\}.
\end{align}
Then, all roots are positive. If $\lambda_\alpha$ are the corresponding
eigenmatrices (normalized so that $\mathrm{Tr} (\tilde \lambda_\alpha
\lambda_\alpha)=-1$), one then has to look for real coefficients $k_\alpha$
such that, setting 
\begin{align}
e_+&=\sum_{\alpha\in \Sigma_{\chi_G}} k_\alpha \lambda_\alpha, \\
e_-&=\sum_{\alpha\in \Sigma_{\chi_G}} k_\alpha \tilde \lambda_\alpha,
\end{align}
then $[e_+,e_-]=-if$. If $\chi_G$ is an admissible characteristic, then in
general there are infinite solutions, but we know that are all equivalent so
it is sufficient to choose one, all the other ones being related to it by
conjugation with elements of the group. Notice that the resulting equations
are in general 
\begin{align}
\sum_{\alpha\neq \beta\in \Sigma_{\chi_G}} k_\alpha k_\beta
[\lambda_\alpha,\tilde\lambda_\beta]&=0,  \label{eqmista} \\
\sum _{\beta\in \Sigma_{\chi_G}} k_\beta^2 n_{\beta,j} &=p_j,  \label{eqdiag}
\end{align}
where we used that any positive root can be written as 
\begin{align}
\beta=\sum_{j=1}^r n_{\beta,j} \alpha_j,
\end{align}
with $n_{\beta,j}$ non-negative integers, and 
\begin{align}
[\tilde \lambda_\beta,\lambda_\beta]=i\sum_{j=1}^r n_{\beta,j} J_j.
\end{align}
In the particular case when $d_j=2$ for all $j$, $\Sigma_{\chi_G}$ consists
of all simple roots and the solution is easily obtained as 
\begin{align}
e_+=\sum_{j=1}^r \sqrt {p_j} \lambda_j, \qquad e_-=\sum_{j=1}^r \sqrt {p_j}
\tilde \lambda_j.
\end{align}
Finally, we can go back to our real case by taking 
\begin{align}
T_1=\frac 1{2} (e_++e_-), \quad\ T_2=\frac 1{2i} (e_+-e_-), \quad\ T_3=\frac
12 f.
\end{align}
Notice that this is the same construction we used to get a periodic
generator $\kappa$ for the Lasagna configurations. This also shows that
indeed we can construct a $\kappa$ matrix for each three dimensional
subalgebra.


\subsection{Energy density and Baryon charge}

Let us determine the energy density and the Baryon charge. The energy
density is defined by the $T_{tt}$ component of the energy-momentum tensor 
\begin{align}
T_{\mu\nu}&=-\frac{K}{2}\mbox{Tr}(T_iT_j)\left[\mathcal{L}_{\mu}^i\mathcal{L}%
_{\nu}^j-\frac{1}{2}g_{\mu\nu}{\mathcal{L}^{\rho}}^i{\mathcal{L}_{\rho}}^j+%
\frac{\lambda}{4}\left(g^{\rho\sigma}G_{\mu\rho}^iG_{\nu\sigma}^j-\frac{1}{4}
g_{\mu\nu}{G^{\rho\sigma}}^iG_{\rho\sigma}^j\right)\right].
\end{align}
A direct computation gives 
\begin{align}
T_{tt}=2I_{G,\rho}\frac{Kp}{4L_\phi^2 L_r L_\theta}\left[\rho_0+2\sin^2(q%
\theta)\rho_1\right],
\end{align}
with 
\begin{align}
\rho_0 = & \frac{L_\phi^2}{p}\biggl[ 4L_r^2q^2 \sin^2\left(\frac{\chi}{2}%
\right) +\biggl(L_\theta^2+q^2 \lambda \sin^2\left(\frac{\chi}{2}\right)%
\biggl)\chi^{\prime 2 } \biggl] , \\
\rho_1 = & p\sin^2\left(\frac{\chi}{2}\right) \biggl[ 4L_r^2\biggl( %
L_\theta^2+q^2\lambda \sin^2\left(\frac{\chi}{2}\right) \biggl) +L_\theta^2
\lambda \chi^{\prime 2 }\biggl] \ .
\end{align}
$I_{G,\rho}$ is the Dynkin index and can be computed as follows. First,
observe that a generic root has the form 
\begin{align}
\beta(f) = \sum_{j=1}^{r}n_{\beta,j}\alpha_j(f) =
\sum_{j=1}^{r}n_{\beta,j}d_j.
\end{align}
By using (\ref{323}), (\ref{eqdiag}) and the definition of $\Sigma_{\chi_G}$%
, we get 
\begin{align}
\mbox{Tr}(ff) &= -\sum_{j=1}^{r}\sum_{k=1}^{r}p_jp_k(\alpha_k|\alpha_j) =
-\sum_{j=1}^{r}\sum_{k=1}^{r}\frac {\|\alpha_k\|^2}2p_jp_kC^G_{kj}  \notag \\
&= -\sum_{j=1}^{r}p_jd_j=-\sum_{j=1}^{r}\sum _{\beta\in \Sigma_{\chi_G}}
k_\beta^2 n_{\beta,j}d_j = -\sum _{\beta\in \Sigma_{\chi_G}} k_\beta^2
\beta(f) = -2\sum _{\beta\in \Sigma_{\chi_G}} k_\beta^2 \\
\mbox{Tr}(e_+e_-) &= -\sum _{\beta\in \Sigma_{\chi_G}}k_{\beta}^2
\end{align}
Therefore, 
\begin{align}
\mbox{Tr}(T_iT_j) = -\frac{\delta_{ij}}{2}\sum _{\beta\in
\Sigma_{\chi_G}}k_{\beta}^2,
\end{align}
and so 
\begin{align}
I_{G,\rho} = \sum _{\beta\in \Sigma_{\chi_G}}\frac{k_{\beta}^2}{4}.
\end{align}
The Baryon charge can be written as 
\begin{align}
B=\frac{1}{24 \pi^2}\int\rho_B \sqrt {g} drd\theta d\phi,
\end{align}
in which $\rho_B$ is the Baryonic density charge 
\begin{align}
\rho_B = \epsilon^{ijk}\mbox{Tr}(\mathcal{L}_i \mathcal{L}_j \mathcal{L}_k).
\end{align}
Recalling the ranges (\ref{ranges-spag}) for the coordinates and that $q=2v+1
$ and $\chi(0)=0$, we get 
\begin{align}
B=\frac{2p}{\pi} I_{G,\rho}\chi(2\pi).
\end{align}
The boundary conditions on $\chi(r)$ depend on the periodicity of $\tau_1$,
which corresponds to the periodicity of $T_3$ ($T_3=\tau_1(\Theta = \pi)$).
We must have $\chi(2\pi)=n\pi T_{G,\rho}$, so that 
\begin{align}
B=2np I_{G,\rho} T_{G,\rho},
\end{align}
where $T_{G,\rho} = 1$ for representations with even dimension and $%
T_{G,\rho} = 2$ for representations with odd dimension.


\subsection{On the "classical" isospin of these configurations}

We have shown in previous sections that the inclusion of a suitable
time-dependence in the ans\"{a}tze, both for lasagna and spaghetti phases
(see Eqs. \eqref{Ansatz-FP} and \eqref{exponential}), is one of the key ingredients
that allows the field equations to be considerably reduced, leading to a
single integrable ODE equation for the profiles. This time-dependence offers
a nice short-cut to estimate the \textquotedblleft classical Isospin" of the
configurations analyzed in the present paper (a relevant question is whether
or not the classical Isospin is large when the Baryonic charge is large). In
particular, one may evaluate the \textquotedblleft cost" of removing such
time-dependence. Such a cost is related to the internal Isospin symmetry of
the theory. This is like trying to estimate the angular momentum of a
spinning top by evaluating the cost to make the spinning top to stop
spinning. In the present case, the time-dependence of the configurations can
be removed from the ans\"{a}tze by introducing a Isospin chemical potential;
then the isospin chemical potential needed to remove such time-dependence is
a measure of the classical Isospin of the present configurations. We will
see how this works for the simplest $SU(2)$ case, where the generators are $%
T_{j}=i\sigma _{j}$, being $\sigma _{j}$ the Pauli matrices (general group $G$
behave in a similar way).

As it is well known, the effects of the Isospin chemical potential can be
taken into account by using the following covariant derivative 
\begin{equation}
\nabla _{\mu }\rightarrow D_{\mu }=\nabla _{\mu }+\bar{\mu}[T_{3},\cdot
]\delta _{\mu 0}\ .  \label{Dmu}
\end{equation}%
Now, we will use exactly the same ansatz as before in the spaghetti $SU(2)$
case, \textit{but this time without the time dependence}: 
\begin{gather*}
U=e^{\chi (x)\,(\vec{n}\cdot \vec{T})}\ , \\
\vec{n}=(\sin {\Theta }\sin \Phi ,\sin \Theta \cos \Phi ,\cos \Theta )\ ,
\end{gather*}%
where 
\begin{gather*}
\chi =\chi \left( r\right) \ ,\quad \Theta =q\theta \ ,\quad \Phi =p\phi \ ,
\\
q=\frac{1}{2}(2v+1)\ ,\quad p,v\in \mathbb{N}\ ,\quad p\neq 0\ ,
\end{gather*}%
together with the introduction of the Isospin chemical potential in Eq. %
\eqref{Dmu} in the theory. One can check directly that the complete set of
Skyrme equations can still be reduced to the same ODE for the profile $\chi
\left( r\right) $ in the case of the spaghetti phase in Eq. \eqref{SkyrmeEqsSpaghetti} 
\textit{only provided the Isospin chemical potential for the spaghetti phase
is given by} 
\begin{equation}
\bar{\mu}_{\text{S}}=\frac{p}{L_{\phi }}\ .
\end{equation}%
In other word, the cost to eliminate the time-dependence is to introduce an
Isospin chemical potential which is large when the Baryonic charge of the
spaghetti is large. Something similar happens in the case of the lasagna
phase. Let us consider the ansatz in terms of the Euler angles \textit{but
without the time-dependence} for the $SU(2)$ case: 
\begin{equation*}
U_{L}=e^{\Phi T_{3}}e^{HT_{2}}e^{\Theta T_{3}}\ ,
\end{equation*}%
where 
\begin{equation*}
\Phi =p\phi \ ,\quad H=h(r)\ ,\quad \Theta =m\theta \ ,\qquad p,m\in \mathbb{%
N}\ .
\end{equation*}%
Let us introduce the Isospin chemical potential, demanding that the profile $%
h(r)$ should be the same as before. Then, as in the spaghetti case, the
Skyrme field equations with chemical potential can still be satisfied by the
very same profile $h(r)$ \textit{provided we fix the Isospin chemical
potential as} 
\begin{equation}
\bar{\mu}_{\text{L}}=\frac{pm}{(p^{2}L_{\phi }^{2}+m^{2}L_{\theta }^{2})^{%
\frac{1}{2}}}\ .
\end{equation}%
At this point it is important to remember that in the $SU(2)$ case the
lasagna and spaghetti type solutions have the following values for the
topological charges 
\begin{equation*}
B_{\text{S}}=np\ ,\qquad B_{\text{L}}=mp\ ,
\end{equation*}%
see \cite{56c} and \cite{crystal1} for more details. 
These arguments show that the
\textquotedblleft classical Isospin" of configurations with high Baryonic
charge is large. Finally, it is important to point out that the large
Isospin case corresponds to either neutron rich or proton rich matter and
due to Coulomb effects (not taken into account in this model), the neutron
rich solution is preferred. This fact is very convenient as far as the
physics of neutron stars is concerned.


\section{Examples: exceptional pasta}

As an example we can consider the ``basic exceptional Skyrmions'', that are
solutions in lowest dimensional representation when $G$ is one of the
exceptional Lie groups. There are five cases that we now recall according to
the dimension of the group. For each of them we know all inequivalent three
dimensional subalgebras, each one determined by the Dynkin characteristic $%
\chi_I (d_1,\ldots,d_r)$, where $I$ is the Dynkin index and $d_j$ are the
coefficients of the characteristic, ordered as the simple root listed in
Appendix \ref{app:roots}.

\ 

The smallest exceptional group is $G_2$, a 14 dimensional group of rank 2
whose smallest irrep is 7 dimensional. There are four different three
dimensional subalgebras. It contains four 3D subalgebras, having
characteristics 
\begin{align*}
\chi_1=(0,1), \quad\ \chi_3=(1,0), \quad\ \chi_4=(0,2), \quad\
\chi_{28}=(2,2).
\end{align*}
$\chi_1$ and $\chi_2$ are regular but not integral, while $\chi_4$ and $%
\chi_{28}$ are not regular but are integral. In particular, the minimal
regular subalgebra containing $\chi_4$ is $\chi_1\oplus \chi_3$, while $%
\chi_{28}$ is maximal so that the smallest regular subalgebra containing it
is $G_2$ itself.

\ 

The next group is $F_4$, a 52 dimensional group of rank 4. Its smallest
irrep is 26 dimensional. It contains 15 $su(2)$ type subalgebras, whose
characteristics are 
\begin{align*}
\chi_1&=(1,0,0,0); \quad\ \chi_2=(0,0,0,1); \quad\ \chi_3=(0,1,0,0); \quad\
\chi_4=(2,0,0,0); \quad\ \chi_6=(0,0,1,0); \\
\chi_8&=(0,0,0,2); \quad\ \chi_9=(0,1,0,1); \quad\ \chi_{10}=(2,0,0,1);
\quad\ \chi_{11}=(1,0,1,0); \quad\ \chi_{12}=(0,2,0,0); \\
\chi_{28}&=(2,2,0,0); \quad\ \chi_{35}=(1,0,1,2); \quad\
\chi_{36}=(0,2,0,2); \quad\ \chi_{60}=(2,2,0,2); \quad\ \chi_{156}=(2,2,2,2).
\end{align*}
The regular subalgebras are $\chi_1$ and $\chi_2$, which are not integral.
The integral subalgebras are $\chi_4,$ $\chi_8,$ $\chi_{12},$ $\chi_{28},$ $%
\chi_{36},$ $\chi_{60}$ and $\chi_{156}$. In particular, $\chi_{156}$ is
maximal.

\ 

The third group is $E_6$, a 78 dimensional group of rank 6. Its smallest
irrep is 27 dimensional. It contains 20 $su(2)$ type subalgebras, whose
characteristics are 
\begin{align*}
\chi_1&=(0,1,0,0,0,0); \quad\ \chi_2=(1,0,0,0,0,1); \quad\
\chi_3=(0,0,0,1,0,0); \quad\ \chi_4=(0,2,0,0,0,0); \\
\chi_5&=(1,1,0,0,0,1);\quad\ \chi_6=(0,0,1,0,1,0); \quad\
\chi_8=(2,0,0,0,0,2); \quad\ \chi_{9}=(1,0,0,1,0,1); \\
\chi_{10}&=(1,2,0,0,0,1); \quad\ \chi_{11}=(0,1,1,0,1,0);\quad\
\chi_{12}=(0,0,0,2,0,0); \quad\ \chi_{20}=(2,2,0,0,0,2); \\
\chi_{21}&=(1,1,1,0,1,1); \quad\ \chi_{28}=(0,2,0,2,0,0); \quad\
\chi_{30}=(1,2,1,0,1,1)\quad\ \chi_{35}=(2,1,1,0,1,2); \\
\chi_{36}&=(2,0,0,2,0,2); \quad\ \chi_{60}=(2,2,0,2,0,2); \quad\
\chi_{84}=(2,2,2,0,2,2)\quad\ \chi_{156}=(2,2,2,2,2,2).
\end{align*}
The only regular subalgebra is $\chi_1$, which is not integral. The integral
subalgebras are $\chi_4, \chi_8, \chi_{12}, \chi_{20}, \chi_{28}, \chi_{36},
\chi_{60}$ and $\chi_{156}$. The last one is also maximal.

\ 

The third group is $E_7$, a 133 dimensional group of rank 7. Its smallest
irrep is 58 dimensional. It contains 44 $su(2)$ type subalgebras, whose
characteristics are 
\begin{align*}
\chi_1&=(1,0,0,0,0,0,0); \quad\ \chi_2=(0,0,0,0,0,1,0); \quad\
\chi_{3^{\prime }}=(0,0,1,0,0,0,0); \quad\ \chi_{3^{\prime \prime
}}=(0,0,0,0,0,0,2); \\
\chi_{4^{\prime }}&=(2,0,0,0,0,0,0);\quad\ \chi_{4^{\prime \prime
}}=(0,1,0,0,0,0,1); \quad\ \chi_5=(1,0,0,0,0,1,0); \quad\
\chi_{6}=(0,0,0,1,0,0,0); \\
\chi_{7}&=(0,2,0,0,0,0,0); \quad\ \chi_{8}=(0,0,0,0,0,2,0);\quad\
\chi_{9}=(0,0,1,0,0,1,0); \quad\ \chi_{10}=(2,0,0,0,0,1,0); \\
\chi_{11^{\prime }}&=(1,0,0,1,0,0,0); \quad\ \chi_{11^{\prime \prime
}}=(2,0,0,0,0,0,2); \quad\ \chi_{12^{\prime }}=(0,0,2,0,0,0,0)\quad\
\chi_{12^{\prime \prime }}=(1,0,0,0,1,0,1); \\
\chi_{13}&=(0,1,1,0,0,0,1); \quad\ \chi_{14}=(0,0,0,1,0,1,0); \quad\
\chi_{15}=(0,0,0,0,2,0,0)\quad\ \chi_{20}=(2,0,0,0,0,2,0); \\
\chi_{21}&=(1,0,0,1,0,1,0); \quad\ \chi_{24}=(0,0,0,2,0,0,0); \quad\
\chi_{28}=(2,0,2,0,0,0,0); \quad\ \chi_{29}=(2,1,1,0,0,0,1); \\
\chi_{30}&=(2,0,0,1,0,1,0);\quad\ \chi_{31}=(2,0,0,0,2,0,0); \quad\
\chi_{35^{\prime }}=(1,0,0,1,0,2,0); \quad\ \chi_{35^{\prime \prime
}}=(2,0,0,0,0,2,2); \\
\chi_{36^{\prime }}&=(0,0,2,0,0,2,0); \quad\ \chi_{36^{\prime \prime
}}=(1,0,0,1,0,1,2);\quad\ \chi_{38}=(0,1,1,0,1,0,2); \quad\
\chi_{39}=(0,0,0,2,0,0,2); \\
\chi_{56}&=(0,0,0,2,0,2,0); \quad\ \chi_{60}=(0,0,2,0,0,2,0); \quad\
\chi_{61}=(2,1,1,0,1,1,0)\quad\ \chi_{62}=(2,1,1,0,1,0,2); \\
\chi_{63}&=(2,0,0,2,0,0,2); \quad\ \chi_{84}=(2,0,0,2,0,2,0); \quad\
\chi_{110}=(2,1,1,0,1,2,2)\quad\ \chi_{111}=(2,0,0,2,0,2,2); \\
\chi_{156}&=(2,0,2,2,0,2,0); \quad\ \chi_{159}=(2,2,2,0,2,0,2); \quad\
\chi_{231}=(2,2,2,0,2,2,2)\quad\ \chi_{399}=(2,2,2,2,2,2,2).
\end{align*}
The only regular subalgebra is $\chi_1$, which is not integral. The integral
subalgebras are $\chi_{3^{\prime \prime }}, \chi_{4^{\prime }}, \chi_{7},
\chi_{8}, \chi_{11^{\prime \prime }}, \chi_{12^{\prime }},$ $\chi_{15},
\chi_{20}, \chi_{24}, \chi_{28}, \chi_{31}, \chi_{35^{\prime \prime }},
\chi_{36^{\prime }}, \chi_{39}, \chi_{56}, \chi_{60}, \chi_{63}, \chi_{84},
\chi_{111}, \chi_{156}, \chi_{159}, \chi_{231}$ and $\chi_{399}$. The last
one is also maximal.

\ 

The third group is $E_8$, a 248 dimensional group of rank 8. Its smallest
irrep is 248 dimensional. It contains 70 $su(2)$ type subalgebras, whose
characteristics are {\small 
\begin{align*}
\chi_{1}&=(0,0,0,0,0,0,0,1); \quad\ \chi_{2}=(1,0,0,0,0,0,0,0); \quad\
\chi_{3}=(0,0,0,0,0,0,1,0); \quad\ \chi_{4^{\prime }}=(0,1,0,0,0,0,0,0); \\
\chi_{4^{\prime \prime }}&=(0,0,0,0,0,0,0,2);\quad\
\chi_{5}=(1,0,0,0,0,0,0,1); \quad\ \chi_{6}=(0,0,0,0,0,1,0,0); \quad\
\chi_{7}=(0,0,1,0,0,0,0,0); \\
\chi_{8}&=(2,0,0,0,0,0,0,0); \quad\ \chi_{9}=(1,0,0,0,0,0,1,0);\quad\
\chi_{10^{\prime }}=(2,0,0,0,0,0,0,1); \quad\ \chi_{10^{\prime \prime
}}=(0,0,0,0,1,0,0,0); \\
\chi_{11}&=(0,0,0,0,0,1,0,1); \quad\ \chi_{12^{\prime }}=(0,0,0,0,0,0,2,0);
\quad\ \chi_{12^{\prime \prime }}=(0,0,1,0,0,0,0,1)\quad\
\chi_{13}=(0,1,0,0,0,0,1,0); \\
\chi_{14}&=(1,0,0,0,0,1,0,0); \quad\ \chi_{15}=(0,0,0,1,0,0,0,0); \quad\
\chi_{16}=(0,2,0,0,0,0,0,0)\quad\ \chi_{20^{\prime }}=(1,0,0,0,1,0,0,0); \\
\chi_{20^{\prime \prime }}&=(2,0,0,0,0,0,0,2); \quad\
\chi_{21}=(1,0,0,0,0,1,0,1); \quad\ \chi_{22}=(0,1,0,0,0,0,0,1); \quad\
\chi_{24}=(0,0,0,0,0,2,0,0); \\
\chi_{25}&=(0,0,1,0,0,1,0,0);\quad\ \chi_{28}=(0,0,0,0,0,0,2,2); \quad\
\chi_{29}=(0,1,0,0,0,0,1,2); \quad\ \chi_{30^{\prime }}=(1,0,0,0,0,1,0,2); \\
\chi_{30^{\prime \prime }}&=(0,0,0,1,0,0,1,0); \quad\
\chi_{31}=(0,0,0,1,0,0,0,2);\quad\ \chi_{32}=(0,2,0,0,0,0,0,2); \quad\
\chi_{34}=(0,0,1,0,0,1,0,1); \\
\chi_{35}&=(2,0,0,0,0,1,0,1); \quad\ \chi_{36^{\prime }}=(1,0,0,1,0,0,0,1);
\quad\ \chi_{36^{\prime \prime }}=(2,0,0,0,0,0,2,0)\quad\
\chi_{37}=(1,0,0,0,1,0,1,0); \\
\chi_{38}&=(0,1,1,0,0,0,1,0); \quad\ \chi_{39}=(0,0,0,1,0,1,0,0); \quad\
\chi_{40}=(0,0,0,0,2,0,0,0)\quad\ \chi_{56}=(2,0,0,0,0,2,0,0); \\
\chi_{57}&=(1,0,0,1,0,1,0,0); \quad\ \chi_{60}=(2,0,0,0,0,0,2,2); \quad\
\chi_{61}=(1,0,0,0,1,0,1,2); \quad\ \chi_{62}=(0,1,1,0,0,0,1,2); \\
\chi_{63}&=(0,0,0,1,0,1,0,2);\quad\ \chi_{64}=(0,0,0,0,2,0,0,2); \quad\
\chi_{70}=(1,0,0,1,0,1,0,1); \quad\ \chi_{84^{\prime }}=(1,0,0,1,0,1,1,0); \\
\chi_{84^{\prime \prime }}&=(2,0,0,0,0,2,0,2); \quad\
\chi_{85}=(1,0,0,1,0,1,0,2);\quad\ \chi_{88}=(0,0,0,2,0,0,0,2); \quad\
\chi_{110}=(2,1,1,0,0,0,1,2); \\
\chi_{111}&=(2,0,0,1,0,1,0,2); \quad\ \chi_{112}=(2,0,0,0,2,0,0,2); \quad\
\chi_{120}=(0,0,0,2,0,0,2,0)\quad\ \chi_{156}=(2,0,0,0,0,2,2,2); \\
\chi_{157}&=(1,0,0,1,0,1,2,2); \quad\ \chi_{159}=(0,1,1,0,1,0,2,2); \quad\
\chi_{160}=(0,0,0,2,0,0,2,2)\quad\ \chi_{166}=(1,0,1,1,0,0,2,2); \\
\chi_{182}&=(2,1,1,0,1,1,0,1); \quad\ \chi_{184}=(2,0,0,2,0,0,2,0); \quad\
\chi_{231}=(2,1,1,0,1,0,2,2); \quad\ \chi_{232}=(2,0,0,2,0,0,2,2); \\
\chi_{280}&=(2,0,0,2,0,2,0,2);\quad\ \chi_{399}=(2,1,1,0,1,1,2,2); \quad\
\chi_{400}=(2,0,0,2,0,2,2,2); \quad\ \chi_{520}=(2,2,2,0,2,0,2,2); \\
\chi_{760}&=(2,2,2,0,2,2,2,2); \quad\ \chi_{1240}=(2,2,2,2,2,2,2,2).
\end{align*}
} The only regular subalgebra is $\chi_1$, which is not integral. The
integral subalgebras are $\chi_{4^{\prime \prime }}, \chi_{8},
\chi_{12^{\prime }}, \chi_{16}, \chi_{20^{\prime \prime }}, \chi_{24},$ $%
\chi_{28}, \chi_{32}, \chi_{36^{\prime \prime }}, \chi_{40}, \chi_{56},
\chi_{60}, \chi_{64}, \chi_{84^{\prime \prime }}, \chi_{112}, \chi_{120},
\chi_{156}, \chi_{160}, \chi_{184}, \chi_{232}, \chi_{280}, \chi_{400},
\chi_{520}, \chi_{760}$ and $\chi_{1240}$. The last one is also maximal.

\ 

As an example, we will finally construct the explicit solutions for $G_2$,
which we can call ``$G_2$ exceptional pasta''. 

\subsection{$G_2$ exceptional Spaghetti}

Here we consider explicit solutions case by case. Our deduction will be
quite general and independent on the specific realization in terms of
matrices, but just on the chosen representation. Nevertheless, for sake of
completeness, in Appendix \ref{app:G2} we will provide an explicit matrix
realization of the subalgebras in the lowest fundamental representation. 

\subsubsection{$\protect\chi_1$-Spaghetti}

Since $\chi_1=(0,1)$, we get $p\equiv (p_1,p_2)=(6,4)$. The only root
satisfying $\alpha(f)=2$ is $\alpha_6=3\alpha_1+2\alpha_2$. Therefore, $%
e_+=k \lambda_6$ and equation (\ref{eqmista}) is trivial, while (\ref{eqdiag}%
) gives 
\begin{align}
k=\sqrt 2.
\end{align}
Therefore, the Spaghetti solution is determined by the matrices 
\begin{align}
T^{(1)}_1&=\frac {\sqrt 2}{2} (\lambda_6+\tilde \lambda_6), \\
T^{(1)}_2&=\frac {\sqrt 2}{2i} (\lambda_6-\tilde \lambda_6), \\
T^{(1)}_3&=3J_1+2J_2.
\end{align}
Notice that up to now this is independent on the choice of the irrep. The
choice of the representation allows to further specify the type of solution.
The fundamental representations of $G_2$ are the $\pmb 7$, with maximal
weight $\alpha_4$, whose seven weight are on the small hexagon given by $%
\pm\alpha_a$, $a=1,3,4$, plus one vanishing weight, and the $\pmb {14}$
which is the adjoint representation, with maximal weight $\alpha_6$ and with
all roots as weight. The action of $\pm\alpha_6$ on the small hexagon shows
that if we choose to work with the $\pmb 7$, then $\mathbb{R}^7$ decomposes
as $\pmb 2\oplus \pmb 2\oplus \pmb 1\oplus \pmb 1\oplus \pmb 1$ under $%
\chi_1 $, so that is an $SU(2)$ type solution. \newline
For $\pmb {14}$, we see that the action of $\pm\alpha_6$ decomposes $\mathbb{%
R}^{14}$ into $\pmb 3\oplus \pmb 2\oplus \pmb 2 \oplus \pmb 2\oplus \pmb 2
\oplus \pmb 1\oplus \pmb 1\oplus \pmb 1$, which is again an $SU(2)$ type
solution.

\subsubsection{$\protect\chi_3$-Spaghetti}

Since $\chi_3=(1,0)$, we get $p\equiv (p_1,p_2)=(12,6)$. The only root
satisfying $\alpha(f)=2$ is $\alpha_4=2\alpha_1+\alpha_2$. Therefore, $e_+=k
\lambda_4$ and equation (\ref{eqmista}) is trivial, while (\ref{eqdiag})
gives 
\begin{align}
k=\sqrt 6.
\end{align}
Therefore, the Spaghetti solution is determined by the matrices 
\begin{align}
T^{(3)}_1&=\frac {\sqrt 6}{2} (\lambda_4+\tilde \lambda_4), \\
T^{(3)}_2&=\frac {\sqrt 6}{2i} (\lambda_4-\tilde \lambda_4), \\
T^{(3)}_3&=6J_1+3J_2.
\end{align}
The action of $\pm \alpha_6$ on the small hexagon shows that if we choose to
work with the $\pmb 7$, then $\mathbb{R}^7$ decomposes as $\pmb 3\oplus \pmb %
2\oplus \pmb 2$ under $\chi_3$, so that is an $SU(2)$ type solution. \newline
For $\pmb {14}$, we see that the action of $\pm \alpha_4$ decomposes $%
\mathbb{R}^{14}$ into $\pmb 4\oplus \pmb 4\oplus \pmb 3 \oplus \pmb 1\oplus %
\pmb 1\oplus \pmb 1$, which is again an $SU(2)$ type solution, since it
contains even dimensional subrepresentations.

\subsubsection{$\protect\chi_4$-Spaghetti}

Since $\chi_4=(0,2)$, we get $p\equiv (p_1,p_2)=(12,8)$. This time there are
four roots satisfying the condition $\alpha(f)=2$, which are $\alpha_2,
\alpha_3, \alpha_4$ and $\alpha_5$. Thus, we can put $e_+=\sum_{j=2}^5
k_j\lambda_j$ and $e_-=\sum_{j=2}^5 k_j\tilde \lambda_j$. Using the results
of Appendix \ref{app:G2}, we see that equations (\ref{eqmista}) and (\ref%
{eqdiag}) become 
\begin{align}
\frac 1{\sqrt 2} k_2k_3 +\sqrt {\frac 23} k_3 k_4+\frac 1{\sqrt 2} k_4
k_5&=0, \\
k_3^2+2k_4^2=3k_5^2&=12, \\
k_2^2+k_3^2+k_4^2+k_5^2&=8.
\end{align}
There are several solutions of this system, but we know that we just need to
find one. A very simple choice is 
\begin{align}
k_3=k_4=0, \qquad k_2=k_5=2.
\end{align}
Therefore, the Spaghetti solution is determined by the matrices 
\begin{align}
T^{(4)}_1&= \lambda_2+\tilde \lambda_2+\lambda_5+\tilde \lambda_5, \\
T^{(4)}_2&=-i(\lambda_2-\tilde \lambda_2+\lambda_5-\tilde \lambda_5), \\
T^{(4)}_3&=6J_1+4J_2.
\end{align}
To understand the type of solution, we notice that the action of of $T_1$
and $T_2$ leave invariant the subspaces $\langle \lambda_3, \tilde
\lambda_1, \tilde \lambda_4\rangle$ and $\langle \tilde\lambda_3, \lambda_1,
\lambda_4\rangle$, so that in the representation $\pmb 7$, $\mathbb{R}^7$
decomposes as $\pmb 3\oplus \pmb 3\oplus \pmb 1$. We see that it is a $SO(3)$%
-type solution. \newline
Starting from the adjoint, we see that the action $\tilde \lambda_2+\tilde
\lambda_5$ applied repeatedly to $\lambda_6$ generates a combination of $%
\lambda_2$ and $\lambda_5$, then an element of $H$, then a combination of $%
\tilde \lambda_2$ and $\tilde \lambda_5$, and finally $\tilde \lambda_6$.
This shows that working with $\pmb {14}$, $\mathbb{R}^{14}$ decomposes as $%
\pmb 5\oplus \pmb 3\oplus \pmb 3\oplus \pmb 3\oplus \pmb 1$. Again, it is an 
$SO(3)$ type solution.

\subsubsection{$\protect\chi_{28}$-Spaghetti}

This is the principal case, with $\chi_{28}=(2,2)$. Therefore $p\equiv
(p_1,p_2)=(36,20)$. We already know the solution in this case. The Spaghetti
solution is 
\begin{align}
T^{(28)}_1&= 3(\lambda_1+\tilde \lambda_1)+\sqrt 5 (\lambda_2+\tilde
\lambda_2), \\
T^{(28)}_2&=-i3(\lambda_1-\tilde \lambda_1)-i\sqrt 5 (\lambda_2-\tilde
\lambda_2), \\
T^{(28)}_3&=18J_1+10J_2.
\end{align}
Because of Proposition \ref{prop-dynkin}, we already know that working in
the adjoint the solution is of $SO(3)$-type. In the representation $\pmb 7$,
it is sufficient to verify that for $T_-=3\tilde \lambda_1+\sqrt 5 \tilde
\lambda_2$, and $v$ the maximal vector of $\pmb 7$, then the vectors $\rho_{%
\pmb 7}^k(T_-) (v)$, $k=0,\ldots, 6$ are all linearly independent. Here $%
\rho_{\pmb 7}:G_2\rightarrow End(\mathbb{R}^7)$ is the representation map of
the algebra. This is proved in Appendix \ref{app:G2} and proves that $%
\mathbb{R}^7$ is irreducible under $\chi_{28}$. Since it is odd dimensional,
it is of $SO(3)$-type.


\subsection{$G_2$ exceptional Lasagna}

For the exceptional Lasagna we can use Proposition \ref{prop-dynkin}. Since
we already know that $n=b$ must be equal to $1$, we get that $%
(p_1,p_2)=(18,10)$, and, if we fix $\psi_j=0$ for simplicity, then 
\begin{align}
\kappa=T^{(28)}_1=3(\lambda_1+\tilde \lambda_1)+\sqrt 5 (\lambda_2+\tilde
\lambda_2).
\end{align}
Moreover, from Proposition \ref{propuno} we get 
\begin{align}
h(z)=\frac z2 T^{(28)}_3.
\end{align}
This defines the simplest exceptional $G_2$ Lasagna.

\ 

\ 


\section{Extended ansatz}

In order to allow for further generalizations, it is convenient to employ
the Euler parameterization in a more general ansatz, after fixing the
matrices $\kappa$ and $f$. Let us consider the Skyrmionic field\footnote{%
In this section we will use the coordinates $\{t,r,\phi,\gamma\}$, however
the results are applicable for both the lasagna and the spaghetti phases.} 
\begin{align}
U(t,r,\phi,\gamma)=e^{\Phi(t,r,\phi,\gamma)\kappa}
e^{\chi(t,r,\phi,\gamma)f} e^{\Theta(t,r,\phi,\gamma)\kappa},
\label{GeneralEuler}
\end{align}
where $\kappa$ is specified in (\ref{kappa}), and $f$ has the same
properties as in (\ref{f}). One of the aim of this generalization is to
provide a description of different \textit{pasta states} without specifying
them \textit{a priori}. This could lead to a comprehensive description of
Skyrmions in a finite volume and to an analytical definition of other
possible states (such as \textit{gnocchi states}) and the transitions
between them. In this work, we did not analyze all these possibilities and
all the limits of this models, but we outline the main properties which
characterize them, namely the wave equations, the topological charge and the
energy density. If we define 
\begin{align}
\alpha=\frac{1}{2}\left(\Theta-\Phi\right)\qquad\mbox{and}\qquad\xi=\frac{1}{%
2}\left(\Phi+\Theta\right),
\end{align}
then 
\begin{align}
U(t,r,\phi,\gamma)=e^{-\alpha(t,r,\phi,\gamma)
\kappa}e^{\xi(t,r,\phi,\gamma) \kappa}e^{\chi(t,r,\phi,\gamma)
f}e^{\xi(t,r,\phi,\gamma) \kappa}e^{\alpha(t,r,\phi,\gamma) \kappa}.
\end{align}
This gives 
\begin{align}
\mathcal{L}_{\mu}=e^{-\alpha \kappa}e^{-\xi \kappa}\left[\partial_{\mu}%
\alpha(\kappa-\hat{\kappa})+\partial_{\mu}\xi(\kappa+\hat{\kappa}%
)+\partial_{\mu}\chi f\right]e^{\xi \kappa}e^{\alpha \kappa},
\end{align}
where we introduced the matrix function 
\begin{align}
\hat{\kappa}=e^{-\chi f}\kappa e^{\chi f}.  \label{chicappuccio}
\end{align}
Since 
\begin{align}
\mbox{tr}(\lambda_j\lambda_k)=0,\qquad\mbox{tr}(\lambda_j\tilde{\lambda}%
_k)=-\delta_{jk},
\end{align}
we have 
\begin{align}
\mbox{tr}(\kappa^2)=\mbox{tr}(\hat{\kappa}^2)=-2\|c\|^2,
\end{align}
and $f$ can be normalized so that 
\begin{align}  \label{Normf}
\mbox{tr}(f^2)=\mbox{tr}(\kappa^2).
\end{align}
This leads to the condition (\ref{CondCF}) and, in particular, $|c_j|^2=%
\frac{b}{2}p_j$. Using these conventions we can now write the Skyrme
equation explicitly.


\subsection{Non-linear wave equations}

We call wave equations to the field equations for the functions $\alpha$, $%
\xi$ and $\chi$. These result to be {\small 
\begin{align}
&\partial_{\mu}\partial^{\mu}\chi \left\{1+b^2\lambda\left[%
\partial_{\mu}\alpha\partial^{\mu}\alpha\sin^2\left(\frac{b\chi}{2}%
\right)+\partial_{\mu}\xi\partial^{\mu}\xi\cos^2\left(\frac{b\chi}{2}\right)%
\right]\right\} -b\sin(b\chi)\left(1-\frac{b^2\lambda}{4}\partial_{\mu}\chi%
\partial^{\mu}\chi\right)\left(\partial_{\nu}\alpha\partial^{\nu}\alpha-%
\partial_{\nu}\xi\partial^{\nu}\xi\right) \cr & -b^3\lambda\sin(b\chi)\cos(b%
\chi)\left[\partial_{\mu}\alpha\partial^{\mu}\alpha\partial_{\nu}\xi%
\partial^{\nu}\xi-\left(\partial_{\mu}\alpha\partial^{\mu}\xi\right)^2\right]
-b^2\lambda\left\{\sin^2\left(\frac{b\chi}{2}\right)\partial_{\mu}\partial^{%
\mu}\alpha\partial_{\nu}\alpha\partial^{\nu}\chi+\cos^2\left(\frac{b\chi}{2}%
\right)\partial_{\mu}\partial^{\mu}\xi\partial_{\nu}\xi\partial^{\nu}\chi%
\right.\cr &+\sin^2\left(\frac{b\chi}{2}\right)\left[\partial_{\mu}\alpha%
\partial^{\mu}\left(\partial_{\nu}\alpha\partial^{\nu}\chi\right)-\partial_{%
\mu}\chi\partial^{\mu}\left(\partial_{\nu}\alpha\partial^{\nu}\alpha\right)%
\right] \left.+\cos^2\left(\frac{b\chi}{2}\right)\left[\partial_{\mu}\xi%
\partial^{\mu}\left(\partial_{\nu}\xi\partial^{\nu}\chi\right)-\partial_{%
\mu}\chi\partial^{\mu}\left(\partial_{\nu}\xi\partial^{\nu}\xi\right)\right]%
\right\}\cr & -\frac{b^3\lambda}{4}\sin(b\chi)\left[\left(\partial_{\mu}%
\alpha\partial^{\mu}\chi\right)^2-\left(\partial_{\mu}\xi\partial^{\mu}\chi%
\right)^2\right]=0,  \label{FirstSkyrme}
\end{align}
}

{\small 
\begin{align}
&4\cos\left(\frac{b\chi}{2}\right)\left\{ \cos\left(\frac{b\chi}{2}\right)%
\Bigl\{ \partial_{\mu}\partial^{\mu}\alpha\left[1+\frac{b^2\lambda}{4}%
\partial_{\nu}\chi\partial^{\nu}\chi\right] -\frac{b^2\lambda}{4}\left[%
\partial_{\mu}\partial^{\mu}\chi\partial_{\nu}\chi\partial^{\nu}\alpha+%
\partial_{\mu}\chi\partial^{\mu}\left(\partial_{\nu}\chi\partial^{\nu}\alpha%
\right)-\partial_{\mu}\alpha\partial^{\mu}\left(\partial_{\nu}\chi\partial^{%
\nu}\chi\right)\right] \Bigr\} \right.\cr & -\sin\left(\frac{b\chi}{2}\right)%
\Bigl\{ \frac{b^2\lambda}{2}\sin(b\chi)\partial_{\mu}\partial^{\mu}\alpha%
\partial_{\nu}\alpha\partial^{\nu}\xi -\frac{b^2\lambda}{2}\sin(b\chi)\left[%
\partial_{\mu}\partial^{\mu}\xi\partial_{\nu}\alpha\partial^{\nu}\alpha+%
\partial_{\mu}
\xi\partial^{\mu}\left(\partial_{\nu}\alpha\partial^{\nu}\alpha\right)-%
\partial_{\mu}\alpha\partial^{\mu}\left(\partial_{\nu}\alpha\partial^{\nu}%
\xi\right)\right]\cr &+b^3\lambda\cos(b\chi)\left[\partial_{\mu}\chi%
\partial^{\mu}\alpha\partial_{\nu}\alpha\partial^{\nu}\xi-\partial_{\mu}\chi%
\partial^{\mu}\xi\partial_{\nu}\alpha\partial^{\nu}\alpha\right]%
+b\partial_{\mu}\chi\partial^{\mu}\xi \Bigr\} \Bigr\} +4\sin\left(\frac{b\chi%
}{2}\right)\left\{ \sin\left(\frac{b\chi}{2}\right)\Bigl\{ %
\partial_{\mu}\partial^{\mu}\xi\left[1+\frac{b^2\lambda}{4}%
\partial_{\nu}\chi\partial^{\nu}\chi\right]\right.\cr & -\frac{b^2\lambda}{4}%
\left[\partial_{\mu}\partial^{\mu}\chi\partial_{\nu}\chi\partial^{\nu}\xi+%
\partial_{\mu}\chi\partial^{\mu}\left(\partial_{\nu}\chi\partial^{\nu}\xi%
\right)-\partial_{\mu}\xi\partial^{\mu}\left(\partial_{\nu}\chi\partial^{%
\nu}\chi\right)\right] \Bigr\} -\cos\left(\frac{b\chi}{2}\right)\Bigl\{ 
\frac{b^2\lambda}{2}\sin(b\chi)\partial_{\mu}\partial^{\mu}\xi\partial_{\nu}%
\alpha\partial^{\nu}\xi\cr &-\frac{b^2\lambda}{2}\sin(b\chi)\left[%
\partial_{\mu}\partial^{\mu}\alpha\partial_{\nu}\xi\partial^{\nu}\xi+%
\partial_{\mu}\alpha\partial^{\mu}\left(\partial_{\nu}\xi\partial^{\nu}\xi%
\right)-\partial_{\mu}\xi\partial^{\mu}\left(\partial_{\nu}\alpha\partial^{%
\nu}\xi\right)\right]\cr & +b^3\lambda\cos(b\chi)\left[\partial_{\mu}\chi%
\partial^{\mu}\xi\partial_{\nu}\alpha\partial^{\nu}\xi-\partial_{\mu}\chi%
\partial^{\mu}\alpha\partial_{\nu}\xi\partial^{\nu}\xi\right]%
-b\partial_{\mu}\chi\partial^{\mu}\alpha \Bigr\} \Bigr\}=0,
\label{SecondSkyrme}
\end{align}
\begin{align}
&4\sin\left(\frac{b\chi}{2}\right)\left\{ \cos\left(\frac{b\chi}{2}\right)%
\Bigl\{ \partial_{\mu}\partial^{\mu}\alpha\left[1+\frac{b^2\lambda}{4}%
\partial_{\nu}\chi\partial^{\nu}\chi\right] -\frac{b^2\lambda}{4}\left[%
\partial_{\mu}\partial^{\mu}\chi\partial_{\nu}\chi\partial^{\nu}\alpha+%
\partial_{\mu}\chi\partial^{\mu}\left(\partial_{\nu}\chi\partial^{\nu}\alpha%
\right)-\partial_{\mu}\alpha\partial^{\mu}\left(\partial_{\nu}\chi\partial^{%
\nu}\chi\right)\right] \Bigr\}\right.\cr & +\sin\left(\frac{b\chi}{2}\right)%
\Bigl\{ \frac{b^2\lambda}{2}\sin(b\chi)\partial_{\mu}\partial^{\mu}\alpha%
\partial_{\nu}\alpha\partial^{\nu}\xi -\frac{b^2\lambda}{2}\sin(b\chi)\left[%
\partial_{\mu}\partial^{\mu}\xi\partial_{\nu}\alpha\partial^{\nu}\alpha+%
\partial_{\mu}\xi\partial^{\mu}\left(\partial_{\nu}\alpha\partial^{\nu}%
\alpha\right)
-\partial_{\mu}\alpha\partial^{\mu}\left(\partial_{\nu}\alpha\partial^{\nu}%
\xi\right)\right]\cr &+b^3\lambda\cos(b\chi)\left[\partial_{\mu}\chi%
\partial^{\mu}\alpha\partial_{\nu}\alpha\partial^{\nu}\xi-\partial_{\mu}\chi%
\partial^{\mu}\xi\partial_{\nu}\alpha\partial^{\nu}\alpha\right]%
+b\partial_{\mu}\chi\partial^{\mu}\xi \Bigr\} \Bigr\}\cr & -4\cos\left(\frac{%
b\chi}{2}\right)\left\{ \sin\left(\frac{b\chi}{2}\right)\Bigl\{ %
\partial_{\mu}\partial^{\mu}\xi\left[1+\frac{b^2\lambda}{4}%
\partial_{\nu}\chi\partial^{\nu}\chi\right] -\frac{b^2\lambda}{4}\left[%
\partial_{\mu}\partial^{\mu}\chi\partial_{\nu}\chi\partial^{\nu}\xi+%
\partial_{\mu}\chi\partial^{\mu}\left(\partial_{\nu}\chi\partial^{\nu}\xi%
\right)-\partial_{\mu}\xi\partial^{\mu}\left(\partial_{\nu}\chi\partial^{%
\nu}\chi\right)\right] \Bigr\}\right. \cr &+\cos\left(\frac{b\chi}{2}\right)%
\Bigl\{ \frac{b^2\lambda}{2}\sin(b\chi)\partial_{\mu}\partial^{\mu}\xi%
\partial_{\nu}\alpha\partial^{\nu}\xi -\frac{b^2\lambda}{2}\sin(b\chi)\left[%
\partial_{\mu}\partial^{\mu}\alpha\partial_{\nu}\xi\partial^{\nu}\xi+%
\partial_{\mu}\alpha\partial^{\mu}\left(\partial_{\nu}\xi\partial^{\nu}\xi%
\right)-\partial_{\mu}\xi\partial^{\mu}\left(\partial_{\nu}\alpha\partial^{%
\nu}\xi\right)\right]\cr & +b^3\lambda\cos(b\chi)\left[\partial_{\mu}\chi%
\partial^{\mu}\xi\partial_{\nu}\alpha\partial^{\nu}\xi-\partial_{\mu}\chi%
\partial^{\mu}\alpha\partial_{\nu}\xi\partial^{\nu}\xi\right]%
-b\partial_{\mu}\chi\partial^{\mu}\alpha \Bigr\} \Bigr\}=0.
\label{ThirdSkyrme}
\end{align}
} 

\subsection{Energy density}

The energy-momentum tensor takes the form {\small 
\begin{align}
T_{\mu\nu}=& \frac{K}{2}\|c\|^2\left\{ 8\left[\partial_{\mu}\alpha\partial_{%
\nu}\alpha\sin^2\left(\frac{b\chi}{2}\right)+\partial_{\mu}\xi\partial_{\nu}%
\xi\cos^2\left(\frac{b\chi}{2}\right)\right]+2\partial_{\mu}\chi\partial_{%
\nu}\chi\right.\cr & \left.-g_{\mu\nu}4\left[\partial_{\rho}\alpha\partial^{%
\rho}\alpha\sin^2\left(\frac{b\chi}{2}\right)+\partial_{\rho}\xi\partial^{%
\rho}\xi\cos^2\left(\frac{b\chi}{2}\right)\right]-g_{\mu\nu}\partial_{\rho}%
\chi\partial^{\rho}\chi\right\}\cr & +\frac{K}{2}\|c\|^2(2b^2\lambda)\Big\{%
\left[\partial_{\mu}\xi\partial_{\nu}\xi\partial_{\rho}\alpha\partial^{\rho}%
\alpha+\partial_{\mu}\alpha\partial_{\nu}\alpha\partial_{\rho}\xi\partial^{%
\rho}\xi-\left(\partial_{\mu}\alpha\partial_{\nu}\xi+\partial_{\mu}\xi%
\partial_{\nu} \alpha\right)\partial_{\rho}\alpha\partial^{\rho}\xi\right]%
\sin^2(b\chi)\cr  & -\left[\partial_{\mu}\alpha\partial_{\nu}\alpha\sin^2%
\left(\frac{b\chi}{2}\right)+\partial_{\mu}\xi\partial_{\nu}\xi\cos^2\left(%
\frac{b\chi}{2}\right)\right]\partial_{\rho}\chi\partial^{\rho}\chi -\left[%
\partial_{\rho}\alpha\partial^{\rho}\alpha\sin^2\left(\frac{b\chi}{2}%
\right)+\partial_{\rho}\xi\partial^{\rho}\xi\cos^2\left(\frac{b\chi}{2}%
\right)\right]\partial_{\mu}\chi\partial_{\nu}\chi\cr & \left.-\sin^2\left(%
\frac{b\chi}{2}\right)\left(\partial_{\mu}\alpha\partial_{\nu}\chi+%
\partial_{\mu}\chi\partial_{\nu}\alpha\right)\partial_{\rho}\alpha\partial^{%
\rho}\chi-\cos^2\left(\frac{b\chi}{2}\right)\left(\partial_{\mu}\xi%
\partial_{\nu}\chi+\partial_{\mu}
\chi\partial_{\nu}\xi\right)\partial_{\rho}\xi\partial^{\rho}\chi\right\}\cr 
& -\frac{K}{2}\|c\|^2b^2\lambda g_{\mu\nu}\left\{ \left[\partial_{\rho}\xi%
\partial^{\rho}\xi\partial_{\sigma}\alpha\partial^{\sigma}\alpha-\left(%
\partial_{\rho}\alpha\partial^{\rho}\xi\right)^2\right]\sin^2(b\chi) \right. %
\cr & +\left[\partial_{\rho}\alpha\partial^{\rho}\alpha\sin^2\left(\frac{%
b\chi}{2}\right)+\partial_{\rho}\xi\partial^{\rho}\xi\cos^2\left(\frac{b\chi%
}{2}\right)\right]\partial_{\sigma}\chi\partial^{\sigma}\chi\cr & \left.-
\left[\left(\partial_{\rho}\alpha\partial^{\rho}\chi\right)^2\sin^2\left(%
\frac{b\chi}{2}\right)+\left(\partial_{\rho}\xi\partial^{\rho}\chi\right)^2%
\cos^2\left(\frac{b\chi}{2}\right)\right]\partial_{\sigma}\chi\partial^{%
\sigma}\chi\right\}.
\end{align}
} From this, we can obtain the energy density as $\rho_E=T_{tt}$. 

\subsection{Baryon charge}

The Baryon charge is 
\begin{align}
B=\frac{1}{24\pi^2}\int\rho_B drd\gamma d\phi,
\end{align}
with 
\begin{align}
\rho_B=-12\|c\|^2\varepsilon^{ijk}\partial_i\alpha\partial_j\xi\partial_k%
\cos(b\chi).
\end{align}
Up to now we have just written local expressions, but in order to compute
the Baryonic charge it is necessary to define the ranges of $\alpha$, $\xi$
and $\chi$. Proposition \ref{prop-dynkin} tells us that the period of $%
e^{g\kappa}$ is $T_{\kappa}=\eta n\frac{2\pi}{b}$, where $\eta=1,2$
depending on the representation, while $n\in\mathbb{Z}$. Following \cite%
{AlCaCaCe}, the ranges must be 
\begin{align}
0\leq\alpha\leq \eta\sigma n\frac{2\pi}{b},\qquad 0\leq\chi\leq\frac{\pi}{b}
\qquad\mbox{and}\qquad 0\leq\xi\leq \eta m\frac{2\pi}{b},
\end{align}
where $\sigma=1$ for odd-dimensional representations and $\frac{1}{2}$ for
even-dimensional representations and $m,n$ are both integer. The integration
of the density charge leads to 
\begin{align}
B=4mn\sigma\eta^2\frac{\|c\|^2}{b^2}.
\end{align}
We can compute the ratio $\frac{\|c\|^2}{b^2}$ in the following way. From (%
\ref{Normf}), we get 
\begin{align}
-2\|c\|^2=\mbox{Tr}(f^2)=\sum_{j=1}^{r}p_j\mbox{Tr}(h_jf),
\end{align}
where the definition $f=\sum_{j=1}^{r}p_jh_j$ has been used. Now, we can
replace the coefficients $p_j$ with (\ref{Coeff_f}) and $\mbox{Tr}%
(h_jf)=\alpha_{j}(f)=ib$ to get 
\begin{align}
\|c\|^2=\sum_{j=1}^{r}p_j\mbox{Tr}(h_jf) = \sum_{j=1}^r\frac {b^2}{%
\|\alpha_j\|^2} \sum_{k=1}^r (C^G)^{-1}_{jk}.
\end{align}
The Baryon charge takes the form 
\begin{align}
B=4mn\sigma\eta^2\sum_{j,k=1}^r\frac {1}{\|\alpha_j\|^2} (C^G)^{-1}_{jk}.
\end{align}

\subsection{Example: the Lasagna case}

Let us now compare the results obtained in this section with the previous
ones. Our quantities can be written in terms of the Lasagna ansatz as
follows 
\begin{align}
\alpha=-\frac{\sigma t}{2L_{\phi}}+\sigma\frac{\phi}{2}+m\frac{\gamma}{2}%
\qquad\mbox{and}\qquad\xi=\frac{\sigma t}{2L_{\phi}}-\sigma\frac{\phi}{2}+m%
\frac{\gamma}{2}.  \label{lasagnansatz}
\end{align}
Moreover, the \textit{profile function} depends only on the parameter $r$ ($%
\chi=\chi(r)$). This leads to the following relations 
\begin{gather*}
\partial_{\mu}\chi\partial^{\mu}\alpha=\partial_{\mu}\chi\partial^{\mu}%
\xi=0,\qquad\partial_{\mu}\partial^{\mu}\alpha=\partial_{\mu}\partial^{\mu}%
\xi=0 \\
\partial_{\mu}\alpha\partial^{\mu}\alpha=\partial_{\mu}\xi\partial^{\mu}\xi=%
\partial_{\mu}\alpha\partial^{\mu}\xi=\frac{1}{4L_{\gamma}^2}.  \notag
\end{gather*}
With these choices the equations (\ref{SecondSkyrme}) and (\ref{ThirdSkyrme}%
) are automatically satisfied. The equation (\ref{FirstSkyrme}) becomes 
\begin{align}
\chi^{\prime \prime }(r) \left(1+\frac{b^2\lambda}{4L_{\gamma}^2}\right)=0,
\end{align}
which leads to the solution 
\begin{align}
\chi(r)=\frac{r}{2b},
\end{align}
where the boundary conditions $\chi(0)=0$ and $\chi(2\pi)=\frac{\pi}{b}$
have been used. Now, it is easy to compute the energy density, which results 
\begin{align}
\rho_E=\frac{K}{2}\|c\|^2\Bigl\{ \frac{2}{L_{\phi}^2}+\frac{1}{L_{\gamma}^2}+%
\frac{1}{4b^2L_{r}^2}\left(1+\frac{b^2\lambda}{4L_{\gamma}^2}\right)+\frac{%
b^2\lambda}{8L_{\phi}^2L_{\gamma}^2}\left[4\sin^2\left(\frac{r}{2}\right)-1%
\right]\Bigr\}.
\end{align}
The integration over the volume of the box gives the total energy of the
Lasagna 
\begin{align}
E=4L_{\phi}L_rL_{\gamma}\pi^3K\frac{\|c\|^2}{b^2}\Bigl\{ \frac{2}{L_{\phi}^2}%
+\frac{1}{L_{\gamma}^2}+\frac{1}{4b^2L_{r}^2}\left(1+\frac{b^2\lambda}{%
4L_{\gamma}^2}\right)+\frac{b^2\lambda}{8L_{\phi}^2L_{\gamma}^2}\Bigr\}.
\end{align}

\ 

\ 


\section{Coupling with $U(1)$ gauge field}

By employing the generalization presented in the previous section, it is now
easy to couple the Skyrmion field to an electromagnetic field $A_{\mu}$. To
this aim we introduce the action 
\begin{align}  \label{ActionMaxwell}
\mathcal{A}=\int d^4x\sqrt{-g}\ \mathrm{Tr} \left[\frac{K}{2}\left(\hat{%
\mathcal{L}}_{\mu}\hat{\mathcal{L}}^{\mu}+\frac{\lambda}{8}\hat{G}_{\mu\nu}%
\hat{G}^{\mu\nu}\right)-\frac{1}{4}F_{\mu\nu}F^{\mu\nu}\right],
\end{align}
where 
\begin{align}
F_{\mu\nu}=\partial_{\mu}A_{\nu}-\partial_{\nu}A_{\mu}
\end{align}
and the \textit{hat} stands for the replacement of the partial derivative
with a covariant derivative 
\begin{align}
D_{\mu}=\partial_{\mu}-A_{\mu}\left[T,\cdot\right],
\end{align}
which means that 
\begin{align}
\hat{\mathcal{L}}_{\mu}=U^{-1}D_{\mu}U=U^{-1}\left(\partial_{\mu}U-A_{\mu}%
\left[T,U\right]\right)\qquad\mbox{and}\qquad\hat{G}_{\mu\nu}=\left[\hat{%
\mathcal{L}}_{\mu},\hat{\mathcal{L}}_{\nu}\right].  \label{connessione}
\end{align}
Here $T$ is any element of the Lie algebra of the group $G$, representing
the direction of the $U(1)$ gauge field. Later, we will identify $T$ with
the generator $T_3$. The action (\ref{ActionMaxwell}) is now invariant under
gauge transformation 
\begin{align}
U\rightarrow e^{-\beta T}Ue^{\beta T},\qquad A_{\mu}\rightarrow
A_{\mu}+\partial_{\mu}\beta.
\end{align}
The gauge invariance appears also in the fact that the theory depends on $%
A_{\mu}$ through the quantity $\partial_{\mu}\alpha-A_{\mu}$, which is
invariant for gauge transformations.


\subsection{Covariant Baryonic charge}

As in \cite{CaWi}, in order to determine a topological invariant, one is
tempted to start directly generalizing (\ref{B}) to the expression 
\begin{align*}
\hat B=\frac{1}{24\pi^2}\int_{\mathcal{V}} \mathrm{Tr}[\hat{\mathcal{L}}%
\wedge\hat{\mathcal{L}}\wedge\hat{\mathcal{L}}],
\end{align*}
which, however, is not a topological invariant if the field-strength $F$ is
non-vanishing. Nevertheless, a topological invariant can be constructed
after a simple subtraction, even for a non-Abelian gauge field. Indeed, we
have:

\begin{prop}
\label{prop4} Let $\mathcal{S}$ be a three dimensional closed compact
manifold, 
\begin{align}
U:\mathcal{S }\longrightarrow G
\end{align}
a differentiable map from $S$ to the Lie group $G$, $\hat{\mathcal{L}}%
_{\mu}= U^{-1}D_{\mu}U$, $\hat{\mathcal{R}}_{\mu}= D_{\mu}U U^{-1}$, with a
non necessarily Abelian connection $\omega$, and $\Omega$ the curvature of $%
\omega$, 
\begin{align}
D_\mu U&= \partial_\mu U +[\omega, U], \\
\Omega&=d\omega +\frac 12 [\omega,\omega].
\end{align}
Hence 
\begin{align}
\hat B&=\frac{1}{24\pi^2}\int_{\mathcal{S}} \mathrm{Tr}\left[ \hat{\mathcal{L%
}}\wedge\hat{\mathcal{L}}\wedge\hat{\mathcal{L}} -3\hat{\mathcal{L}}\wedge
\Omega -3 \hat{\mathcal{R}} \wedge \Omega \right],
\end{align}
is a topological invariant. Moreover, if $H_2(\mathcal{S})=0$ and $A$ is
Abelian, then $\hat B=B$.
\end{prop}

\begin{proof}
In order to prove the proposition, we have to prove that the first variation
of $\hat B$ w.r.t. $U$ and $\omega$ (independently) vanishes at any
functional point, that is independently if $U$ and $\omega$ are constrained
by some equations of motion. Notice that in taking variations, $\delta \omega
$ is a well defined 1-form on $\mathcal{S}$ despite $\omega$ could not be.
To keep notation compact we will use bold round brackets to indicate a trace 
$\pmb( M \pmb)\equiv \mathrm{Tr (M)}$. Moreover, we first recall the
following properties. If $a_j$, $j=1,\ldots,k$ are Lie algebra valued 1-forms then 
\begin{align}
\pmb( a_1 \wedge \cdots \wedge a_{k-1} \wedge a_k \pmb)=(-1)^{k-1} \pmb( a_k
\wedge a_1\wedge \cdots \wedge a_{k-1} \pmb). \qquad\ \mbox{(cyclicity (c))}
\end{align}
If $[,]$ indicates the Lie product (commutator) of matrix valued forms and $%
a,b,c$ are three differential forms of degree $k_a$ $k_b$ and $k_c$
respectively, then 
\begin{align}
[a,b\wedge c]= [a,b]\wedge c +(-1)^{k_ak_b} b\wedge [a,c]. \qquad\ %
\mbox{(graded algebraic derivative (gad))}
\end{align}
Since 
\begin{align}
\pmb( [a,b] \pmb)=0,
\end{align}
in particular, we have that, if $a$ is a 1-form, then 
\begin{align}
\pmb( [a,b]\wedge c \pmb)=(-1)^{k_b} \pmb( b\wedge [a,c]\pmb), \qquad\ %
\mbox{(algebraic integration by parts (aip))}
\end{align}
and 
\begin{align}
\pmb( Db \pmb)=\pmb( db \pmb). \qquad\ \mbox{(algebraic trivialization (at))}%
.
\end{align}
Using these properties, taking a variation $\delta U$ of $U$ we can write 
\begin{align}
\frac 13 \delta_U \pmb( \hat{\mathcal{L}}\wedge \hat{\mathcal{L}}\wedge \hat{\mathcal{L}}\pmb)=& \pmb (
\delta_U \hat{\mathcal{L}}\wedge \hat{\mathcal{L}}\wedge \hat{\mathcal{L}}\pmb)=\pmb(-U^{-1} \delta U \hat{\mathcal{L}}\wedge \hat{\mathcal{L}}\wedge \hat{\mathcal{L}}+ U^{-1} D\delta U \wedge \hat{\mathcal{L}}\wedge \hat{\mathcal{L}}%
\pmb)\cr = & \pmb(-U^{-1} \delta U \hat{\mathcal{L}}\wedge \hat{\mathcal{L}}\wedge \hat{\mathcal{L}}+D(
U^{-1} \delta U) \wedge \hat{\mathcal{L}}\wedge \hat{\mathcal{L}}-D( U^{-1}) \delta U \wedge \hat{
\mathcal{L}} \wedge \hat{\mathcal{L}}\pmb)\cr =& \pmb( D( U^{-1} \delta U) \wedge \hat{\mathcal{L}}\wedge
\hat{\mathcal{L}}\pmb)
\end{align}
where we used $D( U^{-1})=-U^{-1} DU U^{-1}$ and cyclicity in the last term
of the second line. Hence, 
\begin{align}
\frac 13 \delta_U \pmb( \hat{\mathcal{L}}\wedge \hat{\mathcal{L}}\wedge \hat{\mathcal{L}}\pmb)=& \pmb(
D[U^{-1} \delta g \hat{\mathcal{L}}\wedge \hat{\mathcal{L}}] \pmb)-\pmb( U^{-1} \delta U [U^{-1}
D(DU) \wedge \hat{\mathcal{L}}-\hat{\mathcal{L}}\wedge U^{-1} D(DU)] \pmb)\cr = & d \pmb( U^{-1}
\delta g \hat{\mathcal{L}}\wedge \hat{\mathcal{L}}\pmb) -\pmb( U^{-1} \delta U [U^{-1}
[\Omega,U] \wedge \hat{\mathcal{L}}-\hat{\mathcal{L}}\wedge U^{-1} [\Omega,U] \pmb),
\label{primaparte}
\end{align}
where we used (at) in the first term and $D(DU)=[\Omega,U]$ in the other
ones. Now, let us consider 
\begin{align}
\delta_U \pmb(\Omega\wedge \hat{\mathcal{L}}\pmb)=& -\pmb( \Omega \wedge U^{-1} \delta
U \hat{\mathcal{L}}\pmb) +\pmb( \Omega \wedge U^{-1} D \delta U \pmb)\cr = & -\pmb(
\delta U U^{-1} (D\Omega) U^{-1} \pmb) +\pmb(D [\Omega U^{-1} \delta U] \pmb%
) +\pmb( \Omega \wedge \hat{\mathcal{L}}U^{-1} \delta U \pmb)\cr =& -\pmb( \delta U
\hat{\mathcal{L}}\wedge \Omega U^{-1} \pmb) +d\pmb(\Omega U^{-1} \delta U \pmb) +\pmb(
\Omega \wedge \hat{\mathcal{L}}U^{-1} \delta U \pmb),  \label{secondaparte}
\end{align}
where again we used (c) and (at). In the same way 
\begin{align}
\delta_U \pmb(\Omega\wedge \hat{\mathcal{R}}\pmb)=& -\pmb( \delta U g^{-1} \Omega
\wedge \hat{\mathcal{R}} \pmb) +d\pmb(\delta U U^{-1} \Omega \pmb) +\pmb(\delta U
U^{-1} \hat{\mathcal{R}} \wedge \Omega \pmb).  \label{terzaparte}
\end{align}
Subtracting (\ref{secondaparte}) and (\ref{terzaparte}) to (\ref{primaparte}%
), and multiplying times 3, we get 
\begin{align}
\delta_U \pmb( \hat{\mathcal{L}}\wedge \hat{\mathcal{L}}\wedge \hat{\mathcal{L}}\pmb -3 \hat{\mathcal{L}}\wedge
\Omega -3 \hat{\mathcal{R}} \wedge \Omega)=3d \pmb( U^{-1} \delta U \hat{\mathcal{L}}\wedge \hat{\mathcal{L}}-\Omega U^{-1} \delta_U -\Omega \delta U U^{-1} \pmb) .
\end{align}
Since the r.h.s. is the differential of a globally well defined 2-form
and $\mathcal{S}$ is a smooth closed compact manifold, it follows from
Stokes theorem that the first variation of $\hat B$ under variation of $U$ vanishes. 
\newline
As a second step, let us consider a variation $\delta \omega$ of $\omega$.
The strategy is the same as above. For 
\begin{align}
\mu \equiv \frac 13 \pmb( \hat{\mathcal{L}}\wedge \hat{\mathcal{L}}\wedge \hat{\mathcal{L}}\pmb -3 \hat{\mathcal{L}}
\wedge \Omega -3 \hat{\mathcal{R}} \wedge \Omega)
\end{align}
we get 
\begin{align}
\delta_\omega \mu=&\pmb( U^{-1} [\delta \omega,U] \wedge \hat{\mathcal{L}}\wedge \hat{\mathcal{L}}%
\pmb)-\pmb( U^{-1} [\delta \omega, U] \wedge \Omega \pmb) -\pmb( \hat{\mathcal{L}}
\wedge (d\delta\omega+[\omega,\delta\omega]) \pmb)-\pmb(
[\delta\omega,U]U^{-1}\wedge \Omega \pmb) \cr & -\pmb( \hat{\mathcal{R}} \wedge
(d\delta\omega+[\omega,\delta\omega]) \pmb)\cr =&\pmb( U^{-1} [\delta
\omega,U] \wedge \hat{\mathcal{L}}\wedge \hat{\mathcal{L}}\pmb)-\pmb( U^{-1} [\delta \omega, U]
\wedge \Omega \pmb) +d\pmb(\hat{\mathcal{L}}\wedge \delta \omega \pmb)-\pmb( D\hat{\mathcal{L}}
\wedge \delta \omega \pmb) -\pmb( [\delta\omega,U]U^{-1}\wedge \Omega \pmb) %
\cr & +d\pmb(\hat{\mathcal{R}} \wedge \delta \omega \pmb) -\pmb( D\hat{\mathcal{R}} \wedge
\delta\omega \pmb).
\end{align}
Now, 
\begin{align}
D\hat{\mathcal{L}}=d(U^{-1} DU)+[\omega,U^{-1} DU]=-\hat{\mathcal{L}}\wedge \hat{\mathcal{L}}+U^{-1}
(d[\omega,U] +[\omega,[\omega,U]]-\omega \wedge dU)= -\hat{\mathcal{L}}\wedge \hat
{\mathcal{L}}+U^{-1} [\Omega,U],
\end{align}
and similarly 
\begin{align}
D\hat{\mathcal{R}}=\hat{\mathcal{R}}\wedge \hat{\mathcal{R}}+[\Omega,U]U^{-1}.
\end{align}
Finally, noticing that 
\begin{align}
\pmb( U^{-1} [\delta \omega,U] \wedge \hat{\mathcal{L}}\wedge \hat{\mathcal{L}}\pmb)=\pmb( \delta
\omega \wedge \hat{\mathcal{R}} \wedge \hat{\mathcal{R}}\pmb)-\pmb(\delta \omega \wedge \hat{\mathcal{L}}
\wedge \hat{\mathcal{L}}\pmb)
\end{align}
and putting all together, we get 
\begin{align}
\delta_\omega \mu=& d \pmb( (\hat{\mathcal{L}}+\hat{\mathcal{R}} \wedge \delta \omega)\pmb),
\end{align}
which, as above, it proves invariance also under variations of the
connection. Thus, $\hat B$ is topological invariant.\newline
Now, we have to prove the second part. To this end it is convenient to
introduce some further notation. After fixing a basis $\{T_a\}_a$ of $Lie(G)$%
, with structure constants $f^a_{\ bc}$ defined by\footnote{%
We use the Einstein's convention on sums.} 
\begin{align}
[T_b,T_c]=f^a_{\ bc} T_a,
\end{align}
it is convenient to define 
\begin{align}
\tilde T_a:=&U^{-1} T_a U, \\
\tau_a:=& T_a-\tilde T_a, \\
\check\tau_a:=& T_a+\tilde T_a,
\end{align}
so that, writing $\omega=\omega^a \tau_a$, we have 
\begin{align}
\hat{\mathcal{L}}=\mathcal{L}-\omega^a \tau_a
\end{align}
and also 
\begin{align}
\pmb( \hat{\mathcal{L}}\wedge \Omega \pmb) +\pmb( \hat{\mathcal{R}} \wedge \Omega \pmb)=\Omega^a
\wedge \pmb( \hat{\mathcal{L}}\check \tau_a\pmb).
\end{align}
Therefore, 
\begin{align}
\pmb( \hat{\mathcal{L}}\wedge \hat{\mathcal{L}}\wedge \hat{\mathcal{L}}\pmb)-3\pmb( (\hat{\mathcal{L}}+\hat{\mathcal{R}})\wedge
\Omega\pmb) =&\pmb(\mathcal{L}\wedge \mathcal{L}\wedge \mathcal{L }\pmb%
)-3\omega^a \wedge \pmb( \tau_a \mathcal{L }\wedge \mathcal{L }\pmb)
+3\omega^a \wedge \omega^b \wedge \pmb( \tau_a \tau_b \mathcal{L }\pmb)\cr & 
-\omega^a \wedge \omega^b \wedge \omega^c \pmb(\tau_a \tau_b \tau_c \pmb%
)-3\Omega^a \wedge \pmb(\mathcal{L }\check \tau_a\pmb)+3 \Omega^a \wedge
\omega^b \pmb( \tau_b \check \tau_a \pmb).  \label{formulaccia}
\end{align}
By the Maurer-Cartan equation $d\mathcal{L}=-\frac 12 [\mathcal{L},\mathcal{L%
}]$, and we can write, 
\begin{align}
-\omega^a \wedge \pmb( \tau_a \mathcal{L }\wedge \mathcal{L }\pmb)&-\Omega^a
\wedge \pmb(\mathcal{L }\check \tau_a\pmb)=\omega^a \wedge \pmb( \tau_a d%
\mathcal{L }\pmb)-d\omega^a \wedge \pmb(\mathcal{L }\check \tau_a\pmb)
-\frac 12 f^a_{\ bc} \omega^b \wedge \omega^c \wedge \pmb( \mathcal{L }%
\check \tau_a \pmb)\cr & =\omega^a \wedge \pmb( \tau_a d\mathcal{L }\pmb%
)-d[\omega^a \wedge \pmb(\mathcal{L }\check \tau_a\pmb)]-\omega^a \wedge d%
\pmb(\mathcal{L }\check \tau_a\pmb) -\frac 12 f^a_{\ bc} \omega^b \wedge
\omega^c \wedge \pmb( \mathcal{L }\check \tau_a \pmb),
\end{align}
which, using 
\begin{align}
\omega^a \wedge \pmb( \tau_a d\mathcal{L }\pmb) - \omega^a \wedge d\pmb(%
\mathcal{L }\check \tau_a\pmb)=&-2\omega^a \wedge \pmb( \tilde T_a d\mathcal{%
L }\pmb) -\omega^a \wedge \pmb(\mathcal{L}\wedge d\tilde T_a\pmb)\cr = & 
-2\omega^a \wedge \pmb( \tilde T_a d\mathcal{L }\pmb) -\omega^a \wedge \pmb(%
\mathcal{L}\wedge (-[\mathcal{L}, \tilde T_a]\pmb)=0,
\end{align}
because of the Maurer-Cartan equations, becomes 
\begin{align}
-\omega^a \wedge \pmb( \tau_a \mathcal{L }\wedge \mathcal{L }\pmb)-\Omega^a
\wedge \pmb(\mathcal{L }\check \tau_a\pmb)= -d[\omega^a \wedge \pmb(\mathcal{%
L }\check \tau_a\pmb)] -\frac 12 f^a_{\ bc} \omega^b \wedge \omega^c \wedge %
\pmb( \mathcal{L }\check \tau_a \pmb).  \label{parte1}
\end{align}
Next, we rewrite 
\begin{align}
\omega^a \wedge \omega^b \wedge \pmb( \tau_a \tau_b \mathcal{L }\pmb%
)=&\omega^a \wedge \omega^b \wedge \pmb( T_a T_b \mathcal{L }\pmb)+\omega^a
\wedge \omega^b \wedge \pmb( \tilde T_a \tilde T_b \mathcal{L }\pmb)
-\omega^a \wedge \omega^b \wedge \pmb( T_a \tilde T_b \mathcal{L }\pmb%
)-\omega^a \wedge \omega^b \wedge \pmb( \tilde T_a T_b \mathcal{L }\pmb)\cr =
& \frac 12 \omega^a \wedge \omega^b \wedge \pmb( [T_a, T_b] \mathcal{L }\pmb%
)+\frac 12 \omega^a \wedge \omega^b \wedge \pmb( [\tilde T_a, \tilde T_b] 
\mathcal{L }\pmb)-\omega^a \wedge \omega^b \wedge \pmb( [T_a, \tilde T_b] 
\mathcal{L }\pmb)\cr =&\frac 12 \omega^a \wedge \omega^b \wedge \pmb( f^c_{\
ab} T_c \mathcal{L }\pmb)+\frac 12 \omega^a \wedge \omega^b \wedge \pmb(
f^c_{\ ab} \tilde T_c \mathcal{L }\pmb)-\omega^a \wedge \omega^b \wedge \pmb%
( T_a [\tilde T_b, \mathcal{L}] \pmb)\cr = & \frac 12 f^c_{\ ab} \omega^a
\wedge \omega^b \wedge \pmb( \check \tau_c \mathcal{L }\pmb)-\omega^a \wedge
\omega^b \wedge \pmb( T_a d\tilde T_b \pmb)\cr =&\frac 12 f^c_{\ ab}
\omega^a \wedge \omega^b \wedge \pmb( \check \tau_c \mathcal{L }\pmb%
)-d[\omega^a \wedge \omega^b \wedge \pmb( T_a \tilde T_b \pmb)] + d[\omega^a
\wedge \omega^b] \wedge \pmb( T_a \tilde T_b \pmb)\cr = & \frac 12 f^c_{\
ab} \omega^a \wedge \omega^b \wedge \pmb( \check \tau_c \mathcal{L }\pmb%
)-d[\omega^a \wedge \omega^b \wedge \pmb( T_a \tilde T_b \pmb)] + d\omega^a
\wedge \omega^b \wedge \pmb( T_a \tilde T_b-\tilde T_a T_b \pmb).
\label{parte2}
\end{align}
Since $\pmb( \tilde T_a \tilde T_b \pmb)=\pmb( T_a T_b \pmb)$, we also have 
\begin{align}
\Omega^a \wedge \omega^b \pmb(\tau_b \check \tau_a\pmb)=-\Omega^a \wedge
\omega^b \pmb(T_a \tilde T_b-\tilde T_a T_b\pmb).  \label{parte3}
\end{align}
Finally, using also $\pmb( \tilde T_a \tilde T_b \tilde T_c \pmb)=\pmb( T_a
T_b T_c\pmb)$, 
\begin{align}
\omega^a \wedge \omega^b \wedge \omega^c \pmb(\tau_a \tau_b \tau_c \pmb)=&
3\omega^a \wedge \omega^b \wedge \omega^c \pmb(\tilde T_a \tilde T_b T_c
-T_a T_b \tilde T_c \pmb)\cr = & \frac 32\omega^a \wedge \omega^b \wedge
\omega^c \pmb([\tilde T_a, \tilde T_b] T_c -[T_a, T_b] \tilde T_c \pmb) \cr %
=& \frac 32 \omega^a \wedge \omega^b \wedge \omega^c f^d_{\ ab} \pmb(\tilde
T_d T_c -T_d \tilde T_c \pmb).  \label{parte4}
\end{align}
After replacing (\ref{parte1}), (\ref{parte2}), (\ref{parte3}) and (\ref%
{parte4}) in (\ref{formulaccia}) we get 
\begin{align}
\pmb( \hat{\mathcal{L}}\wedge \hat{\mathcal{L}}\wedge \hat{\mathcal{L}}\pmb)-3\pmb( (\hat{\mathcal{L}}+\hat{\mathcal{R}})\wedge
\Omega\pmb) =\pmb(\mathcal{L}\wedge \mathcal{L}\wedge \mathcal{L }\pmb)
-3d[\omega^a \wedge \pmb(\mathcal{L }\check \tau_a\pmb)+\omega^a \wedge
\omega^b \wedge \pmb( T_a \tilde T_b \pmb)].  \label{finaldensity}
\end{align}
In particular, if the connection is Abelian, 
\begin{align}
\pmb( \hat{\mathcal{L}}\wedge \hat{\mathcal{L}}\wedge \hat{\mathcal{L}}\pmb)-3\pmb( (\hat{\mathcal{L}}+\hat{\mathcal{R}})\wedge
\Omega\pmb) =\pmb(\mathcal{L}\wedge \mathcal{L}\wedge \mathcal{L }\pmb)
-3d[\omega^a \wedge \pmb(\mathcal{L }\check \tau_a\pmb)].
\end{align}
In this case $\Omega=d\omega$ and, if $H_2(\mathbb{S})=0$ so that $H^2(%
\mathbb{S})=0$, then $\Omega$ is exact and $\omega$ is well defined
everywhere on $\mathcal{S}$. Therefore, $3d[\omega^a \wedge \pmb(\mathcal{L }%
\check \tau_a\pmb)]$ is an exact form and Stokes theorem ensures that under
these hypotheses $\hat B=B$.
\end{proof}

\ 

Notice that with our conventions in (\ref{connessione}), we have to make the
identifications 
\begin{align}
\omega=-A, \qquad \Omega=-F, \qquad F=dA-\frac 12[A,A]\ .
\end{align}
From (\ref{finaldensity}) we then see that 
\begin{align}
\hat B=\frac 1{24\pi^2} \int_{\mathcal{S}} \hat \rho_B\ dr d\gamma d\phi ,
\end{align}
with 
\begin{align}
\hat \rho_B=\rho_B+3 \varepsilon^{ijk} \partial_i [A_j^a \mathrm{Tr}(%
\mathcal{L}_k (T_a+U^{-1} T_a U))-A^a_j A^b_k \mathrm{Tr} (T_a U^{-1} T_b
U)],  \label{BaryonicDensityGen}
\end{align}
where, according to the conventions in \cite{AlCaCaCe}, the orientation of
the coordinates is such that $\varepsilon^{r\gamma\phi}=1$. In our case $%
\mathcal{S}$ is a closed three dimensional manifold in a semisimple compact
Lie group $G$. Since in this case $H_2(G,\mathbb{Q})=0$, we get that the
correction to the density does not contribute to the integral and we expect $%
\hat B=B$ always. \newline
It is worth to mention that in the construction of the solutions of the
Skyrme equations, however, $\mathcal{S}$ is replaced by $\mathcal{V}$ that
is compact but it is not a closed smooth manifold but a hyperrectangle with
boundary. Therefore, the above integral does not define a topological
invariant unless we impose suitable boundary conditions. To understand which
are the most suitable ones, let us first analyze the case $A=0$. In this
case the map $U$ maps the hyperrectangle in a closed smooth submanifold of $G
$, so $\hat{\mathcal{L}}$ is the pull-back of a 1-form well defined on a
closed compact manifold (indeed, the left-invariant Maurer-Cartan form) and
this is the reason we get a topological invariant. This suggest the boundary
conditions we are looking for. They have to be imposed so that also $A$ is
the pull-back of a well defined 1-form over $G$ (or the image of the
hyperrectangle in $G$).\newline
Under these conditions, the quantities $A_{\mu}$ are not independent, due to
the fact that $\Phi$, $\Theta$ and $\chi$ defines a map $M:\mathbb{R}%
^{3+1}\mapsto\mathbb{R}^3$. Locally, the embedding takes the form $%
A=A_{\Phi}d\Phi+A_{\Theta}d\Theta+A_{\chi}d\chi$, equivalent to 
\begin{align}
A_{\mu}=A_{\Phi}\partial_{\mu}\Phi+A_{\Theta}\partial_{\mu}\Theta+A_{\chi}%
\partial_{\mu}\chi.  \label{DefCons}
\end{align}


\subsection{Example: Lasagna states coupled to an electromagnetic field}

\label{Sec:Boundaries} To be explicit, we now work out the example of
Lasagna states. For this case we choose $T=\kappa$. The covariant derivative
determines the coupling of the gauge field to the Skyrmions, which appears
in the definition of $\hat{\mathcal{L}}_{\mu}$ 
\begin{align}
\hat{\mathcal{L}}_{\mu}=e^{-\alpha \kappa}e^{-\xi \kappa}\left[%
\left(\partial_{\mu}\alpha-A_{\mu}\right)(\kappa-\hat{\kappa}%
)+\partial_{\mu}\xi(\kappa+\hat{\kappa})+\partial_{\mu}\chi f\right]e^{\xi
\kappa}e^{\alpha \kappa}.
\end{align}
Notice that the introduction of the gauge field in the direction $\kappa$
causes a shift in $\mathcal{L}_{\mu}$ given by $\partial_{\mu}\alpha%
\rightarrow\partial_{\mu}\alpha-A_{\mu}$. It results that all the quantities
we computed in the previous section are shifted by this quantity when the
Skyrmions are coupled to a Maxwell field and it is really easy to convert
the uncoupled theory with the coupled one. The covariant Baryon density
charge now becomes 
\begin{align}
\hat{\rho}_B&=\rho_B +3 \varepsilon^{ijk} \partial_i [A_j \mathrm{Tr}(%
\mathcal{L}_k (\kappa+U^{-1} \kappa U))] \cr & =\rho_B +3 \varepsilon^{ijk}
\partial_i \mathrm{Tr}\left[ \left(\partial_{k}\alpha (\kappa-\hat{\kappa}%
)+\partial_{k}\xi(\kappa+\hat{\kappa})+\partial_{k}\chi f\right)
(\kappa+\hat \kappa)\right].
\end{align}
Using that 
\begin{align}
\mathrm{Tr} [(\kappa-\hat \kappa)(\kappa+\hat \kappa)]=\mathrm{Tr}
[f(\kappa+\hat \kappa)]=0 ,
\end{align}
and 
\begin{align}
\mathrm{Tr} [(\kappa-\hat \kappa)(\kappa+\hat \kappa)]=2(1+\cos (b\chi)) 
\mathrm{Tr} \kappa^2=-4\|c\|^2 (1+\cos (b\chi)) ,
\end{align}
\begin{align}
\hat{\rho}_B&=\rho_B-12\|c\|^2\varepsilon^{ijk}\partial_i\left[%
A_j\partial_k\xi (1+\cos(b\chi))\right] ,
\end{align}
where $\rho_B$ is the uncoupled density. The correction to $\rho_B$ is a
total derivative, so it depends only on the boundary conditions, as
discussed above. Differently from \cite{CaWi}, our system lives in a box,
so, the electromagnetic field is not constrained to zero at the boundaries.
Therefore, the Baryonic charge is not necessarily a topological invariant
and not even expected to be an integer. As we said above, we can fix this
problem by requiring for $A$ to be the pull-back of a well defined potential
over the homology cycle of $G$ selected by the map $U$. This is easily
accomplished by looking at the form of the ansatz for the Lasagna states. As 
$t$ is irrelevant, we fix $t=0$ to simplify the expressions: 
\begin{align}
U(r,\gamma,\phi)=e^{-\phi\sigma \kappa} e^{\chi(r)f} e^{m\gamma
\kappa}=e^{-\phi\sigma \kappa} e^{m\gamma \hat \kappa(r)} e^{\chi(r)f}.
\end{align}
Since $b\chi(2\pi)=\pi$, we see that $\hat \kappa(0)=-\hat
\kappa(2\pi)=\kappa$, so that, if, for a generic fixed $r$, $%
U(r,\gamma,\phi) $ defines a two dimensional surface in $G$, for $r=0,2\pi$
it collapses down to one dimensional circles: 
\begin{align}
U(0,\gamma,\phi)&=e^{(m\gamma-\sigma\phi)\kappa}=e^{\xi\kappa}, \\
U(2\pi,\gamma,\phi)&=e^{-(m\gamma+\sigma\phi)\kappa}
e^{\chi(2\pi)f}=e^{-\alpha\kappa} e^{\chi(2\pi)f}.
\end{align}
This degeneration means that well defined 1-forms on the whole manifold must
have components only along the direction on the degeneration submanifolds,
which in our case means 
\begin{align}
A_\alpha (r=0)&=0, \\
A_\xi (r=2\pi)&=0,
\end{align}
which in the original coordinates becomes 
\begin{align}
\frac 1{2m} A_\gamma (r=0) +\frac 1{2\sigma} A_\phi(r=0)=&0,  \label{bound-0}
\\
\frac 1{2m} A_\gamma (r=2\pi) -\frac 1{2\sigma} A_\phi(r=2\pi)=&0.
\label{bound-2pi}
\end{align}
Also, one between $\phi$ and $\gamma$ has to be identified periodically,
while the other one is periodic or ``antiperiodic''\footnote{%
The antiperiodicity is not exact and in general some matrix components are
periodic and other are antiperiodic. However, what happens is that points
are identified in the image, in such a way to respect orientation, so the
corresponding differential forms are periodically identified.} according to
the cases if the cycle is of $SO(3)$ or $SU(2)$, respectively. Therefore, in
any case, the 1-forms in the image of the embedding have to be periodically
identified so that the integrals at the ``boundaries'' $\phi=0$ and $%
\phi=2\pi$ cancel out and the same happens for the boundaries at $\gamma=0$
and $\gamma=2\pi$. So, the only boundaries that may contribute are the ones
at $r=0$ and $r=2\pi$, which we collectively call $\partial^r B$. Therefore,
the Baryonic charge results 
\begin{align}
\hat{B}=&B-\frac{\|c\|^2}{2\pi^2}\int_{\partial B}(1+\cos(b\chi))A\wedge
d\xi=B+\frac{\|c\|^2}{2\pi^2}\int_{\partial^r B}(1+\cos(b\chi)) (\sigma
A_\gamma +mA_\phi) d\gamma \wedge d\phi \cr = & B-\frac{\|c\|^2}{\pi^2}
\int_{[0,2\pi]\times [0,2\pi]} (\sigma A_\gamma (0,\gamma,\phi) +mA_\phi
(0,\gamma,\phi)) d\gamma d\phi=B
\end{align}
because of the above boundary conditions, and we used that $A_\alpha=\sigma
A_\gamma +mA_\phi$.


\subsection{Decoupling of Skyrme equations and free-force conditions}

To the Skyrmion equation coupled to a Maxwell field, obtained by shifting $%
\partial_{\mu}\alpha\rightarrow\partial_{\mu}\alpha-A_{\mu}$ in (\ref%
{FirstSkyrme}), (\ref{SecondSkyrme}) and (\ref{ThirdSkyrme}), we have to add
the Maxwell equations, which are given by 
\begin{align}
\nabla_{\nu}F^{\nu\mu} - \mbox{Tr}\left\{\frac{K}{2}D\left(\hat{\mathcal{R}}
^{\mu}+\frac{\lambda}{4}\left[\hat{\mathcal{R}}_{\nu},\hat{G}^{\mu\nu}\right]%
\right)\right\}=0.
\end{align}
In the generic Euler parameterization, they become 
\begin{align}
&\partial_{\nu}\partial^{\nu}A_{\mu}-\partial_{\mu}\left(\partial_{\nu}%
\partial^{\nu}\alpha\right)-\frac{K}{2}\|c\|^2\Bigr\{ \left(\partial_{\mu}%
\alpha-A_{\mu}\right)\Bigr[8\sin^2\left(\frac{a\chi}{2}\right)\left(1+\frac{%
a^2\lambda}{4}\partial_{\nu}\chi\partial^{\nu}\chi\right) \\
&\qquad\qquad\qquad\qquad\qquad\qquad+2a^2\lambda\sin^2(a\chi)\partial_{\nu}%
\xi\partial^{\nu}\xi\Bigr]-2a^2\lambda\left[\partial_{\mu}\xi(\partial_{\nu}%
\alpha-A_{\nu})\partial^{\nu}\xi+\partial_{\mu}\chi(\partial_{\nu}\alpha-A_{%
\nu})\partial^{\nu}\chi\right]\Bigr\}=0.  \notag
\end{align}
To look for explicit solutions, we aim to decouple the Skyrme equations from
the Maxwell field. Since $A_{\mu}$ appears in the products $%
\left(\partial_{\mu}\alpha-A_{\mu}\right)\left(\partial^{\mu}\alpha-A^{\mu}%
\right)$, $\left(\partial_{\mu}\alpha-A_{\mu}\right)\partial^{\mu}\xi$ and $%
\left(\partial_{\mu}\alpha-A_{\mu}\right)\partial^{\mu}\chi$ and in the
derivative $\partial_{\mu}\left(\partial^{\mu}\alpha-A^{\mu}\right)$, we can
separate the Skyrme equations from the rest by looking for solutions where
these terms are a priori fixed functions 
\begin{gather}  \label{ProdCondGen}
\left(\partial_{\mu}\alpha-A_{\mu}\right)\left(\partial^{\mu}\alpha-A^{\mu}%
\right)=f(t,r,\theta,\phi),\qquad\partial_{\mu}\left(\partial^{\mu}%
\alpha-A^{\mu}\right)=g(t,r,\theta,\phi), \\
\left(\partial_{\mu}\alpha-A_{\mu}\right)\partial^{\mu}\xi=p(t,r,\theta,%
\phi),\qquad\left(\partial_{\mu}\alpha-A_{\mu}\right)\partial^{\mu}%
\chi=q(t,r,\theta,\phi).  \notag
\end{gather}
Recall that, the quantity $\partial_{\mu}\alpha-A_{\mu}$ is gauge invariant.


\subsubsection{Free-force conditions}

Due to gauge invariance, we can introduce the new gauge field $\tilde{A}%
_{\mu}=A_{\mu}-\partial_{\mu}\alpha$. Imposing the conditions 
\begin{align}
f(t,r,\theta,\phi)=g(t,r,\theta,\phi)=p(t,r,\theta,\phi)=q(t,r,\theta,%
\phi)=0\qquad\mbox{and}\qquad \tilde{A}^{\nu}\partial_{\nu}\tilde{A}_{\mu}=0,
\end{align}
the so called \textit{free-force conditions} are satisfied \cite{crystal2},
namely 
\begin{align}
\tilde{F}_{\mu\nu}J^{\nu}=0,\qquad \tilde{J}^{\nu}=\nabla_{\rho}\tilde{F}%
^{\rho\nu},
\end{align}
where $\tilde{F}_{\mu\nu}$ is the field-strength of $\tilde{A}_{\mu}$. The
wave equations become 
\begin{align}  \label{FirstSkyrmeFF}
&\partial_{\mu}\partial^{\mu}\chi \left[1+b^2\lambda\partial_{\nu}\xi%
\partial^{\nu}\xi\cos^2\left(\frac{b\chi}{2}\right)\right]%
+b\sin(b\chi)\left(1-\frac{b^2\lambda}{4}\partial_{\mu}\chi\partial^{\mu}%
\chi\right)\partial_{\nu}\xi\partial^{\nu}\xi  \notag \\
&\qquad-b^2\lambda\left\{\cos^2\left(\frac{b\chi}{2}\right)\partial_{\mu}%
\partial^{\mu}\xi\partial_{\nu}\xi\partial^{\nu}\chi+\cos^2\left(\frac{b\chi%
}{2}\right)\left[\partial_{\mu}\xi\partial^{\mu}\left(\partial_{\nu}\xi%
\partial^{\nu}\chi\right)-\partial_{\mu}\chi\partial^{\mu}\left(\partial_{%
\nu}\xi\partial^{\nu}\xi\right)\right]\right\}  \notag \\
&\qquad\qquad\qquad\qquad+\frac{b^3\lambda}{4}\sin(b\chi)\left(\partial_{%
\mu}\xi\partial^{\mu}\chi\right)^2=0,
\end{align}
\begin{align}  \label{SecondSkyrmeFF}
&4b\cos\left(\frac{b\chi}{2}\right)\partial_{\mu}\chi\partial^{\mu}\xi-4
\sin\left(\frac{b\chi}{2}\right)\Bigl\{ \partial_{\mu}\partial^{\mu}\xi\left[%
1+\frac{b^2\lambda}{4}\partial_{\nu}\chi\partial^{\nu}\chi\right]  \notag \\
&\qquad\qquad\qquad\qquad-\frac{b^2\lambda}{4}\left[\partial_{\mu}\partial^{%
\mu}\chi\partial_{\nu}\chi\partial^{\nu}\xi+\partial_{\mu}\chi\partial^{\mu}%
\left(\partial_{\nu}\chi\partial^{\nu}\xi\right)-\partial_{\mu}\xi\partial^{%
\mu}\left(\partial_{\nu}\chi\partial^{\nu}\chi\right)\right] \Bigr\}=0,
\end{align}
\begin{align}  \label{ThirdSkyrmeFF}
&4b\sin\left(\frac{b\chi}{2}\right)\partial_{\mu}\chi\partial^{\mu}\xi-4\cos%
\left(\frac{b\chi}{2}\right) \Bigl\{ \partial_{\mu}\partial^{\mu}\xi\left[1+%
\frac{b^2\lambda}{4}\partial_{\nu}\chi\partial^{\nu}\chi\right]  \notag \\
&\qquad\qquad\qquad\qquad-\frac{b^2\lambda}{4}\left[\partial_{\mu}\partial^{%
\mu}\chi\partial_{\nu}\chi\partial^{\nu}\xi+\partial_{\mu}\chi\partial^{\mu}%
\left(\partial_{\nu}\chi\partial^{\nu}\xi\right)-\partial_{\mu}\xi\partial^{%
\mu}\left(\partial_{\nu}\chi\partial^{\nu}\chi\right)\right] \Bigr\}=0.
\end{align}
The last two equations imply that 
\begin{align}
\partial_{\mu}\chi\partial^{\mu}\xi=0,
\end{align}
so 
\begin{align}
\partial_{\mu}\partial^{\mu}\xi\left[1+\frac{b^2\lambda}{4}%
\partial_{\nu}\chi\partial^{\nu}\chi\right]=0
\end{align}
gives $\partial_{\mu}\partial^{\mu}\xi=0$. With these solutions, the Eq. (%
\ref{FirstSkyrmeFF}) takes the simpler form 
\begin{align}  \label{FirstSkyrmeFFSol}
&\partial_{\mu}\partial^{\mu}\chi \left[1+b^2\lambda\partial_{\nu}\xi%
\partial^{\nu}\xi\cos^2\left(\frac{b\chi}{2}\right)\right]%
+b\sin(b\chi)\left(1-\frac{b^2\lambda}{4}\partial_{\mu}\chi\partial^{\mu}%
\chi\right)\partial_{\nu}\xi\partial^{\nu}\xi=0.
\end{align}
We can also apply them to the Maxwell equations 
\begin{align}\label{EulerMaxwellDecoupled}
&\partial_{\nu}\partial^{\nu}A_{\mu}-\partial_{\mu}\left(\partial_{\nu}%
\partial^{\nu}\alpha\right)+\frac{K}{2}\|c\|^2
\left(\partial_{\mu}\alpha-A_{\mu}\right)\Bigr[8\sin^2\left(\frac{b\chi}{2}%
\right)\left(1+\frac{b^2\lambda}{4}\partial_{\nu}\chi\partial^{\nu}\chi%
\right) \\
&\qquad\qquad\qquad\qquad\qquad\qquad+2b^2\lambda\sin^2(b\chi)\partial_{\nu}%
\xi\partial^{\nu}\xi\Bigr]=0.  \notag
\end{align}
In particular, the energy density is 
\begin{align}\label{FFEnergyDensity}
\rho_E&=4K\|c\|^2\Bigr\{ \left[\tilde{A}_t^2\sin^2\left(\frac{b\chi}{2}%
\right)+(\partial_{t}\xi)^2\cos^2\left(\frac{b\chi}{2}\right)\right]\left(1-%
\frac{b^2\lambda}{4}\partial_{\rho}\chi\partial^{\rho}\chi\right)+\frac{%
b^2\lambda}{4}\tilde{A}_t^2\partial_{\rho}\xi\partial^{\rho}\xi\sin^2(b\chi)
\\
&+\left[1-b^2\lambda\partial_{\rho}\xi\partial^{\rho}\xi\cos^2\left(\frac{%
b\chi}{2}\right)\right]\frac{(\partial_{t}\chi)^2}{4}-\frac{g_{tt}}{2}%
\partial_{\rho}\xi\partial^{\rho}\xi\cos^2\left(\frac{b\chi}{2}%
\right)\left(1+\frac{b^2\lambda}{4}\partial_{\sigma}\chi\partial^{\sigma}%
\chi\right)-\frac{g_{tt}}{8}\partial_{\rho}\chi\partial^{\rho}\chi \Bigl\} 
\notag
\end{align}


\subsection{Example: Lasagna, again}

We can use the results of this section in order to study the behavior of
Lasagna when coupled to the $U(1)$ gauge field. To simplify the results, we
use the free-force conditions 
\begin{gather}  \label{FFLasagne}
\left(\partial_{\mu}\alpha-A_{\mu}\right)\left(\partial^{\mu}\alpha-A^{\mu}%
\right)=0,\qquad\partial_{\mu}\left(\partial^{\mu}\alpha-A^{\mu}\right)=0,%
\cr \left(\partial_{\mu}\alpha-A_{\mu}\right)\partial^{\mu}\xi=0,\qquad%
\left(\partial_{\mu}\alpha-A_{\mu}\right)\partial^{\mu}\chi=0,
\end{gather}
and for the gauge field we make the ansatz 
\begin{align}
A_{\mu}=\left(A_{t}(r),0,A_{\gamma}(r),A_{\phi}(r)\right).  \label{ansatzAmu}
\end{align}
If we simply shift the gauge field, we can write 
\begin{gather}  \label{FFLasagneShifted}
\tilde{A}_{\mu}\tilde{A}^{\mu}=0,\qquad\partial_{\mu}\tilde{A}%
^{\mu}=0,\qquad \tilde{A}_{\mu}\partial^{\mu}\xi=0,\qquad \tilde{A}%
_{\mu}\partial^{\mu}\chi=0.
\end{gather}
These conditions are easily solved by using (\ref{lasagnansatz}) (which also
apply to $\tilde{A}_{\mu}$) together with (\ref{ansatzAmu}). This leads to
the solution $\tilde{A}_t=-\frac{\tilde{A}_{\phi}}{L_\phi}$ and $\tilde{A}%
_{\gamma}=0$ ($A_\gamma=\frac{m}{2}$); so, only one gauge field results to
be independent, for instance we can take $\tilde A_{\phi}$. Thus, the wave
equations and the Maxwell equation become 
\begin{align}  \label{MaxwellFFSolLasagne}
\frac{\chi^{\prime \prime }}{L_r^2} \left[1+\frac{b^2\lambda}{4L_{\gamma}^2}%
\cos^2\left(\frac{b\chi}{2}\right)\right]+\frac{b}{4L_{\gamma}^2}%
\sin(b\chi)\left(1-\frac{b^2\lambda}{4L_r^2}\chi^{\prime 2}\right)=0, \\
\frac{\tilde{A}_{\phi}^{\prime \prime }}{L_r^2}+\frac{K}{2}\|c\|^2 \tilde{A}%
_{\phi}\Bigr[8\sin^2\left(\frac{b\chi}{2}\right)\left(1+\frac{b^2\lambda}{%
4L_r^2}\chi^{\prime 2}\right)+\frac{b^2\lambda}{2L_{\gamma}^2}\sin^2(b\chi)%
\Bigr]=0.
\end{align}
The first equation can be rewritten as 
\begin{align}
\frac{d}{dr}\left\{ \frac{\chi^{\prime 2}}{2L_r^2}\left[1+\frac{b^2\lambda}{%
4L_{\gamma}^2}\cos^2\left(\frac{b\chi}{2}\right)\right]-\frac{1}{%
2L_{\gamma}^2}\cos^2\left(\frac{b\chi}{2}\right) \right\}=0,
\end{align}
so that 
\begin{align}
\chi^{\prime 2}(r)=L_r^2\frac{M+\frac{1}{2L_{\gamma}^2}\cos^2\left(\frac{%
b\chi}{2}\right)}{1+\frac{b^2\lambda}{4L_{\gamma}^2}\cos^2\left(\frac{b\chi}{%
2}\right)},
\end{align}
where $M$ is an integration constant. This determines the boundary values of 
$\chi^{\prime }$ from the ones of $\chi$ 
\begin{align}
\chi^{\prime 2}(0)&=L_r^2\frac{M+\frac{1}{2L_{\gamma}^2}}{1+\frac{b^2\lambda%
}{4L_{\gamma}^2}},  \label{ChiPQ} \\
\chi^{\prime 2}(2\pi)&=L_r^2M.
\end{align}
Vice versa, we can write $M$ in terms of $\chi^{\prime }(0)$ 
\begin{align}
M=\frac{\chi^{\prime 2}(0)}{L_r^2}\left(1+\frac{b^2\lambda}{4L_{\gamma}^2}%
\right)-\frac{1}{2L_{\gamma}^2},
\end{align}
so that 
\begin{align}
\chi^{\prime 2}(2\pi)=\chi^{\prime 2}(0)\left(1+\frac{b^2\lambda}{%
4L_{\gamma}^2}\right)-\frac{L_r^2}{2L_{\gamma}^2}.
\end{align}
We can put this result into the Maxwell equations, getting 
\begin{align}
\frac{\tilde{A}_{\phi}^{\prime \prime }}{L_r^2}+\frac{K}{2}\|c\|^2 \tilde{A}%
_{\phi}\sin^2\left(\frac{b\chi}{2}\right)V_M(\chi)=0,
\label{MaxwellFFLasagneVM}
\end{align}
where 
\begin{align}
V_M(\chi)=8\left[1+\frac{b^2\lambda}{4L_{\gamma}^2}\cos^2\left(\frac{b\chi}{2%
}\right)\right]+2b^2\lambda\frac{M+\frac{1}{2L_{\gamma}^2}\cos^2\left(\frac{%
b\chi}{2}\right)}{1+\frac{b^2\lambda}{4L_{\gamma}^2}\cos^2\left(\frac{b\chi}{%
2}\right)}.
\end{align}
From (\ref{ChiPQ}) we can locally write $r$ in terms of $\chi$ as 
\begin{align}
r(\chi)=\pm\frac{1}{L_r}\int_{0}^{\chi}\sqrt{\frac{1+\frac{b^2\lambda}{%
4L_{\gamma}^2}\cos^2\left(\frac{b\hat{\chi}}{2}\right)}{M+\frac{1}{%
2L_{\gamma}^2}\cos^2\left(\frac{b\hat{\chi}}{2}\right)}}d\hat{\chi},
\end{align}
which leads to a definition of $\tilde{A}_{\phi}$ in terms of $\chi$, let us
call $B(\chi)=\tilde{A}_{\phi}(r(\chi))$. This way, equation (\ref%
{MaxwellFFLasagneVM}) can be entirely written in terms of $\chi$ 
\begin{align}  \label{MaxwellFFLasagneB}
B^{\prime \prime }(\chi)-\frac{b^2\lambda}{16L_{\gamma}^2}B^{\prime
}(\chi)\sin(b\chi)\frac{1-\hat{M}+D_1(\chi)}{D_1(\chi)D_{\hat{M}}(\chi)}%
+K\|c\|^2b^2\lambda B(\chi)\sin^2\left(\frac{b\chi}{2}\right)\frac{%
2D_1^2(\chi)+D_{\hat{M}}(\chi)}{D_{\hat{M}}(\chi)}=0,
\end{align}
where a prime indicates derivative with respect to $\chi$, and the following
quantities have been introduced 
\begin{align}
\hat{M}=M\frac{b^2\lambda}{2},\qquad D_a=a+\frac{b^2\lambda}{4L_{\gamma}^2}%
\cos^2\left(\frac{b\chi}{2}\right).
\end{align}
Replacing 
\begin{align}
B(\chi)=C(\chi)\exp\int_0^{\chi}\frac{b^2\lambda}{32L_{\gamma}^2}\sin(b\hat{%
\chi})\frac{1-\hat{M}+D_1(\hat{\chi})}{D_1(\hat{\chi})D_{\hat{M}}(\hat{\chi})%
}d\hat{\chi},
\end{align}
in (\ref{MaxwellFFLasagneB}), we get 
\begin{align}  \label{MaxwellFFLasagneHILL}
C^{\prime \prime }(\chi)+W_{\hat{M}}(\chi)C(\chi)=0,
\end{align}
with 
\begin{align}  \label{MaxwellFFLasagneHILLPotential}
W_{\hat{M}}(\chi)=&\frac{b^2\lambda}{32L_{\gamma}^2}\left[\cos(b\chi)\frac{1-%
\hat{M}+D_1(\hat{\chi})}{D_1(\hat{\chi})D_{\hat{M}}(\hat{\chi})} +\frac{%
b\lambda}{8L_{\gamma}^2}\sin^2(b\chi)\frac{1}{D_{\hat{M}}^2}\left(1+\frac{1-%
\hat{M}}{D_1}\right)\left(b+1+\frac{1-\hat{M}}{D_1}\right)\right]  \notag \\
&+K\|c\|^2b^2\lambda \sin^2\left(\frac{b\chi}{2}\right)\frac{2D_1^2(\chi)+D_{%
\hat{M}}(\chi)}{D_{\hat{M}}(\chi)}.
\end{align}
We can use 
\begin{align}
\frac{1}{D_a}=\frac{1}{a}\sum_{n=0}^{\infty}\left[-\frac{b^2\lambda}{%
4aL_{\gamma}^2}\cos^2\left(\frac{b\chi}{2}\right)\right]^n,
\end{align}
and 
\begin{align}
\cos^{2n}(x)=2^{-2n}\left\{ \frac{(2n)!}{(n!)^2}+2\sum_{s=0}^{n-1}{\binom{{2n%
}}{s}}\cos\left[2(n-s)x\right] \right\}.
\end{align}
to rewrite the potential (\ref{MaxwellFFLasagneHILLPotential}) in series of $%
\cos(kb\chi)$, where $k$ runs on all integers, and recognize (\ref%
{MaxwellFFLasagneHILL}) as a Hill equation. Notice that putting 
\begin{align}
x=b\chi/2,
\end{align}
we have that $x$ must vary in the interval $[0,\pi/2]$. Also, we have seen
that for the Lasagna states such interval must be one quarter of the period
over the cycle, which means that the solution we are looking for must be
periodic with period $2\pi$ as a function of $x$. Therefore, it is worth
mentioning the following result \cite{Magnus} (see also \cite{Schafke}):

\begin{prop}
Let $y_1(x)$ and $y_2(x)$ be the solutions of the Hill equation 
\begin{align}
y^{\prime \prime }+f(x)y=0,  \label{Hill1}
\end{align}
where $f(x)$ is an even function of the form 
\begin{align}
f(x)=2\sum_{n=1}^\infty f_n \cos (2n x), \qquad\ \sum_{n=1}^\infty
|f_n|<\infty,
\end{align}
with Cauchy conditions 
\begin{align}
y_1(0) =&1, \qquad\ y^{\prime }_1(0)=0, \\
y_2(0)=&0, \qquad\ y^{\prime }_2(0)=0.
\end{align}
Then, equation (\ref{Hill1}) has:

\begin{enumerate}
\item \label{a1} an even solution of period $\pi$ if and only if $y^{\prime
}_1(\pi/2)=0$;

\item \label{a2} an odd solution of period $\pi$ if and only if $%
y_2(\pi/2)=0 $;

\item \label{a3} an even solution of period $2\pi$ if and only if $%
y_1(\pi/2)=0$;

\item \label{a4} an odd solution of period $2\pi$ if and only if $y^{\prime
}_2(\pi/2)=0$.
\end{enumerate}
\end{prop}

Now, the boundary conditions for our Hill equation (\ref{MaxwellFFLasagneVM}%
) are 
\begin{align}
\tilde A_\phi (0)=&-m, \\
\tilde A_\phi (2\pi)=&0.
\end{align}
Therefore, if we call $y_{(j)}$, $j=1,\ldots,4$ the solutions corresponding
to the four points of the above proposition, if they exist, we get that the
solutions of interest for us have the general form 
\begin{align}
\tilde A_\phi(r)= B(\chi)= -m y_{(3)}(x) + \kappa y_{(2)}(x),
\end{align}
for $\kappa$ an arbitrary constant. The question on the existence of such
solutions is investigated in \cite{Magnus}. For example, the existence of
solution $y_{(2)}$ is granted if and only if $\omega$ takes values for which
the determinant of the infinite dimensional matrix 
\begin{align}
M_{ab}=\delta_{ab}+\frac {f_{a-b}-f_{a+b}}{\omega^2-a^2}, \qquad\
a,b=1,2,3,\ldots,
\end{align}
vanishes.\footnote{%
We use the notation $f_{a-b}=0$ if $a-b\leq 0$.} However, for any practical
purposes, such a way is impracticable for looking for explicit solutions in
this very general case. Therefore, in place of pursuing this very general
analysis, we move now to a particular but more tractable case.


\subsubsection{Linear solution of the Skyrme equations}

In the particular case when $M=\frac{2}{b^2\lambda}$, we can find a very
simple solution: 
\begin{align}
\chi^{\prime 2}(r)=\frac{2L_r^2}{b^2\lambda}\qquad\Rightarrow\qquad\chi(r)=%
\pm\sqrt{\frac{2L_r^2}{b^2\lambda}}r,
\end{align}
paying the price of fixing $\frac{L_r^2}{\lambda}$ (which corresponds to $%
\hat{M}=1$). Indeed, the boundary conditions on $\chi$ give $\chi^2(2\pi)=%
\frac{2L_r^2}{b^2\lambda}4\pi^2=\frac{\pi^2}{b^2}$. The Maxwell equation
takes the form 
\begin{align}  \label{Whittaker-Hill}
\frac{\tilde{A}_{\phi}^{\prime \prime }}{L_r^2}+K\|c\|^2 \tilde{A}_{\phi}%
\left[\left(3+\frac{b^2\lambda}{8L_{\gamma}^2}\right)- 3\cos\left(\frac{r}{2}%
\right)-\frac{b^2\lambda}{8L_{\gamma}^2}\cos r\right]=0,
\end{align}
which is a Whittaker-Hill equation \cite{Arscott67,Whittaker-Hill}, see also 
\cite{ArscottPeriodicEq}. It is convenient to introduce the variable change $%
\hat{r}=\frac{r}{4}$, so $\hat r$ has range $0\leq \hat r\leq\frac{\pi}{2}$.
This way, equation (\ref{Whittaker-Hill}) takes the canonical form 
\begin{align}
\tilde{B}_{\phi}^{\prime \prime }+16K L_r^2 \|c\|^2 \left[\left(3+\frac{%
b^2\lambda}{8L_{\gamma}^2}\right)- 3\cos\left(2\hat{r}\right)-\frac{%
b^2\lambda}{8L_{\gamma}^2}\cos(4\hat{r})\right] \tilde{B}_{\phi}=0.
\label{Whittaker-Hill_C}
\end{align}
We can therefore determine the solutions $y_{(2)}$ and $y_{(3)}$, following 
\cite{Whittaker-Hill}. Using the same notations of that paper, we can
identify 
\begin{align}
\omega=& 4 \sqrt {K\lambda} L_r \|c\| \frac{b}{L_{\gamma}}, \\
\eta=& 48 K L_r^2 \|c\|^2, \\
\rho=& -12 \sqrt {\frac K\lambda} L_r L_\gamma \frac {\|c\|}b.
\end{align}
In particular, $\eta=-\omega \rho$. For the function $\phi_{(3)}$ we have to
take (see \cite{Whittaker-Hill}) 
\begin{align}
\phi_{(3)}(x)=Re \left[ e^{-\frac i2 \omega \cos (2x)} \psi_{(3)} (x) \right]%
,
\end{align}
where 
\begin{align}
\psi_{(3)}(x)=\sum_{n=0}^\infty A_n B_n \cos ((2n+1)x),
\end{align}
with 
\begin{align}
A_0=1, \qquad\ A_n=\prod_{j=1}^n (\rho+2ji), \quad j>0,
\end{align}
and $B_n$ solves the recursion relations 
\begin{align}
(2-\eta) B_0-\frac 1\omega (\eta^2+4\omega^2) B_1=&0, \\
-\omega B_n+2[(2n+1)^2-\eta] B_{n+1} -\frac 1\omega [\eta^2+4(n+1)^2]
B_{n+2}=&0, \quad n\geq 1.
\end{align}
To find the periodic solution of period $2\pi$, $\omega$ and $\eta$ must be
constrained by the following trascendental equation, expressed in terms of a
continued fraction (see \cite{Whittaker-Hill}, formula (5.1)) 
\begin{align}
1-\frac 12\eta=\frac {\frac 14 (\eta^2+4\omega^2)}{9-\eta-} \frac {\frac 14
(\eta^2+16\omega^2)}{25-\eta-}\frac {\frac 14 (\eta^2+36 \omega^2)}{49-\eta-}%
\cdots,  \label{trasc1}
\end{align}
which then gives the solution 
\begin{align}
\frac {B_{n}}{B_{n-1}}=\frac {\frac 12 \omega}{(2n+1)^2-\eta} \frac {\frac
14 [\eta^2+4(n+1)^2\omega^2]}{(2n+3)^2-\eta-}\cdots \frac {\frac 14
[\eta^2+4(n+s-1)^2\omega^2]}{(2n+2s-1)^2-\eta-}\cdots, \quad n\geq 1.
\end{align}
Finally, $B_0$ is fixed by the condition $\phi_{(3)}(0)=1$. \newline
As what concerns the solution $\kappa \phi_{(2)}(x)$, we have to consider 
\begin{align}
\kappa \phi_{(2)}(x)=Re \left[ e^{-\frac i2 \omega \cos (2x)} \psi_{(2)} (x) %
\right],
\end{align}
where 
\begin{align}
\psi_{(2)}(x)=\sum_{n=1}^\infty C_n D_n \sin (2nx),
\end{align}
with 
\begin{align}
C_1=1, \qquad\ C_n=\prod_{j=1}^{n-1} (\rho+(2j+1)i), \quad j>1,
\end{align}
and $D_n$ solves the recursion relations 
\begin{align}
(4-\eta) D_1-\frac 1{2\omega} (\eta^2+9\omega^2) D_2=&0, \\
-\omega D_{n-1}+2[4n^2-\eta] D_{n} -\frac 1\omega [\eta^2+(2n+1)^2]
D_{n+1}=&0, \quad n> 2.
\end{align}
To find the periodic solution of period $\pi$, $\omega$ and $\eta$ must be
constrained by the following trascendental equation, expressed in terms of a
continued fraction (see \cite{Whittaker-Hill}, formula (5.3)) 
\begin{align}
4-\eta=\frac {\frac 14 (\eta^2+9\omega^2)}{16-\eta-} \frac {\frac 14
(\eta^2+25\omega^2)}{36-\eta-}\frac {\frac 14 (\eta^2+49 \omega^2)}{64-\eta-}%
\cdots,  \label{trasc2}
\end{align}
which then gives the solution 
\begin{align}
\frac {D_{n}}{D_{n-1}}=\frac {\frac 12 \omega}{4n^2-\eta} \frac {\frac 14
[\eta^2+(2n+1)^2\omega^2]}{4(n+1)^2-\eta-}\cdots \frac {\frac 14
[\eta^2+(2n+2s-1)^2\omega^2]}{4(n+s-1)^2-\eta-}\cdots, \quad n\geq 2.
\end{align}
Since $\kappa$ is arbitrary, no normalization is required for $D_1$.
However, notice that $\kappa \phi_{(2)}(x)$ can be considered only if
equation (\ref{trasc2}) has common solutions with (\ref{trasc1}).


\subsection{Example: Spaghetti states coupled to an electromagnetic field}

In the case of Spaghetti, we do not use a parameterization of Euler type but
the exponential parameterization of Section \ref{SpaghettiGroup}. Still, the
analysis can be easily extended to this case. Following \cite{FSpaghettiMax}%
, the gauge field is described by 
\begin{align}
A_{\mu}^aT_a=A_{\mu}T_3.  \label{AmuT3}
\end{align}
Also in this case, the free-force conditions decouples the Skyrme equations
from the Maxwell field. In particular, we take 
\begin{align}  \label{FFSpaghetti}
A_{\mu}A^{\mu}=0,\qquad\partial_{\mu}A^{\mu}=0\qquad
A_{\mu}\partial^{\mu}\Phi=0,\qquad A_{\mu}\partial^{\mu}\Theta=0.
\end{align}
A reasonable explicit form of a gauge field with these properties is given
by 
\begin{align}
A_{\mu}(t,r,\theta,\phi)=(A_t(r,\theta),0,0,-L_\phi A_t(r,\theta)).
\end{align}
From (\ref{AmuT3}) we see that the equations of motion for the Skyrme field
are 
\begin{align}
D^{\mu}\left(\hat{\mathcal{L}}_{\mu}+\frac{\lambda}{4}\hat{\mathcal{G}}%
_{\mu}\right)=0,
\end{align}
where $\hat {\mathcal{L}}_\mu$ is given by (\ref{connessione}) with $T=T_3$,
and 
\begin{align}
\hat{\mathcal{G}}_{\mu}=\left[\hat{\mathcal{L}}^{\nu},\left[\hat{\mathcal{L}}%
_{\mu},\hat{\mathcal{L}}_{\nu}\right]\right].
\end{align}
Conditions (\ref{FFSpaghetti}) lead to 
\begin{align}
\partial^{\mu}\left(\mathcal{L}_{\mu}+\frac{\lambda}{4}\mathcal{G}%
_{\mu}\right)=0,
\end{align}
which are the uncoupled Skyrme equations. Notice that from (\ref%
{SkyrmeEqsSpaghetti}) we get 
\begin{align}
\frac{d}{dr}\left[2\chi^{\prime 2}\left(\lambda q^2\sin^2\left(\frac{\chi}{2}%
\right)+L_\theta^2\right)+4q^2L_r^2\cos\chi\right]=0,
\end{align}
which means that 
\begin{align}
2\chi^{\prime 2}\left(\lambda q^2\sin^2\left(\frac{\chi}{2}%
\right)+L_\theta^2\right)+4q^2L_r^2\cos\chi=2Z,
\end{align}
where $Z$ is a constant, which is equivalent to 
\begin{align}
\chi^{\prime 2}(r)=\frac{Z-2q^2L_r^2+4q^2L_r^2\sin^2\left(\frac{\chi}{2}%
\right)}{L_\theta^2+\lambda q^2\sin^2\left(\frac{\chi}{2}\right)}.
\end{align}
The Maxwell equations for $A^{\phi}=-L_{\phi}A_t$ are 
\begin{align}\label{SpaghettiMaxwellDecoupled}
\frac{1}{L_r^2}\partial_r^2A^{\phi}+\frac{1}{L_\theta^2}\partial_\theta^2A^{%
\phi}+2KI_{G,\rho}\sin^2(q\theta)\sin^2\left(\frac{\chi}{2}%
\right)\left(A^{\phi}-\frac{p}{L_{\phi}^2}\right)\left[1+\frac{\chi^{\prime
2}}{L_r^2}+\frac{4q^2}{L_\theta^2}\sin^2\left(\frac{\chi}{2}\right) \right]%
=0.
\end{align}
This is a stationary Schr\"odinger equation with a double periodic potential
of finite type. In particular, here one is interested in the zero eigenvalue
case. Both the direct and inverse problem for this kind of equation is well
studied and much more involuted than the one dimensional case (already
highly non-trivial). Here we simply defer the reader to the literature (see 
\cite{Lubcke} and reference therein), and limit ourselves to discuss the
boundary conditions.\newline
As we discussed in the previous sections, the boundary conditions on $%
A_{\mu} $ are outlined by the behavior of the Skyrme field in the edges of
the box, namely (once again, we fix $t=0$) 
\begin{align}
U(0,\theta,\phi)&=1, \\
U(2\pi,\theta,\phi)&=e^{2\pi\tau_1}, \\
U(r,0,\phi)&=e^{\chi T_3}, \\
U(r,\pi,\phi)&=e^{\chi T_3}.
\end{align}
This requires the following constraints 
\begin{align}
A_{\phi}(0,\theta)=A_{\phi}(r,0)=A_{\phi}(r,\pi)=0.
\end{align}
The contribution of the gauge field to the Baryonic density can be always
computed from (\ref{BaryonicDensityGen}). This gives 
\begin{align}
\hat{\rho}_B=\rho_B+3\partial_{\theta}\left[A_{\phi}\mbox{Tr}\left(\mathcal{L%
}_r\left(T_3+U^{-1}T_3U\right)\right)\right]-3\partial_{r}\left[A_{\phi}%
\mbox{Tr}\left(\mathcal{L}_\theta\left(T_3+U^{-1}T_3U\right)\right)\right],
\end{align}
where $\mathcal{L}_r$ and $\mathcal{L}_\theta$ are specified in (\ref{MLr})
and (\ref{MLthet}). We easily find 
\begin{align}
U^{-1}T_3U&=T_3+\sin(q\theta)\tau_2-\sin(q\theta)\cos\chi\tau_2+\sin(q%
\theta)\sin\chi\tau_3.
\end{align}
This leads to 
\begin{align}
\mbox{Tr}(\mathcal{L}_\theta T_3)&=\mbox{Tr}(\mathcal{L}_\theta
U^{-1}T_3U)=I_{G,\rho}q\sin\chi\sin(q\theta), \\
\mbox{Tr}(\mathcal{L}_r T_3)&=\mbox{Tr}(\mathcal{L}_r
U^{-1}T_3U)=-I_{G,\rho}\chi^{\prime }\cos(q\theta).
\end{align}
Thus, we can check that the contribution to the Baryon charge becomes 
\begin{align}
\hat{B}=B+\frac{I_{G,\rho}}{4\pi}\int_0^{2\pi}\left(A_\phi(r,\pi)+A_%
\phi(r,0)\right)\chi^{\prime }dr=B,
\end{align}
according to the boundary conditions specified above, and the general
results following Proposition \ref{prop4}.




\section*{Acknowledgments}

{F. C. has been funded by Fondecyt Grant 1200022. M. L. is funded by
FONDECYT post-doctoral Grant 3190873. A. V. is funded by FONDECYT
post-doctoral Grant 3200884. The Centro de Estudios Cient\'{\i}ficos (CECs)
is funded by the Chilean Government through the Centers of Excellence Base
Financing Program of ANID.}

\newpage 
\begin{appendix}
\section{Roots of simple algebras}\label{app:roots}
Here we list the roots of all simple algebras.
\subsection{\boldmath{$A_{N}$}}
The corresponding complex algebra is $\mathfrak{sl}(N+1)$, while the compact form is $\mathfrak{su}(N+1)$.
If $e_a$, $a=1,\ldots,N+1$, is the canonical basis\footnote{In the literature, the canonical basis $e_a$ is also commonly denoted as $L_a$.} of $\mathbb R^{N+1}$, then the real linear space generated by the roots is isomorphic to a hyperplane in $\mathbb R^{N+1}$ in which all non-vanishing roots are represented by the vectors
$\alpha_{j,k}=e_j-e_k$, $j\neq k$. The simple roots are $\alpha_j=e_j-e_{j+1}$, $j=1,\ldots,N$. If $\lambda_j$ are the root matrices corresponding to the simple roots, then 
\begin{align}
 [\lambda_j,\lambda_k]=0 \quad \mbox{  if  }\quad  j-k\neq \pm1.
\end{align}
The split subalgebra is $\mathfrak {so}(N+1)$.

\subsection{\boldmath{$B_{N}$}}
The corresponding compact form is $\mathfrak{so}(2N+1)$.
The real linear space generated by the roots is isomorphic to $\mathbb R^{N}$. If $e_a$, $a=1,\ldots,N$, is the canonical basis of $\mathbb R^{N}$, then all non-vanishing roots are represented by the vectors
$e_j-e_k$, $j\neq k$, $\pm(e_j+e_k)$, $j< k$, $\pm e_j$. The simple roots are $\alpha_j=e_j-e_{j+1}$, $j=1,\ldots,N-1$ and $\alpha_N=e_N$. If $\lambda_j$ are the root matrices corresponding to the simple roots, then 
\begin{align}
 [\lambda_j,\lambda_k]=0 \quad \mbox{  if  }\quad  j-k\neq \pm1.
\end{align}
The split subalgebra is $\mathfrak {so}(N)\oplus\mathfrak {so}(N+1)$.

\subsection{\boldmath{$C_{N}$}}
The corresponding compact form is $\mathfrak{us_p}(2N)$, the compact symplectic algebra.
The real linear space generated by the roots is isomorphic to $\mathbb R^{N}$. If $e_a$, $a=1,\ldots,N$, is the canonical basis of $\mathbb R^{N}$, then all non-vanishing roots are represented by the vectors
$e_j-e_k$, $j\neq k$, $\pm(e_j+e_k)$, $j< k$, $\pm2 e_j$. The simple roots are $\alpha_j=e_j-e_{j+1}$, $j=1,\ldots,N-1$ and $\alpha_N=2e_N$. If $\lambda_j$ are the root matrices corresponding to the simple roots, then 
\begin{align}
 [\lambda_j,\lambda_k]=0 \quad \mbox{  if  }\quad  j-k\neq \pm1.
\end{align}
The split subalgebra is $\mathfrak {u}(N)$.

\subsection{\boldmath{$D_{N}$}}
The corresponding compact form is $\mathfrak{so}(2N)$.
The real linear space generated by the roots is isomorphic to $\mathbb R^{N}$. If $e_a$, $a=1,\ldots,N$, is the canonical
basis of $\mathbb R^{N}$, then all non-vanishing roots are represented by the vectors
$e_j-e_k$, $j\neq k$, $\pm(e_j+e_k)$, $j<k$.
 The simple roots are $\alpha_j=e_j-e_{j+1}$, $j=1,\ldots,N-1$ and $\alpha_N=e_{N-1}+e_N$. If $\lambda_j$ are the root matrices corresponding to the simple roots, the relevant non-vanishing commutators are
\begin{align}
 [\lambda_j,\lambda_{j+1}] \quad j=1,\ldots, N-2, \qquad [\lambda_{N-2},\lambda_N].
\end{align}
The split subalgebra is $\mathfrak {so}(N)\oplus\mathfrak {so}(N)$.

\subsection{\boldmath{$G_{2}$}}
The corresponding compact form is $\mathfrak g_2$.
The real linear space generated by the roots is isomorphic to a hyperplane in $\mathbb R^3$. If $e_a$, $a=1,\ldots,3$, is the canonical basis of $\mathbb R^{3}$, 
then all non-vanishing roots are represented by the vectors
$e_j-e_k$, $j\neq k$, and  $\pm(e_1+e_2+e_3-3e_s)$, s=1,2,3. The simple roots are $\alpha_1=e_2-e_3$ and $\alpha_2=e_1-2e_2+e_3$. If $\lambda_j$ are the root matrices corresponding to the simple roots, the relevant non-vanishing commutator is
$
 [\lambda_1,\lambda_2] .
$
The split subalgebra is $\mathfrak {so}(4)$.

\subsection{\boldmath{$F_{4}$}}
The corresponding compact form is $\mathfrak f_4$.
The real linear space generated by the roots is isomorphic to $\mathbb R^{4}$. If $e_a$, $a=1,\ldots,4$, is the canonical basis of $\mathbb R^{4}$, then all non-vanishing roots are represented by the vectors
$e_j-e_k$, $j\neq k$, $\pm(e_j+e_k)$, $j<k$, $\pm e_j$, $\frac 12 (\pm e_1\pm e_2\pm e_3\pm e_4)$. 
The simple roots are $\alpha_1=e_2-e_3$, $\alpha_1=e_3-e_4$, $\alpha_3=e_4$, and $\alpha_4=\frac 12 (e_1-e_2-e_3-e_4)$. If $\lambda_j$ are the root matrices corresponding to the simple roots, the relevant non-vanishing commutators are
\begin{align}
  [\lambda_j,\lambda_{j+1}] .
\end{align}
The split subalgebra is $\mathfrak {us_p}(6)\oplus\mathfrak {us_p}(2)$.

\subsection{\boldmath{$E_{6}$}}
The corresponding compact form is $\mathfrak e_6$.
The real linear space generated by the roots is isomorphic to $\mathbb R^{6}$. If $e_a$, $a=1,\ldots,6$, is the canonical basis of $\mathbb R^{6}$, then all non-vanishing roots are represented by the vectors
$\pm(e_j-e_k)$, $j< k<6$, $\pm(e_j+e_k)$, $j<k<6$, 
\begin{align*}
 \pm\frac 12 \{ \pm e_1\pm e_2\pm e_3 \pm e_4 \pm e_5 +\sqrt 6\ e_6\},
\end{align*}
where in the parenthesis only an even number of minus signs can appear. 
The simple roots are 
\begin{align*}
 \alpha_1=\frac 12 \{ e_1-e_2-e_3-e_4-e_5 +\sqrt 6\ e_6\}, \qquad \alpha_2=e_1+e_2, \qquad \alpha_k=e_{k-1}-e_{k-2},\quad k=3,\ldots,6.
\end{align*}
If $\lambda_j$ are the root matrices corresponding to the simple roots, the relevant non-vanishing commutators are
\begin{align}
  [\lambda_j,\lambda_{j+1}] ,\ j\neq 1, 2, \qquad [\lambda_1,\lambda_3],\quad  [\lambda_2,\lambda_4].
\end{align}
The split subalgebra is $\mathfrak {us_p}(8)$.
 
\subsection{\boldmath{$E_{7}$}}
The corresponding compact form is $\mathfrak e_7$.
The real linear space generated by the roots is isomorphic to $\mathbb R^{7}$. If $e_a$, $a=1,\ldots,7$, is the canonical basis of $\mathbb R^{7}$, then all non-vanishing roots are represented by the vectors
$\pm(e_j-e_k)$, $j< k<7$, $\pm(e_j+e_k)$, $j<k<7$, $\pm \sqrt 2\ e_7$,
\begin{align*}
 \pm\frac 12 \{ \pm e_1\pm e_2\pm e_3 \pm e_4 \pm e_5 \pm e_6 +\sqrt 2\ e_7\},
\end{align*}
where in the parenthesis only an odd number of minus signs can appear. 
The simple roots are 
\begin{align*}
 \alpha_1=\frac 12 \{ e_1-e_2-e_3-e_4-e_5-e_6 +\sqrt 2\ e_7\}, \qquad \alpha_2=e_1+e_2, \qquad \alpha_k=e_{k-1}-e_{k-2},\quad k=3,\ldots,7.
\end{align*}
If $\lambda_j$ are the root matrices corresponding to the simple roots, the relevant non-vanishing commutators are
\begin{align}
  [\lambda_j,\lambda_{j+1}] ,\ j\neq 1, 2, \qquad [\lambda_1,\lambda_3],\quad  [\lambda_2,\lambda_4].
\end{align}
The split subalgebra is $\mathfrak {su}(8)$.

\subsection{\boldmath{$E_{8}$}}
The corresponding compact form is $\mathfrak e_8$.
The real linear space generated by the roots is isomorphic to $\mathbb R^{8}$. If $e_a$, $a=1,\ldots,8$, is the canonical basis of $\mathbb R^{8}$, then all non-vanishing roots are represented by the vectors
$\pm(e_j-e_k)$, $j< k$, $\pm(e_j+e_k)$, $j<k$, 
\begin{align*}
 \frac 12 \{ \pm e_1\pm e_2\pm e_3 \pm e_4 \pm e_5 \pm e_6 \pm e_7 \pm\ e_8\},
\end{align*}
where in the parenthesis all signs can appear. 
The simple roots are 
\begin{align*}
 \alpha_1=\frac 12 \{ e_1-e_2-e_3-e_4-e_5-e_6 - e_7+e_8 \}, \qquad \alpha_2=e_1+e_2, \qquad \alpha_k=e_{k-1}-e_{k-2},\quad k=3,\ldots,8.
\end{align*}
If $\lambda_j$ are the root matrices corresponding to the simple roots, the relevant non-vanishing commutators are
\begin{align}
  [\lambda_j,\lambda_{j+1}] ,\ j\neq 1, 2, \qquad [\lambda_1,\lambda_3],\quad  [\lambda_2,\lambda_4].
\end{align}
The split subalgebra is $\mathfrak {so}(16)$.

\subsection{Resuming} In conclusion, we see that the commutators we need are strictly related to the Dynkin diagram of the algebra: a commutator between eigenmatrices of two simple roots is non zero only if the roots are linked, that is if the
scalar product is not zero. This is simply related to the fact that, with obvious notation, the commutators $[\lambda_{\alpha_i},\lambda_{\alpha_j}]$ or is an eigenmatrix for $\alpha_i+\alpha_j$, or it vanishes. We also recall here some very well known facts.
The fact that Dynkin diagrams have no loops allows to choose the normalization of the matrices $\lambda_j$ so that  
\begin{align}
 [\lambda_{\alpha_j},\lambda_{\alpha_k}]=
 {\rm sign} (k-j)(\delta_{jk} (\alpha_j|\alpha_j)-(\alpha_j|\alpha_k)) \lambda_{\alpha_j+\alpha_k} \label{lala}
\end{align}
t.i. $ [\lambda_{\alpha_j},\lambda_{\alpha_k}]=-(\alpha_j|\alpha_k) \lambda_{\alpha_j+\alpha_k}$ if $j<k$ and with the opposite sign if we change $j$ and $k$. Here $(|)$ is the scalar product in the space of roots.  Notice that also
\begin{align}
 [\tilde\lambda_{\alpha_j},\tilde\lambda_{\alpha_k}]=
 {\rm sign} (k-j)(\delta_{jk} (\alpha_j|\alpha_j)-(\alpha_j|\alpha_k)) \tilde\lambda_{\alpha_j+\alpha_k}. \label{tlatla}
\end{align}
Remember that the trace product is proportional to the Killing product and that the only non-Killing orthogonal root spaces are the ones corresponding to opposite roots.
This allows to fix a global normalization so that
\begin{align}
 {\rm Tr} (\tilde \lambda_j \lambda_k)=-\delta_{jk}. 
\end{align}
We also have, for the simple roots $\alpha_j$,
\begin{align}
 [\tilde\lambda_j,\lambda_k]=-i\delta_{jk} J_j, \label{tlala}
\end{align}
where $J_j$ are in a Cartan algebra. From the fact that the simple roots are linearly independent, it easily follows that the $J_j$, $j=1,\ldots,r$, are a basis for the Cartan subalgebra. This is also sufficient to fix the scalar product in the space of
roots so that
\begin{align}
 (\alpha_j|\alpha_k)=\alpha_j(J_k),
\end{align}
if the roots are defined as
\begin{align}
 [h,\lambda_j]=i\alpha_j(h) \lambda_j ,
\end{align}
for any $h\in H$.
In particular, using $ad$ invariance of the trace product we get
\begin{align*}
 {\rm Tr} (J_jJ_k)=i{\rm Tr}([\tilde\lambda_j,\lambda_j] J_k)=i{\rm Tr}(\tilde\lambda_j [\lambda_j,J_k])=\alpha_j(J_k) {\rm Tr} (\tilde \lambda_j \lambda_j),
\end{align*}
so that
\begin{align}
 {\rm Tr} (J_jJ_k)=-(\alpha_j|\alpha_k).
\end{align}
Finally, recall that any given simple Lie algebra is characterized by the $r\times r$ Cartan matrix 
\begin{align}
C^G_{jk}=2 \frac {(\alpha_j|\alpha_k)}{(\alpha_j|\alpha_j)}. 
\end{align}
With this we can rewrite the normalization conditions as
\begin{align}
 {\rm Tr} (\tilde \lambda_j \lambda_k)&=-\delta_{jk}, \\
 {\rm Tr} (J_jJ_k)&=-C^G_{jk} \frac {(\alpha_j|\alpha_j)}2.
\end{align}
The Cartan matrices of simple Lie groups can be found, for example, in \cite{Sl}, Table 6. 

\section{A proposition}

\begin{prop}\label{prop6}
 Let $\kappa=\sum_{j=1}^{r} (c_j \lambda_j+c_j^* \tilde \lambda_j)$, and $h\in H$ a matrix such that $\alpha_j(h)=\varepsilon_j a$, where $\varepsilon_j$ is a sign, $j=1,\ldots,r$, and set $x:=e^{-h z} \kappa e^{h z}$, $z\in \mathbb R$. Then
 \begin{align}
{\rm Tr} \kappa^2&=-2\|\underline c\|^2, \label{trakquadro}
\end{align}
\begin{align}
{\rm Tr} ([h,\kappa][h,\kappa])&={\rm Tr} ([h,x][h,x])=-2a^2 \| \underline c \|^2, \label{trhkhk}
\end{align}
and
\begin{eqnarray}
&&{\rm Tr} ([x,\kappa][x,\kappa])=-4\sin^2 (az)  \left(\sum_{j=1}^{r} \|\alpha_j\|^2|c_j|^4+\sum_{j<k} |c_j|^2 |c_{k}|^2 (\alpha_j|\alpha_k)[2\varepsilon_j \varepsilon_{k}+(\alpha_j|\alpha_k)(1-\varepsilon_j \varepsilon_{k})] \right). \label{trxkxk}
\end{eqnarray}   
\end{prop}
\begin{proof}
Let us start with
\begin{align}
 {\rm Tr}\kappa^2&=\sum_{j=1}^{r} \sum_{k=1}^{r} \{ c_jc_k{\rm Tr} (\lambda_j \lambda_k) +c^*_jc^*_k{\rm Tr} (\tilde\lambda_j \tilde\lambda_k) +c^*_jc_k{\rm Tr} (\tilde\lambda_j \lambda_k) +c_jc^*_k{\rm Tr} (\lambda_j \tilde\lambda_k) \}\cr
 &=\sum_{j=1}^{r} \sum_{k=1}^{r} \{ c^*_jc_k (-\delta_{jk}) +c_jc^*_k (-\delta_{jk}) \}=-2\sum_{j=1}^r |c_j|^2,
\end{align}  
where we used the normalization in the previous section. This proves (\ref{trakquadro}).\\
For (\ref{trhkhk}), we first note that
\begin{align}
 [h,x]=[h,e^{-hz}\kappa e^{hz}]=e^{-hz} [h,\kappa]e^{hz}.
\end{align}
Thus,
\begin{align}
 {\rm Tr} ([h,x][h,x])&={\rm Tr} (e^{-hz}[h,\kappa][h,\kappa]e^{hz}) ={\rm Tr} ([h,\kappa][h,\kappa]).
\end{align}
Now, we can use
\begin{align}
 [h,\kappa]&=\sum_{j=1}^{r} (c_j[h,\lambda_j]+c^*_j [h,\tilde\lambda_j])=i\sum_{j=1}^{r}\alpha_j(h) (c_j \lambda_j-c^*_j \tilde \lambda_j).
\end{align}
to get
 \begin{align}
 {\rm Tr} ([h,\kappa][h,\kappa])&=-{\rm Tr} \left( \sum_{j=1}^{r} (i\alpha_j(h)\lambda_j c_j-i \alpha_j(h) c^*_j \tilde\lambda_j)\sum_{k=1}^{r} (i\alpha_k(h)\lambda_k c_k-i \alpha_k(h) c^*_k \tilde\lambda_k)  \right)\cr
 &=-2\sum_{j=1}^{r} \alpha_j (h)^2 c_jc^*_j =-2a^2 \|\underline c\|^2,
\end{align}
where we used that $ \alpha_j (h')^2=(\varepsilon_j a)^2=a^2$, and that the only non-vanishing traces are Tr$(\tilde \lambda_m \lambda_n)=-\delta_{mn}$. This proves (\ref{trhkhk}).\\
For (\ref{trxkxk}), we use that
\begin{align}
 x=&\sum_{j=1}^{r} (c_j e^{-hz} \lambda_j e^{hz} +c_j e^{-hz} \tilde \lambda_j e^{hz})=\sum_{j=1}^{r} (c_j e^{-\alpha_j(h)z} \lambda_j +c^*_j e^{\alpha_j(h)z} \tilde \lambda_j )\cr
 =&\sum_{j=1}^{r} (c_j e^{-i\varepsilon_j az} \lambda_j +c^*_j e^{i\varepsilon_j az} \tilde \lambda_j ).
\end{align}
Therefore,
\begin{align}
[x,\kappa]=\sum_{j,k}\left( c_j c_k e^{-i\varepsilon_jaz} [\lambda_j,\lambda_k]+c^*_j c^*_k e^{i\varepsilon_j az} [\tilde\lambda_j,\tilde\lambda_k] + c_j c^*_k e^{-i\varepsilon_j az} [\lambda_j,\tilde\lambda_k]
+c^*_j c_k e^{i\varepsilon_j az} [\tilde\lambda_j,\lambda_k] \right) .
\end{align} 
Using (\ref{lala}), (\ref{tlatla}) and (\ref{tlala}), it can be rewritten as
\begin{align}
[x,\kappa]=&-\sum_{j<k}\left( c_j c_k (\alpha_j|\alpha_k) (e^{-i\varepsilon_jaz}-e^{-i\varepsilon_kaz})\lambda_{\alpha_j+\alpha_k}  +c^*_j c^*_k (\alpha_j|\alpha_k) (e^{i\varepsilon_jaz}-e^{i\varepsilon_kaz})\tilde \lambda_{\alpha_j+\alpha_k} \right) \cr
& -i\sum_{j=1}^{r} |c_j|^2  (e^{-i\varepsilon_jaz}-e^{i\varepsilon_jaz}) J_j \cr
=&-\sum_{j<k}\left( c_j c_k (\alpha_j|\alpha_k) (e^{-i\varepsilon_jaz}-e^{-i\varepsilon_kaz})\lambda_{\alpha_j+\alpha_k}  +c^*_j c^*_k (\alpha_j|\alpha_k) (e^{i\varepsilon_jaz}-e^{i\varepsilon_kaz})\tilde \lambda_{\alpha_j+\alpha_k} \right) \cr
& -2  \sum_{j=1}^{r} |c_j|^2 \sin(\varepsilon_jaz) J_j.
\end{align}
With our normalization for the scalar products we get
\begin{align}
 {\rm Tr}([x,\kappa][x,\kappa])=& -2 \sum_{j<k} |c_j|^2|c_{k}|^2 |(\alpha_j|\alpha_k)|^2 \left| e^{-i\varepsilon_j az}-e^{-i\varepsilon_{k}az} \right|^2-4\sum_{j=1}^{r} |c_j|^4 \|\alpha_j\|^2 \sin^2 (\varepsilon_j a z)\cr
 &-8\sum_{j<k} |c_j|^2 |c_{k}|^2 (\alpha_j|\alpha_k)  \sin (\varepsilon_j ar) \sin (\varepsilon_{k} az).
\end{align}
Next, we use
\begin{align}
 \left| e^{-i\varepsilon_j az}-e^{-i\varepsilon_{k}az} \right|^2 &=4\sin^2 \left( az \frac {\varepsilon_j-\varepsilon_{k}}2 \right)=2\sin^2 (az) (1- \varepsilon_j \varepsilon_{k}),
\end{align}
where we used that $(\varepsilon_j-\varepsilon_{k})/2=0,\pm1$. Finally, using
\begin{align}
 \sin (\varepsilon_j az) \sin (\varepsilon_{k} az)=\sin^2 (az) \varepsilon_j \varepsilon_{k},
\end{align}
we get (\ref{trxkxk}). \end{proof}

\section{Connection between Lasagna and Spaghetti}
It is interesting to compare the results obtained from Lasagna and Spaghetti parameterization. We can do this by determining the relation between the two parameterizations via de identification
\begin{align}
U\equiv e^{\chi (\sin \Theta \cos \Phi T_1+ \sin \Theta \sin \Phi T_2+ \cos \Theta T_3)}=e^{\Psi(\chi,\Theta,\Phi) T_1} e^{H(\chi,\Theta,\Phi)T_3} e^{\Gamma (\chi,\Theta,\Phi) T_1}. \label{comparison}
\end{align}
A priori one could expect this parameterization to be dependent on the representation the $T_j$ are belonging to, since of course the exponentials do. Nevertheless, in both cases the left invariant current $\mathcal{L}_\mu =U^{-1} \partial_\mu U$ for fixed coordinates
is independent on the representation but depends only on the normalizations. If we normalize the matrices as in the previous sections, after writing $U^{-1}dU|_{Exp}=U^{-1} dU|_{Eul}$ we get the differential equation relating the exponential coordinates to 
the Euler ones. These are easily obtained, but writing them is not helpful since it would be quite difficult to solve them by brute force. Instead, we can find a solution without knowing them. The shown independence on the representation suggests to write 
down (\ref{comparison}) in the lowest representation. This is achieved by choosing 
\begin{align}
T_1=\frac 12
\begin{pmatrix}
	0 & i \\ i & 0 
\end{pmatrix},
\qquad
T_2=\frac 12
\begin{pmatrix}
	0 & -1 \\ 1 & 0 
\end{pmatrix},
\qquad
T_3=\frac 12
\begin{pmatrix}
	i & 0 \\ 0 & -i 
\end{pmatrix}.
\end{align}
With this choice (\ref{comparison}) becomes
\begin{align}
\cos \frac \chi2\ \mathbb I_2&+2\sin \frac \chi2 \sin \Theta \cos \Phi\ T_1 +2 \sin \frac \chi2 \sin \Theta \sin \Phi\ T_2+2\sin \frac \chi2 \cos \Theta\ T_3\cr
=& \cos \frac H2 \cos \frac {\Psi+\Gamma}2 \mathbb I_2 +2 \cos \frac H2 \sin \frac {\Psi+\Gamma}2 T_1+2 \sin \frac H2 \sin \frac {\Psi-\Gamma}2 T_2 +2 \sin \frac H2 \cos \frac {\Psi-\Gamma}2 T_3.
\end{align}
This gives
\begin{align}
H&=2 \arcsin (\sin \frac \chi2  \sqrt {1+\sin^2 \Phi}), \\
\Psi+\Gamma &= 2\arctan (\tan \frac \chi2 \sin \Theta \cos \Phi), \\
\Psi-\Gamma &= 2\arctan (\tan \Theta \sin \Phi), 
\end{align}
and the inverse
\begin{align}
\chi&=2 \arccos (\cos \frac H2 \  \cos \frac {\Psi+\Gamma}2), \\
\Theta&= \arctan \left(\tan \frac {\Psi-\Gamma}2\ \sqrt {1+ \frac 1{\tan^2 \frac H2}\ \frac {\sin^2 \frac {\Psi+\Gamma}2}{\sin^2 \frac {\Psi-\Gamma}2}} \right), \\
\Phi&=\arctan \left(\tan \frac H2 \  \frac {\sin \frac {\Psi-\Gamma}2}{\sin \frac {\Psi+\Gamma}2}\right).
\end{align}
For example, the Lasagna solutions have the form 
\begin{align}
\Psi&=\frac t{L_\phi} -\phi, \\
H&=ar, \\
\Gamma&=m\theta,
\end{align}
so that in the exponential form they take the very complicated form
\begin{align}
\chi(t,r,\theta,\phi)&=2 \arccos \left(\cos \frac {ar}2 \  \cos \frac {t -L_\phi(\phi-\theta)}{2L_\phi}\right), \\
\Theta(t,r,\theta,\phi)&= \arctan \left(\tan \frac {t -L_\phi(\phi+\theta)}{2L_\phi}\ \sqrt {1+ \frac 1{\tan^2 \frac {ar}2}\ \frac {\sin^2 \frac {t -L_\phi(\phi-\theta)}{2L_\phi}}{\sin^2 \frac {t -L_\phi(\phi+\theta)}{2L_\phi}}} \right), \\
\Phi(t,r,\theta,\phi)&=\arctan \left(\tan \frac {ar}2 \  \frac {\sin \frac {t -L_\phi(\phi+\theta)}{2L_\phi}}{\sin \frac {t -L_\phi(\phi-\theta)}{2L_\phi}}\right).
\end{align}

\section{Some technical details about $G_2$}\label{app:G2}
There are different ways of constructing a convenient basis for the Lie algebra of $G_2$. We will refer to \cite{CaCeDeVeOrSc}. In that notation a basis is $C_J$, $J=1,\ldots, 14$. The only maximal regular subgroup is $SO(4)$ generated 
by $C_L$, $L=1,2,3,8,9,10$. The remaining matrices generate $\mathfrak p$. A convenient Cartan subspace is thus
\begin{align}
 H=\langle C_5, C_{11} \rangle_{\mathbb R}.
\end{align}
As a basis, we take $h_1=C_{11}$ and $h_2=C_5$. One can easily diagonalize the action of $ad(H)$. If we set
\begin{align}
 \lambda_1&\equiv k_1+i p_1= \frac 1{4\sqrt 2} (\sqrt {3} C_3-C_8)+ i \frac 1{2\sqrt 2} C_{12}, \\
 \lambda_2&\equiv k_2+i p_2=\frac 18 (C_1+C_2-\sqrt 3 C_9-\sqrt 3 C_{10}) +i \frac 18 (C_6+C_7+\sqrt 3 C_{13}-\sqrt 3 C_{14}),
\end{align}
then, they satisfy ${\rm Tr}(\lambda_i \tilde \lambda_j)=-\delta_{ij}$ and 
\begin{align}
[h_1,\lambda_1]&=i\frac {2}{\sqrt 3} \lambda_1, \quad\ [h_2,\lambda_1]=0, \label{questa} \\
[h_1,\lambda_2]&=-i {\sqrt 3} \lambda_2, \quad\ [h_2,\lambda_2]=i\lambda_2. \label{quella}
\end{align}
To keep contact with our conventions we have to redefine the basis for $H$ as $J_1$ and $J_2$, defined by (notice that $\tilde \lambda_j$ is simply the complex conjugate of $\lambda_j$)
\begin{align}
 J_1&=-i[\tilde \lambda_1,\lambda_1]=-\frac 1{2\sqrt 3} C_{11}, \label{D6}\\
 J_2&=-i[\tilde \lambda_2,\lambda_2]=\frac {\sqrt 3}4 C_{11} -\frac 14 C_5.
\end{align}
This gives us the geometry of roots, so that
\begin{align}
 (\alpha_1|\alpha_1)&=-{\rm Tr} (J_1 J_1)=\frac 13,\\
 (\alpha_2|\alpha_2)&=-{\rm Tr} (J_2 J_2)=1, \\
 (\alpha_1|\alpha_2)&=-{\rm Tr} (J_1 J_2)=-\frac 12. \\
\end{align}
We can represent this vectors in the canonical euclidean $\mathbb R^2$ as the vectors
\begin{align}
 \alpha_1\equiv \left(\frac {1}{\sqrt 3},0\right), \qquad \alpha_2\equiv \left(-\frac {\sqrt 3}2,\frac 12\right).
\end{align}
The corresponding Cartan matrix is
\begin{align}
 C^G=
\begin{pmatrix}
 2 & -3\\
 -1 & 2
\end{pmatrix},
\end{align}
with inverse
\begin{align}
  (C^G)^{-1}=
\begin{pmatrix}
 2 & 3\\
 1 & 2
\end{pmatrix}.
\end{align}
It is also useful to determine the basis for all eigenspaces, in a convenient way, normalized so that ${\rm Tr}(\tilde \lambda_\alpha \lambda_\alpha)=-1$ for any root. After setting
\begin{align}
 \alpha_3=\alpha_1+\alpha_2, \quad\ \alpha_4=2\alpha_1+\alpha_2, \quad\ \alpha_5=3\alpha_1+\alpha_2,\quad\  \alpha_6=3\alpha_1+2\alpha_2,
\end{align}
we can state the following proposition.
\begin{prop}
 A suitable choice of the eigenmatrices associated to the roots $\alpha_j$, $j=3,4,5,6$, is given by
\begin{align}
 \lambda_3=\sqrt 2 [\lambda_1,\lambda_2], \quad\ \lambda_4=\sqrt {\frac 32}  [\lambda_1,\lambda_3], \quad\ \lambda_5=\sqrt {2}  [\lambda_1,\lambda_4], \quad\ \lambda_6=\sqrt {2}  [\lambda_3,\lambda_4].
\end{align}
Moreover, $\tilde \lambda_j$ is the complex conjugate of $\lambda_j$.
\end{prop}
\begin{proof}
We know from the general theory that if $\lambda_a$ and $\lambda_b$ are eigenmatrices of the roots $\alpha_a$ and $\alpha_b$ respectively, and if $\alpha_a+\alpha_b$ is also root, than the eigenmatrices of $\alpha_a+\alpha_b$ have the form 
$\mu[\lambda_a,\lambda_b]$  for any given constant $\mu$. Since $\lambda_1$ and $\lambda_2$ are eigenmatrices for the fundamental roots $\alpha_1$ and $\alpha_2$, we have that the matrices $\lambda_j$ specified above are surely eigenmatrices for 
the corresponding roots $\alpha_j$, $j=3,4,5,6$. We have only to explain the choices of the constant factors. These are chosen to be real and such that Tr$(\lambda_j \tilde \lambda_j)=-1$. To prove it,
first notice that necessarily $[\tilde \lambda_i, \tilde \lambda_j]$ are eigenmatrices for $-\alpha_i-\alpha_j$, so we can identify 
$\tilde \lambda_3=\sqrt 2 [\tilde \lambda_1,\tilde \lambda_2]$ and so on. The last part of the proposition then follows from the fact that it is true for $\lambda_1$ and $\lambda_2$ and that all the coefficient we chosen for defining the remaining $\lambda_j$ are
real. Then, using ad-invariance of the trace product, that is Tr$([A,B]C)=$Tr$(A,[B,C])$, and the Jacobi identity $[A,[B,C]]=[[A,B],C]+[B,[A,C]]$, we have
\begin{align}
 {\rm Tr} (\lambda_3\tilde \lambda_3)=&{\rm Tr} (2[\lambda_1,\lambda_2][\tilde \lambda_1,\tilde \lambda_2])=2{\rm Tr} (\lambda_1[\lambda_2,[\tilde \lambda_1,\tilde \lambda_2]])
 =2{\rm Tr} (\lambda_1([[\lambda_2,\tilde \lambda_1],\tilde \lambda_2]+[\tilde \lambda_1,[\lambda_2, \tilde \lambda_2]])).
\end{align}
Since $\alpha_1$ and $\alpha_2$ are simple roots, we have $[\lambda_2,\tilde \lambda_1]=0$. From (\ref{D6}) we see that $[\lambda_2, \tilde \lambda_2]=iJ_2$ and since 
$[J_2,\tilde \lambda_1]=-i\alpha_1(J_2) \tilde \lambda_1=-i(\alpha_1|\alpha_2)\tilde \lambda_1=\frac i2 \tilde \lambda_1$, we get
\begin{align}
 {\rm Tr} (\lambda_3\tilde \lambda_3)={\rm Tr} (\lambda_1\tilde \lambda_1)=-1.
\end{align}
Next, consider 
\begin{align}
 {\rm Tr} (\lambda_4\tilde \lambda_4)=&{\rm Tr} (\frac 32[\lambda_1,\lambda_3][\tilde \lambda_1,\tilde \lambda_3])=-\frac 32 {\rm Tr} (\lambda_3 [\lambda_1,[\tilde \lambda_1,\tilde \lambda_3]])=
 -\frac 32 {\rm Tr} (\lambda_3 [[\lambda_1,\tilde \lambda_1],\tilde \lambda_3])-\frac 32 {\rm Tr} (\lambda_3 [\tilde \lambda_1,[\lambda_1,\tilde \lambda_3]])\cr
 =&-\frac 32 {\rm Tr} (\lambda_3 [iJ_1,\tilde \lambda_3])-\frac 32 {\rm Tr} ([\tilde \lambda_1,\lambda_3] [\lambda_1,\tilde \lambda_3])=-\frac 32 i (-i\alpha_3(J_1)){\rm Tr} (\lambda_3 \tilde \lambda_3)
 -\frac 32 {\rm Tr} ([\tilde \lambda_1,\lambda_3] [\lambda_1,\tilde \lambda_3])\cr 
 =& \frac 32 (\alpha_3|\alpha_1) -\frac 32 {\rm Tr} ([\tilde \lambda_1,\lambda_3] [\lambda_1,\tilde \lambda_3]).
\end{align}
Since $\alpha_3=\alpha_1+\alpha_2$, we have
\begin{align}
 (\alpha_3|\alpha_1)=(\alpha_1|\alpha_1)+(\alpha_2|\alpha_1)=\frac 13 -\frac 12.
\end{align}
On the other hand
\begin{align}
 [\lambda_1,\tilde \lambda_3]=\sqrt 2 [\lambda_1,[\tilde \lambda_1,\tilde \lambda_2]]=\sqrt 2 [[\lambda_1,\tilde \lambda_1],\tilde \lambda_2]=i\sqrt 2 [J_1,\tilde \lambda_1]=i\sqrt 2 (-i\alpha_2(J_1)) \tilde \lambda_2,
\end{align}
where we used again that $[\lambda_1,\tilde \lambda_2]=0$, and, therefore,
\begin{align}
 {\rm Tr} ([\tilde \lambda_1,\lambda_3] [\lambda_1,\tilde \lambda_3])=2  (\alpha_2(J_1))^2 {\rm Tr}(\tilde \lambda_2 \lambda_2)=-2 (-1/2)^2,
\end{align}
and putting all together we get ${\rm Tr} (\lambda_4\tilde \lambda_4)=-1$. \\
The remaining two cases are proved exactly in the same way.
\end{proof}

\subsection{The fundamental irreps of $G_2$}
$G_2$ has 12 non null roots forming two concentric hexagons in $H^*$, plus two vanishing roots, like in Fig. \ref{radici}.
\begin{figure}[!htbp]
\begin{center}
\begin{tikzpicture}[rounded corners]
\draw [thick] (-4,0)--(4,0);
\draw [thick] (0,-4)--(0,4);
\draw [thick,red] (0,2*1.732) -- (3,1.732) -- (3,-1.732) -- (0,-2*1.732) -- (-3,-1.732)  -- (-3,1.732) -- cycle;
\draw [thick,blue] (2,0) -- (1,1.732) -- (-1,1.732) -- (-2,0) -- (-1,-1.732)  -- (1,-1.732) -- cycle;
\draw [ultra thick,-latex,red] (0,0) -- (2,0); 
\draw [ultra thick,-latex] (0,0) -- (-2,0);
\draw [ultra thick,-latex,blue] (0,0) -- (0,2*1.732);
\draw [ultra thick,-latex] (0,0) -- (0,-2*1.732);  
\draw [ultra thick,-latex,blue] (0,0) -- (1,1.732);
\draw [ultra thick,-latex] (0,0) -- (-1,1.732);
\draw [ultra thick,-latex] (0,0) -- (1,-1.732);
\draw [ultra thick,-latex] (0,0) -- (-1,-1.732);
\draw [ultra thick,-latex] (0,0) -- (3,1.732);
\draw [ultra thick,-latex,red] (0,0) -- (-3,1.732);
\draw [ultra thick,-latex] (0,0) -- (3,-1.732);
\draw [ultra thick,-latex] (0,0) -- (-3,-1.732);
\draw [cyan,fill] (-0.2,0) circle (2pt);
\draw [cyan,fill] (0.2,0) circle (2pt);
\node at (2.2,0.2) {$\pmb{\alpha_1}$};
\node at (-2.2,-0.2) {$\pmb{-\alpha_1}$};
\node at (-3.2,1.932) {$\pmb{\alpha_2}$};
\node at (3.2,-1.932) {$\pmb{-\alpha_2}$};
\node at (-1.2,1.932) {$\pmb{\alpha_3}$};
\node at (1.2,-1.932) {$\pmb{-\alpha_3}$};
\node at (1.2,1.932) {$\pmb{\alpha_4}$};
\node at (-1.2,-1.932) {$\pmb{-\alpha_4}$};
\node at (3.2,1.932) {$\pmb{\alpha_5}$};
\node at (-3.2,-1.932) {$\pmb{-\alpha_5}$};
\node at (-0.2,2*1.832) {$\pmb{\alpha_6}$};
\node at (0.35,-2*1.832) {$\pmb{-\alpha_6}$};
\end{tikzpicture}
\caption{The twelve roots of $G_2$, two of which are zero. $\alpha_1$ and $\alpha_2$ are the simple roots. $\alpha_6$ and $\alpha_4$ are the fundamental weights.}\label{radici}
\end{center}
\end{figure}
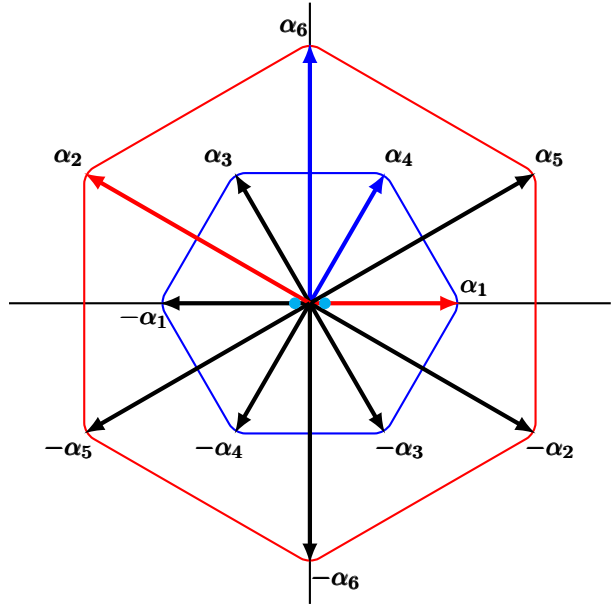

$\alpha_1, \ldots, \alpha_6$ are the positive roots. To each of them, $\alpha_j$, it corresponds an eigenmatrix $\lambda_j$ and to each negative root $-\alpha_j$ it corresponds $\tilde \lambda_j$. To the vanishing roots one associates the
matrices $h_a=i[\lambda_a,\tilde \lambda_a]$, $a=1,2$. The $14$ matrices $h_a$, $\lambda_j+\tilde \lambda_j$, $i(\lambda_j-\tilde \lambda_j)$, $a=1,2$, $j=1,\ldots,6$ form a basis for the adjoint representation $\pmb {14}$, with maximal
weight $\mu_1=\alpha_6$.\\ 
The second fundamental irreducible representation has maximal weight $\mu_2=\alpha_4$. The weights of such representation are $0, \pm\alpha_b$, $b=1,3,4$, each one with multiplicity 1, so that it is a seven dimensional representation, $\pmb 7$.
It is depicted in Fig. \ref{irrep7}.

\begin{figure}[!htbp]
\begin{center}
\begin{tikzpicture}[rounded corners]
\draw [thick] (-3,0)--(3,0);
\draw [thick] (0,-3)--(0,3);
\draw [thick,blue] (2,0) -- (1,1.732) -- (-1,1.732) -- (-2,0) -- (-1,-1.732)  -- (1,-1.732) -- cycle;
\draw [ultra thick,-latex] (0,0) -- (2,0); 
\draw [ultra thick,-latex] (0,0) -- (-2,0);
\draw [ultra thick,-latex,blue] (0,0) -- (1,1.732);
\draw [ultra thick,-latex] (0,0) -- (-1,1.732);
\draw [ultra thick,-latex] (0,0) -- (1,-1.732);
\draw [ultra thick,-latex] (0,0) -- (-1,-1.732);
\draw [cyan,fill] (0,0) circle (2pt);
\node at (2.2,0.2) {$\pmb{v_3}$};
\node at (-2.2,-0.2) {$\pmb{v_5}$};
\node at (-1.2,1.932) {$\pmb{v_2}$};
\node at (1.2,-1.932) {$\pmb{v_6}$};
\node at (1.2,1.932) {$\pmb{v_1}$};
\node at (-1.2,-1.932) {$\pmb{v_7}$};
\node at (-0.3,0.2) {$\pmb{v_4}$};
\end{tikzpicture}
\caption{The irrep $\pmb 7$ of $G_2$. To each weight it corresponds a vector of the basis defining a 7 dimensional vector space.}\label{irrep7}
\end{center}
\end{figure}
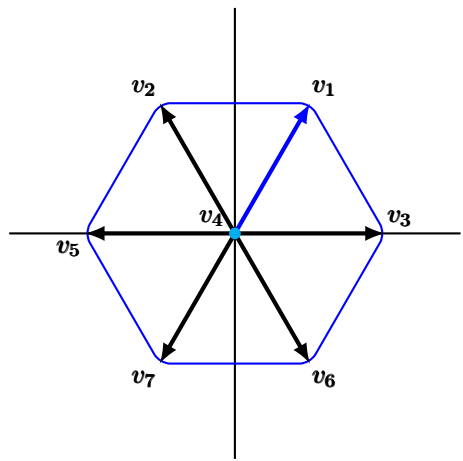

The matrices in this representation are $\rho_{\pmb 7}(h_a)$, $\rho_{\pmb 7}(\lambda_j)+\rho_{\pmb 7}(\tilde \lambda_j)$, $i(\rho_{\pmb 7}(\lambda_j)-\rho_{\pmb 7}(\tilde \lambda_j))$, $a=1,2$, $j=1,\ldots,6$, and can be understood
by noticing that $\rho_{\pmb 7}(\tilde \lambda_j)$ shifts the weights by $\alpha_j$, giving zero if and only if the result is not a weight, and similar for $\tilde \lambda_j$. 

\subsection{Irreducibility of $\chi_{28}$ in representation $\pmb 7$} 
Let us consider $\rho_{\pmb 7}(T_-)=3\rho_{\pmb 7}(\tilde \lambda_1)+\sqrt 5 \rho(\tilde \lambda_2)$ acting on the maximal vector $v_1$ of $\pmb 7$. Since $\alpha_4-\alpha_2$ is not a weight of $\pmb 7$, while $\alpha_4-\alpha_1=\alpha_3$ is,
we have that $\rho_{\pmb 7}(T_-)v_1=3 \rho(\tilde \lambda_1)v_1=k_2 v_2$, with\footnote{We could compute it explicitly, but it is not necessary for our purposes.} $k_2\neq 0$. In the same way we have the chain of relations:
\begin{align*}
 \rho_{\pmb 7}(T_-)v_2&=\sqrt 5\rho_{\pmb 7}(\tilde \lambda_2)v_2=k_3 v_3, \\
 \rho_{\pmb 7}(T_-)v_3&=3 \rho_{\pmb 7}(\tilde \lambda_1)v_3=k_4 v_4, \\
 \rho_{\pmb 7}(T_-)v_4&=3 \rho_{\pmb 7}(\tilde \lambda_1)v_4=k_5 v_5, \\
 \rho_{\pmb 7}(T_-)v_5&=\sqrt 5\rho_{\pmb 7}(\tilde \lambda_2)v_5=k_6 v_6, \\
 \rho_{\pmb 7}(T_-)v_6&=3 \rho_{\pmb 7}(\tilde \lambda_1)v_6=k_7 v_7, \\
 \rho_{\pmb 7}(T_-)v_7&=0,
\end{align*}
with all $k_j$ different from zero. Therefore, $\chi_{28}$ is a representation of spin $3$ and $\pmb 7$ is irreducible under $\chi_{28}$. 

\subsection{Explicit matrix realizations} 
Here we present the explicit matrix representation of the three dimensional subalgebras in the irrep $\pmb 7$ of $G_2$.
These are
\begin{align}
 T^{(1)}_1&= 
\begin{pmatrix}
 0 & 0 & 0 & 0 & 0 & 0 & 0 \\ 
 0 & 0 & -\frac 12 & 0 & 0 & 0 & 0 \\ 
 0 & \frac 12 & 0 & 0 & 0 & 0 & 0 \\ 
 0 & 0 & 0 & 0 & 0 & 0 & 0 \\ 
 0 & 0 & 0 & 0 & 0 & 0 & 0 \\ 
 0 & 0 & 0 & 0 & 0 & 0 & -\frac 12 \\ 
 0 & 0 & 0 & 0 & 0 & \frac 12 & 0
\end{pmatrix}
 ,\\
 T^{(1)}_2&=
\begin{pmatrix}
 0 & 0 & 0 & 0 & 0 & 0 & 0 \\ 
 0 & 0 & 0 & 0 & 0 & 0 & -\frac 12 \\ 
 0 & 0 & 0 & 0 & 0 & -\frac 12 & 0 \\ 
 0 & 0 & 0 & 0 & 0 & 0 & 0 \\ 
 0 & 0 & 0 & 0 & 0 & 0 & 0 \\ 
 0 & 0 & \frac 12 & 0 & 0 & 0 & 0 \\ 
 0 & \frac 12 & 0 & 0 & 0 & 0 & 0
\end{pmatrix}
 ,\\
 T^{(1)}_3&=
 \begin{pmatrix}
 0 & 0 & 0 & 0 & 0 & 0 & 0 \\ 
 0 & 0 & 0 & 0 & 0 & \frac 12 & 0 \\ 
 0 & 0 & 0 & 0 & 0 & 0 & -\frac 12 \\ 
 0 & 0 & 0 & 0 & 0 & 0 & 0 \\ 
 0 & 0 & 0 & 0 & 0 & 0 & 0 \\ 
 0 & -\frac 12 & 0 & 0 & 0 & 0 & 0 \\ 
 0 & 0 & \frac 12 & 0 & 0 & 0 & 0
\end{pmatrix},
\end{align}

\begin{align}
 T^{(3)}_1&=\frac 1{2\sqrt 2} 
\begin{pmatrix}
 0 & 1 & 1 & 0 & 0 & 0 & 0 \\ 
 -1 & 0 & 0 & 0 & 0 & 0 & 0 \\ 
 -1 & 0 & 0 & 0 & 0 & 0 & 0 \\ 
 0 & 0 & 0 & 0 & 0 & -1 & 1 \\ 
 0 & 0 & 0 & 0 & 0 & 2 & 2 \\ 
 0 & 0 & 0 & 1 & -2 & 0 & 0 \\ 
 0 & 0 & 0 & -1 & -2 & 0 & 0
\end{pmatrix}
 ,\\
 T^{(3)}_2&=\frac 1{2\sqrt 2} 
\begin{pmatrix}
 0 & 0 & 0 & 0 & 0 & -1 & 1 \\ 
 0 & 0 & 0 & 1 & -2 & 0 & 0 \\ 
 0 & 0 & 0 & 1 & 2 & 0 & 0 \\ 
 0 & -1 & -1 & 0 & 0 & 0 & 0 \\ 
 0 & 2 & -2 & 0 & 0 & 0 & 0 \\ 
 1 & 0 & 0 & 0 & 0 & 0 & 0 \\ 
 -1 & 0 & 0 & 0 & 0 & 0 & 0
\end{pmatrix}
 ,\\
 T^{(3)}_3&=\frac 1{4} 
 \begin{pmatrix}
 0 & 0 & 0 & 2 & 0 & 0 & 0 \\ 
 0 & 0 & 0 & 0 & 0 & 3 & 1 \\ 
 0 & 0 & 0 & 0 & 0 & -1 & -3 \\ 
 -2 & 0 & 0 & 0 & 0 & 0 & 0 \\ 
 0 & 0 & 0 & 0 & 0 & 0 & 0 \\ 
 0 & -3 & 1 & 0 & 0 & 0 & 0 \\ 
 0 & -1 & 3 & 0 & 0 & 0 & 0
\end{pmatrix},
\end{align}

\begin{align}
 T^{(4)}_1&= 
\begin{pmatrix}
 0 & 1 & 0 & 0 & 0 & 0 & 0 \\ 
 -1 & 0 & 0 & 0 & 0 & 0 & 0 \\ 
 0 & 0 & 0 & 0 & 0 & 0 & 0 \\ 
 0 & 0 & 0 & 0 & 0 & 0 & -1 \\ 
 0 & 0 & 0 & 0 & 0 & 0 & 0 \\ 
 0 & 0 & 0 & 0 & 0 & 0 & 0 \\ 
 0 & 0 & 0 & 1 & 0 & 0 & 0
\end{pmatrix}
 ,\\
 T^{(4)}_2&=
\begin{pmatrix}
 0 & 0 & 0 & 0 & 0 & -1 & 0 \\ 
 0 & 0 & 0 & 0 & 0 & 0 & 0 \\ 
 0 & 0 & 0 & -1 & 0 & 0 & 0 \\ 
 0 & 0 & 1 & 0 & 0 & 0 & 0 \\ 
 0 & 0 & 0 & 0 & 0 & 0 & 0 \\ 
 1 & 0 & 0 & 0 & 0 & 0 & 0 \\ 
 0 & 0 & 0 & 0 & 0 & 0 & 0
\end{pmatrix}
 ,\\
 T^{(4)}_3&=
 \begin{pmatrix}
 0 & 0 & 0 & 0 & 0 & 0 & 0 \\ 
 0 & 0 & 0 & 0 & 0 & 1 & 0 \\ 
 0 & 0 & 0 & 0 & 0 & 0 & -1 \\ 
 0 & 0 & 0 & 0 & 0 & 0 & 0 \\ 
 0 & 0 & 0 & 0 & 0 & 0 & 0 \\ 
 0 & -1 & 0 & 0 & 0 & 0 & 0 \\ 
 0 & 0 & 1 & 0 & 0 & 0 & 0
\end{pmatrix},
\end{align}

\begin{align}
 T^{(28)}_1&= \frac 12
\begin{pmatrix}
 0 & \sqrt 5 & \sqrt 5 & 0 & 0 & 0 & 0 \\ 
 -\sqrt 5 & 0 & \sqrt 6 & 0 & 0 & 0 & 0 \\ 
 -\sqrt 5 & -\sqrt 6 & 0 & 0 & 0 & 0 & 0 \\ 
 0 & 0 & 0 & 0 & -2\sqrt 6 & \sqrt 5 & -\sqrt 5 \\ 
 0 & 0 & 0 & 2\sqrt 6 & 0 & 0 & 0 \\ 
 0 & 0 & 0 & -\sqrt 5 & 0 & 0 & -\sqrt 6 \\ 
 0 & 0 & 0 & \sqrt 5 & 0 & \sqrt 6 & 0
\end{pmatrix}
 ,\\
 T^{(28)}_2&=\frac 12
\begin{pmatrix}
 0 & 0 & 0 & 0 & -2\sqrt 6 & -\sqrt 5 & \sqrt 5 \\ 
 0 & 0 & 0 & -\sqrt 5 & 0 & \sqrt 6 & 0 \\ 
 0 & 0 & 0 & -\sqrt 5 & 0 & 0 & \sqrt 6 \\ 
 0 & \sqrt 5 & \sqrt 5 & 0 & 0 & 0 & 0 \\ 
 2\sqrt 6 & 0 & 0 & 0 & 0 & 0 & 0 \\ 
 \sqrt 5 & -\sqrt 6 & 0 & 0 & 0 & 0 & 0 \\ 
 -\sqrt 5 & 0 & -\sqrt 6 & 0 & 0 & 0 & 0
\end{pmatrix}
 ,\\
 T^{(28)}_3&=\frac 12
 \begin{pmatrix}
 0 & 0 & 0 & 2 & 0 & 0 & 0 \\ 
 0 & 0 & 0 & 0 & 0 & 5 & 1 \\ 
 0 & 0 & 0 & 0 & 0 & -1 & -5 \\ 
 -2 & 0 & 0 & 0 & 0 & 0 & 0 \\ 
 0 & 0 & 0 & 0 & 0 & 0 & 0 \\ 
 0 & -5 & 1 & 0 & 0 & 0 & 0 \\ 
 0 & -1 & 5 & 0 & 0 & 0 & 0
\end{pmatrix}.
\end{align}

\end{appendix}



\newpage

\begin{thebibliography}{99}
\bibitem{pasta1} P. de Forcrand, Proc. Sci., LAT2009 (2009) 010
[arXiv:1005.0539].

\bibitem{pasta2} N. Brambilla et al., Eur. Phys. J. C 74, 2981 (2014).

\bibitem{pasta3} D. G. Ravenhall, C. J. Pethick, and J. R. Wilson, Phys.
Rev. Lett. 50, 2066 (1983).

\bibitem{pasta4} M. Hashimoto, H. Seki, and M. Yamada, Prog. Theor. Phys.
71, 320 (1984).

\bibitem{pasta5} C. J. Horowitz, D. K. Berry, C.M. Briggs, M. E. Caplan, A.
Cumming, and A. S. Schneider, Phys. Rev. Lett. 114, 031102 (2015).

\bibitem{pasta6} D. K. Berry, M. E. Caplan, C. J. Horowitz, G. Huber, and A.
S. Schneider, Phys. Rev. C 94, 055801 (2016).

\bibitem{skyrme} T. Skyrme, \textit{Proc. R. Soc. London} \textbf{A 260},
127 (1961); \textit{Proc. R. Soc. London} \textbf{A 262}, 237 (1961); 
\textit{Nucl. Phys.}\textbf{\ 31}, 556 (1962).

\bibitem{Witten} C. G. Callan Jr. and E. Witten, \textit{Nucl. Phys. B} 
\textbf{239} (1984) 161-176.

\bibitem{GSkyrme} B.M.A.G. Piette, D.H. Tchrakian, Phys. Rev. D 62, 025020
(2000).

\bibitem{witten0} E. Witten, Nucl. Phys. \textbf{B 223} (1983), 422; Nucl.
Phys. \textbf{B 223}, 433 (1983).

\bibitem{bala0} A. P. Balachandran, V. P. Nair, N. Panchapakesan, S. G.
Rajeev, 
Phys. Rev. \textbf{D 28} (1983), 2830.

\bibitem{Bala1} A. P. Balachandran, A. Barducci, F. Lizzi, V.G.J. Rodgers,\
A. Stern, 
Phys. Rev. Lett. \textbf{52 }(1984), 887; A.P. Balachandran, F. Lizzi,
V.G.J. Rodgers,\ A. Stern, 
Nucl. Phys. \textbf{B 256}, 525-556 (1985).

\bibitem{ANW} G. S. Adkins, C. R. Nappi, E. Witten, \textit{Nucl. Phys}. 
\textbf{B 228} (1983), 552-566.

\bibitem{manton} N.~Manton and P.~Sutcliffe, \textit{Topological Solitons},
(Cambridge University Press, Cambridge, 2007).

\bibitem{BaMa} A. Balachandran, G. Marmo, B. Skagerstam, A. Stern, \textit{%
Classical Topology and Quantum States}, World Scientific (1991).

\bibitem{manton0.5} N.S. Manton, P.J. Ruback, Phys.Lett.B 181 (1986) 137-140.

\bibitem{manton0} A. S. Goldhaber, N.S. Manton, Phys.Lett.B 198 (1987)
231-234.

\bibitem{manton0.25} N. S. Manton, Comm. Math. Phys. 111, 469--478 (1987).

\bibitem{manton0.3} M. F. Atiyah, N. S. Manton, Comm. Math. Phys. 153,
391--422 (1993).

\bibitem{manton1} N.S. Manton, B.J. Schroers, M.A. Singer, Commun. Math.
Phys. 245 (2004) 123-147.

\bibitem{spreight} J. M. Speight, Comm. Math. Phys. 332, 355--377 (2014).

\bibitem{56b} S. Chen, Y. Li, Y. Yang, \textit{Phys. Rev.} \textbf{D 89}
(2014), 025007.

\bibitem{56} F. Canfora, \textit{Phys. Rev.} \textbf{D 88}, (2013), 065028;
E. Ayon-Beato, F. Canfora, J. Zanelli, \textit{Phys. Lett.} \textbf{B 752, }%
(2016) 201-205; L. Aviles, F. Canfora, N. Dimakis, D. Hidalgo, \textit{Phys.
Rev. }\textbf{D 96} (2017), 125005; F.~Canfora, M.~Lagos, S.~H.~Oh, J.~Oliva
and A.~Vera, Phys.\ Rev.\ D \textbf{98}, no. 8, 085003 (2018); F. Canfora,
N. Dimakis, A. Paliathanasis, \textit{Eur.Phys.J.} \textbf{C79} (2019) no.2,
139.

\bibitem{56b1} E.~Ayon-Beato, F.~Canfora, M.~Lagos, J.~Oliva, A.~Vera, Eur.\
Phys.\ J.\ C \textbf{80}, no. 5, 384 (2020).

\bibitem{56c} P. D. Alvarez, F. Canfora, N. Dimakis and A. Paliathanasis, 
\textit{Phys. Lett}. \textbf{B 773}, (2017) 401-407.

\bibitem{crystal1} F.~Canfora, \textit{Eur.\ Phys.\ J}.\ \textbf{C 78}, no.
11, 929 (2018).

\bibitem{crystal2} F. Canfora, S.-H. Oh, A. Vera, \textit{Eur.Phys. J}. 
\textbf{C 79} (2019) no.6, 485.

\bibitem{crystal3} F.~Canfora, M.~Lagos and A.~Vera, ``Crystals of
superconducting Baryonic tubes in the low energy limit of QCD at finite
density,'' \textit{Eur.\ Phys.\ J.}\ \textbf{C} \textbf{80}, no. 8, 697
(2020).

\bibitem{crystal3.1} M. Barsanti, S. Bolognesi, F. Canfora, G. Tallarita, 
\textit{Eur.Phys.J. }\textbf{C 80} (2020) 12, 1201.

\bibitem{crystal4} F.~Canfora, S.~Carignano, M.~Lagos, M.~Mannarelli and
A.~Vera, ``Pion crystals hosting topologically stable baryons,'' Phys.\
Rev.\ D \textbf{103}, no. 7, 076003 (2021).

\bibitem{firstube} F. Canfora, A. Giacomini, M. Lagos, S. H. Oh, A. Vera, 
\textit{Eur.Phys. J.} \textbf{C 81} (2021) 1, 55.

\bibitem{firsttube2} F.~Canfora, A.~Cisterna, D.~Hidalgo and J.~Oliva,
``Exact $pp$-waves, (A)dS waves, and Kundt spaces in the Abelian-Higgs
model,'' Phys.\ Rev.\ D \textbf{103}, no. 8, 085007 (2021).

\bibitem{gaugsksu(n)} P.~D.~Alvarez, S.~L.~Cacciatori, F.~Canfora and
B.~L.~Cerchiai, ``Analytic SU(N) Skyrmions at finite Baryon density,''
Phys.\ Rev.\ D \textbf{101}, no. 12, 125011 (2020).

\bibitem{euler1} S. Bertini, S. L. Cacciatori, B. L. Cerchiai, J. Math.
Phys. (N.Y.) 47, 043510 (2006).

\bibitem{euler2} S. L. Cacciatori, F. Dalla Piazza, and A. Scotti, Trans.
Am. Math. Soc. 369, 4709 (2017).

\bibitem{euler3} T. E. Tilma and G. Sudarshan, J. Geom. Phys. 52, 263 (2004).

\bibitem{Cacciatori:2021neu} S.~L.~Cacciatori, F.~Canfora, M.~Lagos,
F.~Muscolino and A.~Vera, ``Analytic multi-Baryonic solutions in the
SU(N)-Skyrme model at finite density and a novel transition,''
[arXiv:2105.10789 [hep-th]].


\bibitem{Dy-57} E. B. Dynkin, ``Semisimple subalgebras of semisimple Lie
algebras,'' Am. Math. Soc. Transl., Ser. 2 6, 111, 245 (1957).

\bibitem{FFP1} T. Wiegelmann and T. Sakurai, Living Rev. Solar Phys. 9, 5
(2012).

\bibitem{FFP2} T. Wiegelmann, J. Geophys. Res. 113, A03S02 (2008).

\bibitem{FFP3} R. D. Blandford and R. L. Znajek, Mon. Not. R. Astron. Soc.
179, 433 (1977).

\bibitem{FFP4} B. Carter, in General Relativity: An Einstein Centenary
Survey (1979), pp. 294--369.

\bibitem{FFP5} T. D. Brennan, S. E. Gralla, and T. Jacobson, Classical
Quantum Gravity 30, 195012 (2013).

\bibitem{Magnus} W. Magnus, ``Infinite determinants associated with Hill's
equation,'' Pacific J. Math. 5(S2): 941-951 (1955).

\bibitem{Whittaker-Hill} Urwin,~K., and Arscott,~F. (1970). ``Theory of the
Whittaker Hill Equation,`` Proceedings of the Royal Society of Edinburgh.
Section A. Mathematical and Physical Sciences, 69(1), 28-44.
doi:10.1017/S0080454100008530

\bibitem{AlCaCaCe} P.~D.~Alvarez, S.~L.~Cacciatori, F.~Canfora and
B.~L.~Cerchiai, ``Analytic SU(N) Skyrmions at finite Baryon density,'' Phys.
Rev. D \textbf{101} (2020) no.12, 125011

\bibitem{CaDaPiSc} S.~L.~Cacciatori, F.~Dalla Piazza and A.~Scotti,
``Compact Lie groups: Euler constructions and generalized Dyson
conjecture,'' Trans. Am. Math. Soc. \textbf{369} (2017) no.7, 4709-4724

\bibitem{He} S. Helgason, Differential geometry, Lie groups, and symmetric
spaces, Pure and Applied Mathematics, vol. 80, Academic Press, Inc.
[Harcourt Brace Jovanovich, Publishers], New York-London, 1978.

\bibitem{CaWi} Jr.~C.~G.~Callan and E.~Witten, \textquotedblleft Monopole
Catalysis of Skyrmion Decay,\textquotedblright\ Nucl. Phys. B \textbf{239}
(1984), 90088-9.

\bibitem{Schafke} F. W. Sch\"afke, ``\"Uber die Stabilit\"atskarte der
Mathieuschen Differentialglaichung,'' Math. Nachr., \textbf{4} (1951),
176--183.

\bibitem{Arscott67} Arscott,~FM. ``The Whittaker-Hill Equation and the Wave
Equation in Paraboloidal Co-ordinates,`` Proceedings of the Royal Society of
Edinburgh Section A Mathematical and Physical Sciences. 1967;67(4):265-276.
doi:10.1017/S008045410000813X

\bibitem{ArscottPeriodicEq} Reiszig,~R., F.~M.~Arscott, ``Periodic
Differential Equations,`` (International Series of Monographs in Pure and
Applied Mathematics, Volume 66) X + 283 S. m. Abb.
Oxford/London/Edinburgh/New York/Paris/Frankfurt 1964. Pergamon Press. Preis
geb. 60 s.net. Z. angew. Math. Mech., 45: 453-454.
https://doi.org/10.1002/zamm.19650450632

\bibitem{FSpaghettiMax} Canfora,~F., Oh,~S.H., Vera,~A. ``Analytic crystals
of solitons in the four dimensional gauged non-linear sigma model,`` Eur.
Phys. J. C 79, 485 (2019). https://doi.org/10.1140/epjc/s10052-019-6994-y

\bibitem{Lubcke} L\"ubcke, Eva. ``The direct and the inverse problem of
finite type Fermi curves of two-dimensional double-periodic Schr\"{o}dinger
operators,'' Inauguraldissertation zur Erlangung des akademischen Grades
eines Doktors der Naturwissenschaften der Universit\"{a}t Mannheim (2017),
https://madoc.bib.uni-mannheim.de/43675

\bibitem{Sl} R.~Slansky, ``Group Theory for Unified Model Building,'' Phys.
Rept. \textbf{79} (1981), 1-128

\bibitem{CaCeDeVeOrSc} S.~L.~Cacciatori, B.~L.~Cerchiai, A.~Della Vedova,
G.~Ortenzi and A.~Scotti, \textquotedblleft Euler angles for
G(2),\textquotedblright\ J. Math. Phys. \textbf{46} (2005), 083512

\bibitem{wittenstrings} E. Witten, Nucl. Phys. B 249, 557 (1985).

\bibitem{Kaup} D.~J.~Kaup, Phys. Rev. 172, 1331 (1968)

\bibitem{Liebling} S.~L.~Liebling, C.~Palenzuela, Living Rev. Relativ. 15,6 (2012)

\bibitem{Derrick1964} G.~H.~Derrick, J. Math. Phys. 1252 (1964)

\end{thebibliography}
\end{document}